\documentclass[acmsmall,nonacm]{acmart}

\setcopyright{rightsretained}
\acmJournal{PACMMOD}
\acmYear{2024} \acmVolume{2} \acmNumber{2 (PODS)} 
\acmMonth{5}\acmDOI{10.1145/3651605}

\received{December 2023}
\received[revised]{February 2024}
\received[accepted]{March 2024}

\usepackage{amsmath}
\usepackage{graphicx} %
\usepackage{url}

\usepackage{bm} %

\usepackage{upgreek} %
\usepackage{soul} %
\usepackage{bbding} %
\usepackage{pifont} %

\usepackage{listings} %

\usepackage{tabularx} %
\usepackage{booktabs} %
\usepackage{multirow}
\usepackage{hhline} %
\usepackage{diagbox}
\usepackage{pgfplotstable} %

\usepackage{colortbl}
\definecolor{dg}{cmyk}{0.60,0,0.88,0.27}	

\usepackage{booktabs}  %

\usepackage{easybmat} %

\usepackage{rotating}		%

\usepackage{tikz} %
\usetikzlibrary{matrix}
\usetikzlibrary{calc}
\usetikzlibrary{math}
\usetikzlibrary{positioning,chains,fit,shapes,calc}
\definecolor{myblue}{RGB}{80,80,160}
\definecolor{mygreen}{RGB}{80,160,80}

\usetikzlibrary{arrows,shapes,trees,backgrounds,automata}
\usetikzlibrary{decorations.markings}
\usepackage{tikzsymbols}

\usepackage{subcaption} %
\usepackage[ruled,noend,linesnumbered]{algorithm2e} %
\usepackage{etoolbox} %

\DontPrintSemicolon     %
\SetNlSty{}{}{}                %
\SetAlgoInsideSkip{smallskip}   %
\SetInd{1.5mm}{1.5mm}             %

\SetAlFnt{\small}			%
\SetAlCapFnt{\small}		%
\SetAlCapNameFnt{\small}

	\makeatletter
	\patchcmd{\algocf@makecaption@ruled}{\hsize}{\columnwidth}{}{} %
	\patchcmd{\@algocf@start}{-1.5em}{0em}{}{} %
	\makeatother

	\SetAlCapHSkip{0pt} %
\SetCommentSty{mycommfont}

\newtheorem{questionW}{Question}
\newtheorem{resultW}{Result}
\newtheorem{hypothesis}{Hypothesis}		

\newtheorem{example}{Example}					%

\newtheorem*{proofintuition*}{Proof Intuition}
\usepackage[inline]{enumitem}
\usepackage{thmtools}
\usepackage{thm-restate} 

\newlist{thmlist}{enumerate*}{1}
\setlist[thmlist]{label=(\alph{thmlisti}),
                  ref=\thethm.(\alph{thmlisti}),
                  noitemsep}

\usepackage{hyperref}
\usepackage[capitalize,nameinlink]{cleveref}

\Crefname{thm}{Theorem}{Theorems}
\Crefname{lem}{Lemma}{Lemmas}
\Crefname{listthm}{Theorem}{Theorems}
\Crefname{listlem}{Lemma}{Lemmas}
\Crefname{thmlisti}{Theorem}{Theorems}

\addtotheorempostheadhook[thm]{\crefalias{thmlisti}{listthm}}
\addtotheorempostheadhook[lem]{\crefalias{thmlisti}{listlem}}

{\end{itshape}
}
\newcommand{\hide}[1]{} 		%

\crefformat{equation}{(#2#1#3)}						%
\crefrangeformat{equation}{(#3#1#4) to~(#5#2#6)}
\crefmultiformat{equation}{(#2#1#3)}%
{ and~(#2#1#3)}{, (#2#1#3)}{ and~(#2#1#3)}

\AtEndPreamble{%
    \hypersetup{colorlinks,
      linkcolor=purple,
      citecolor=blue,
      urlcolor=ACMDarkBlue,
      filecolor=ACMDarkBlue}}

\usepackage{tcolorbox}		%
\tcbuselibrary{breakable,skins}		%

\definecolor{mygreen}{rgb}{0,0.6,0}
\definecolor{mygray}{rgb}{0.5,0.5,0.5}
\definecolor{mymauve}{rgb}{0.58,0,0.82}
\definecolor{myred}{rgb}{1,0,0}
\definecolor{mylightblue}{rgb}{0.953, 0.957, 0.980} 	%
\definecolor{mylightblueborder}{rgb}{0.567, 0.627, 0.757} 	%
\definecolor{mylightyellow}{rgb}{0.988, 0.957, 0.863} 	%

\tcbset{examplestyle/.style={
		enhanced jigsaw,	%
		colback=mylightblue,	%
		colframe=blue!10,	%
		arc=1mm,
		boxrule=0pt,		%
		left=1mm,
		right=1mm,
		left skip=0mm,  %
		right skip=0mm, %
		top=1mm,		%
		bottom=1mm,		%
		breakable,		%
		parbox = false,		%
		before={\par\pagebreak[0]\vspace{1mm}\parindent=0pt},
		after={\par\pagebreak[0]\vspace{1mm}\parindent=0pt},		
		bottomrule = 0mm,
		boxsep = 0mm,					%
		topsep at break=0pt,			%
		bottomsep at break=0pt,			%
		pad at break=0mm,
		pad before break=0mm,		
		pad after break=1mm,		
		bottomrule at break=0mm,
		toprule at break=0mm,		
		}}

\tcolorboxenvironment{example}{examplestyle}

\tcbset{qrboxstyle/.style={
		enhanced jigsaw,	%
		colback=gray!20,
		colframe=gray!40,
		arc=0mm,
		boxrule=1pt,		%
		left=1pt,
		right=1pt,
		topsep at break=1mm,			%
		top=1pt,		%
		bottom=0mm,		%
		breakable,		%
		parbox = false		%
		}}

\newcounter{resultboxenv}

\newsavebox{\coloredbgbox}

\makeatletter
\renewcommand\subsubsection{\@secnumfont}{\bfseries}%
\renewcommand\subsubsection{\@startsection{subsubsection}{3}
  \z@{.5\linespacing\@plus.7\linespacing}{-.5em}%
  {\bfseries}}  
\makeatother

\RequirePackage{scalerel}
\RequirePackage{tikz}
\usetikzlibrary{svg.path}

\definecolor{orcidlogocol}{HTML}{A6CE39}
\tikzset{
  orcidlogo/.pic={
    \fill[orcidlogocol] svg{M256,128c0,70.7-57.3,128-128,128C57.3,256,0,198.7,0,128C0,57.3,57.3,0,128,0C198.7,0,256,57.3,256,128z};
    \fill[white] svg{M86.3,186.2H70.9V79.1h15.4v48.4V186.2z}
                 svg{M108.9,79.1h41.6c39.6,0,57,28.3,57,53.6c0,27.5-21.5,53.6-56.8,53.6h-41.8V79.1z M124.3,172.4h24.5c34.9,0,42.9-26.5,42.9-39.7c0-21.5-13.7-39.7-43.7-39.7h-23.7V172.4z}
                 svg{M88.7,56.8c0,5.5-4.5,10.1-10.1,10.1c-5.6,0-10.1-4.6-10.1-10.1c0-5.6,4.5-10.1,10.1-10.1C84.2,46.7,88.7,51.3,88.7,56.8z};
  }
}

\RequirePackage{etoolbox}
\DeclareRobustCommand\orcidicon[1]{\href{https://orcid.org/#1}{\mbox{\scalerel*{
\begin{tikzpicture}[yscale=-1, transform shape]
    \pic{orcidlogo};
\end{tikzpicture}
}{|}}}}

\usepackage{adjustbox}
\def\eox{\unskip\kern 10pt{\unitlength1pt\linethickness{.4pt}$\diamondsuit${}}}

\makeatletter
\DeclareRobustCommand*\uell{\mathpalette\@uell\relax}
\newcommand*\@uell[2]{
  \setbox0=\hbox{$#1\ell$}
  \setbox1=\hbox{\rotatebox{10}{$#1\ell$}}
  \dimen0=\wd0 \advance\dimen0 by -\wd1 \divide\dimen0 by 2
  \mathord{\lower 0.1ex \hbox{\kern\dimen0\unhbox1\kern\dimen0}}
}
\makeatother

\newcommand{\introparagraph}[1]{\textbf{#1.}} %

\renewcommand{\epsilon}{\varepsilon} %
\newcommand{\set}[1]{\{#1\}}                    %

\newcommand{\datarule}{{\,:\!\!-\,}} %
\renewcommand{\vec}[1]{\boldsymbol{\mathbf{#1}}}

\newcommand{\N}{\mathbb{N}} %
\newcommand{\R}{{\mathbb{R}}} %
\newcommand{\PP}[1]{\mathbb{P} [ #1 ]} %

\newcommand{\dom}{\texttt{dom}}
\newcommand{\EVar}{\textup{\texttt{EVar}}}  	
\newcommand{\at}{\textup{\texttt{at}}}  		
\newcommand{\fact}{\mathtt{FACT}} 
\newcommand{\prob}{\mathtt{PROB}} 
\newcommand{\minfact}{\mathtt{minFACT}} 
\newcommand{\minfactilp}{\mathtt{minFACT\_ILP}} 
\newcommand{\Var}{\textup{\texttt{var}}}
\newcommand{\len}{\textup{\texttt{len}}}
 
\newcommand{\res}{\mathtt{RES}} 
\newcommand{\PTIME}{\textup{\textsf{PTIME}}\xspace} 

\newcommand{\factorizationtree}{\texttt{FT}} 
\newcommand{\veo}{\texttt{VEO}} 
\newcommand{\veoff}{\texttt{VEOFF}}

\newcommand{\Provenance}{\textup{\texttt{Prov}}}  	
\newcommand{\NP}{\textup{\textsf{NP}}\xspace}
\newcommand{\npc}{\textup{\textsf{NP-C}}\xspace}
\newcommand{\npcomplete}{\textup{\textsf{NP-complete}}\xspace}
\newcommand{\HVar}{\textup{\texttt{HVar}}}

\newcommand{\projpd}[1]{\pi^p_{#1}}
\newcommand{\projp}[1]{\pi^p_{\!-\!#1}}
\newcommand{\joinp}[2]{\Join^p\!\! \big[#2\big]}
  
\newcommand{\join}[2]{\Join\!\! \big(#2\big)}
\newcommand{\proj}[1]{\pi_{\!-\!#1}}
\newcommand{\projd}[1]{\pi_{#1}}

\newcommand{\var}{\textup{\texttt{var}}}      	

\newcommand{\true}{\texttt{true}\xspace}		
	
\newcommand{\bigset}[2]{\bigl\{ #1 \,\bigm|\, #2 \bigr\} }

\newcommand{\Count}{\texttt{count}}

\newcommand{\w}{{\vec w}}

\newcounter{myeqn}

\usepackage{xspace}            %

\newcommand{\witnesses}{\texttt{witnesses}\xspace}		%
\newcommand{\arity}{\texttt{arity}\xspace}		%

\newcommand{\np}{\mbox{{\rm NP}}}

\newcommand{\D}{D}

\newcommand{\mveo}{\mathtt{mveo}}
\newcommand{\PV}{\mathit{PV}}
\newcommand{\QPV}{\mathit{QPV}}
\newcommand{\PVF}{\mathit{PVF}}
\newcommand{\wt}{c}

\newcommand{\qtwochain}{Q_{2}^\infty}
\newcommand{\qthreechain}{Q_{3}^\infty}
\newcommand{\qfourchain}{Q_{4}^\infty}
\newcommand{\qfivechain}{Q_{5}^\infty}

\newcommand{\qthreestar}{Q_{3}^\star}
\newcommand{\qtwostar}{Q_{2}^\star}
\newcommand{\qtriangle}{Q^\triangle}
\newcommand{\qtriangleunary}{Q^\triangle_U}

\newcommand{\qdominatedtriangle}{Q_{\textrm{cod}}^\triangle}

\newcommand{\qsixcyclewe}{Q_{6\mathit{WE}}^\circ}

\definecolor{dg}{cmyk}{0.60,0,0.88,0.27}

\renewcommand{\angle}[1]{ \langle #1 \rangle }

\newcommand{\IS}{\textrm{IndSet}}

\usepackage{graphbox} %

\newtheoremstyle{dotless}{}{}{\itshape}{}{\bfseries}{}{ }{}
\theoremstyle{dotless}
\newtheorem*{specialexamplenotitle*}{}

\usepackage{multibib}
\newcites{SM}{SM References}

\usepackage{etoolbox}

\usepackage{tikz}
\newcommand*\circled[1]{\tikz[baseline=(char.base)]{
            \node[shape=circle,fill=black,draw,text=white,inner sep=1.2pt] (char) {#1};}}

\newtoggle{fullappendix}
\toggletrue{fullappendix}

\makeatletter
\gdef\@copyrightpermission{
  \begin{minipage}{0.2\columnwidth}
   \href{https://creativecommons.org/licenses/by/4.0/}{\includegraphics[width=0.90\textwidth]{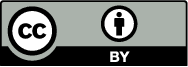}}
  \end{minipage}\hfill
  \begin{minipage}{0.8\columnwidth}
   \href{https://creativecommons.org/licenses/by/4.0/}{This work is licensed under a Creative Commons Attribution International 4.0 License.}
  \end{minipage}
  \vspace{5pt}
}
\makeatother

\begin{document}

\title[Minimally Factorizing the Provenance of Self-join Free Conjunctive Queries]{Minimally Factorizing the Provenance of Self-join Free Conjunctive Queries}

\author{Neha Makhija}
\orcid{0000-0003-0221-6836}
\affiliation{%
    \orcidicon{0000-0003-0221-6836}
	Northeastern University \country{USA}}
\email{makhija.n@northeastern.edu}

\author{Wolfgang	Gatterbauer}
\orcid{0000-0002-9614-0504}
\affiliation{%
    \orcidicon{0000-0002-9614-0504}
	Northeastern University \country{USA}}
\email{w.gatterbauer@northeastern.edu}

\begin{abstract}
We consider the problem of finding the minimal-size factorization of the provenance of self-join-free conjunctive queries, i.e.,\ 
we want to find a formula that \emph{minimizes the number of variable repetitions}. 
This problem is equivalent to solving the fundamental Boolean formula factorization problem for the restricted setting of the provenance formulas of self-join free queries.
While general Boolean formula minimization is $\Sigma^p_2$-complete, 
we show that the problem is $\npcomplete$ in our case.
Additionally, we identify a large class of queries that can be solved in $\PTIME$, expanding beyond the previously known tractable cases of read-once formulas and hierarchical queries.

We describe connections between factorizations, Variable Elimination Orders ($\veo$s), and minimal query plans. 
We leverage these insights to create an Integer Linear Program (ILP) that can solve the minimal factorization problem exactly. 
We also propose a Max-Flow Min-Cut (MFMC) based algorithm that gives an efficient approximate solution. 
Importantly,
we show that both the Linear Programming (LP) relaxation of our ILP, 
and our MFMC-based algorithm are \emph{always correct for all currently known $\PTIME$ cases}.
Thus, we present two unified algorithms (ILP and MFMC) that can both recover all known $\PTIME$ cases in $\PTIME$, 
yet also solve $\npcomplete$ cases either exactly (ILP) or approximately (MFMC), as desired.

\end{abstract}

\ccsdesc[500]{Theory of computation~Database theory}
\ccsdesc[500]{Information systems~Data management systems}

\keywords{Factorization, Provenance, Boolean Formulas}

\maketitle

\section{Introduction}
\label{sec:intro}

Given the provenance formula for a Boolean query, what is its \emph{minimal size equivalent formula}?
And under what conditions can this problem be solved efficiently?
This paper investigates the complexity of $\minfact$, i.e.\ the problem of finding a minimal factorization for 
the provenance of self-join-free conjunctive queries
(sj-free CQs).
While the general Boolean formula minimization is $\Sigma^p_2$-complete~\cite{buchfuhrer2011complexity}, 
several important tractable subclasses have been identified, such as read-once formulas ~\cite{golumbic2008improvement}. 
In this paper, we identify additional tractable cases by identifying a large class of queries for which the minimal factorization of any provenance formula can be found in $\PTIME$.

We focus on provenance formulas for two key reasons: 
1) Provenance computation and storage is utilized in numerous database applications. 
The issue of storing provenance naturally raises the question: How can provenance formulas be represented minimally? 
This problem has previously been investigated in this context ~\cite{DBLP:conf/tapp/Zavodny11,olteanu2012factorised}, where algorithms were described for factorizations with asymptotically optimal sizes, leading to work on factorized databases. 
However, finding instance-optimal factorizations i.e.\ factorizations that are guaranteed to be the smallest possible, for any arbitrary input, remains an open challenge, and is the focus of our work.

2) Minimal factorizations of provenance formulas can be used to obtain probabilistic inference bounds.
Prior approaches for 
approximate
probabilistic inference are either incomplete i.e.\ focus on just $\PTIME$ cases~\cite{DBLP:journals/vldb/DalviS07,DBLP:conf/icdt/RoyPT11,SenDeshpandeGetoor2010:ReadOnce}, or do not solve all $\PTIME$ cases exactly~\cite{DBLP:journals/vldb/DalviS07,DBLP:journals/vldb/GatterbauerS17}. 
As we show, using minimal factorization as a preprocessing step achieves the best of both worlds: It is complete (i.e.\ it applies to easy and hard cases) 
while recovering all known $\PTIME$ cases exactly.

In this paper, we prove that the minimal factorization problem is $\npcomplete$ ($\npc$) for provenance formulas, and give two algorithms for all sj-free CQs that are unified algorithms in the sense that they solve all known tractable cases in $\PTIME$, and provide approximations for hard cases. 
We further place the set of tractable queries firmly between the tractable queries for two other related problems: resilience~\cite{FreireGIM15} and probabilistic query evaluation~\cite{DBLP:conf/vldb/DalviS04}
(\cref{Fig_minFACT_venn}).

\begin{figure}
    \centering
	\hspace{-2mm}
	\includegraphics[scale=0.5]{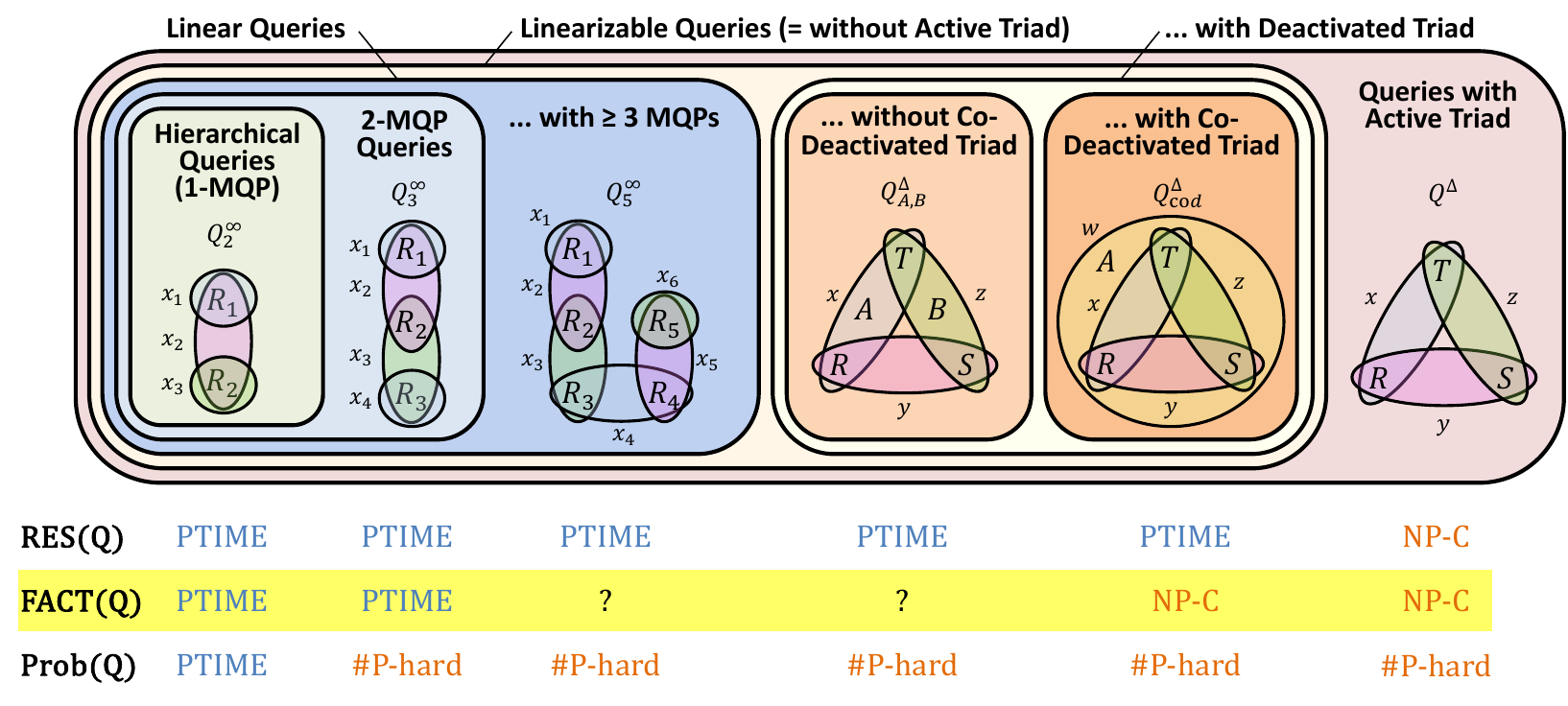}    
    \caption{This paper gives hardness results, identifies $\PTIME$ cases, and gives exact and approximate algorithms for self-join-free conjunctive queries.
	We prove that
	the tractable queries for $\minfact$ 
	reside firmly between the tractable cases for 
	probabilistic query evaluation ($\prob$) = the hierarchical queries with one minimal query plan,
	and those for resilience ($\res$) = queries without active triads.
	The open cases are 
	linear queries with $\geq 3$ minimal query plans (though we know that $\qfourchain$ is in $\PTIME$),
	and linearizable queries with deactivated triads and without co-deactivated triads (though we know that the triangle unary query $\qtriangleunary$ is in $\PTIME$).}
	\label{Fig_minFACT_venn}
\end{figure}

\introparagraph{Contributions \& Outline}
\circled{1}
The $\minfact$ problem has strong ties to the diverse problems of Boolean factorization, factorized databases, probabilistic inference, and resilience, among others. 
\cref{sec:related-work} explains these connections before \Cref{sec:background} formalizes the problem.
\circled{2}
\Cref{sec:factorizationspace} describes connections between provenance factorizations, variable elimination orders ($\veo$s) and query plans.
These connections allow us to reformulate $\minfact$ as the problem of assigning each witness to one of several ``minimal $\veo$s.''
\circled{3}
\Cref{sec:ilp} develops an ILP encoding to solve $\minfact$
for any sj-free CQ exactly.
We are not aware of any prior ILP formulation for minimal-size encodings of propositional formulas, for restricted cases like for monotone formulas.
\circled{4}
\Cref{sec:approx} describes our two unified $\PTIME$ algorithms that are exact for \emph{all known $\PTIME$ cases}, and approximations otherwise.
The first one encodes the problem in the form of a ``factorization flow graph'' s.t.\ a \emph{minimal cut} of the graph corresponds to a valid factorization of the instance. 
We refer to this algorithm as the MFMC (Max-Flow Min-Cut) based algorithm.
The second is an LP relaxation of our ILP encoding, for which we also prove a guaranteed constant factor approximation for hard queries.
\circled{5}
\Cref{sec:read:once} proves that both our unified algorithms can solve the $\minfact$ problem \emph{exactly} 
if the database instance allows a \emph{read-once factorization}.
This implies that our algorithms recover and generalize prior approaches \cite{DBLP:conf/icdt/RoyPT11} 
that are limited to read-once formulas.
\circled{6}
\Cref{sec:easycase} provides a large class of queries for which our $\PTIME$ algorithms can solve the $\minfact$ problem exactly over any database instance.
This class includes hierarchical queries as a strict subset, proving that the tractable queries for $\minfact$ are a strict superset of those for probabilistic query evaluation ($\prob$)~\cite{DBLP:journals/vldb/DalviS07}.
\circled{7}
\Cref{sec:hardcase} proves that the decision variant of $\minfact$ is $\npc$ for a set of queries that form a strict superset of queries that contain ``active triads''.
This result proves that the intractable queries for $\minfact$ are a strict superset of those that 
are intractable for resilience ($\res$)~\cite{FreireGIM15}, thereby bounding the tractable queries for our problem firmly between those tractable for $\prob$ and those tractable for $\res$.

\introparagraph{Appendix}
\iftoggle{fullappendix}{%
The appendix contains all proofs and many illustrating examples that did not fit within the main paper.
\Cref{sec:application:pdbs} shows that using minimal factorization can lead to more accurate probabilistic inference.
\Cref{sec:experiments} contains experiments evaluating the performance and results of the ILP encoding, LP relaxation, and MFMC-based algorithm.
}
{
	The appendix contains details about notations and some additional discussion about related work. However, a much longer online appendix~\cite{MinfacFull} contains all proofs, more illustrating examples, and further discussions.
	It also shows that using minimal factorization can lead to more accurate probabilistic inference ~\cite{MinfacFull}.
	We also perform experiments evaluating the performance and results of the ILP encoding, LP relaxation, or MFMC-based algorithm ~\cite{MinfacFull}.
}

\section{Related Work}
\label{sec:related-work}
\label{SEC:RELATEDWORK}

\introparagraph{Boolean Factorization}
Minimum Equivalent Expression (MEE) is the problem of deciding whether 
a given Boolean formula $\varphi$ (note that we use the terms expressions and formulas interchangeably) has a logically equivalent formula $\varphi'$ that contains $\leq k$ 
occurrences of literals.
It was known to be at least \np-hard for over 40 years~\cite[Section 7.2]{garey1979computers}
and was shown to be $\Sigma^2_p$-complete only 10 years ago~\cite{buchfuhrer2011complexity}.
The problem is more tractable for certain restrictions like
Horn formulas~\cite{hammer1993optimal} as input, 
or if allowing arbitrary Boolean functions as connectors~\cite{hemaspaandra2011minimization},
or if posed as the Minimum Formula Size Problem (MFSP) that
takes the uncompressed truth table as input~\cite{Allender:MinimizingDNF:2008,ilango2022minimum}.
There is a lot of work on approximate Boolean function factorization \cite{mintz2005factoring, martins2010boolean}, 
however efficient, exact methods are limited to classes such as read-once \cite{golumbic2008improvement} 
and read-polarity-once formulas~\cite{callegaro2013read}
(see ~\cite[Section 10.8]{CramaHammer2010:BooleanFunctions} for a detailed historical overview).
Our problem restricts the formula to be minimized to the provenance of a sj-free conjunctive query
(i.e.\ a monotone, $m$-partite DNF that follows join dependencies), 
with the goal of uncovering important classes that permit a $\PTIME$ exact evaluation.
\iftoggle{fullappendix}{
We illustrate an overview of known results in \cref{sec:appendix:rw:booleanfactorization}. 
}
{An illustration of an overview of known results can be found in the full online appendix.~\cite[Fig 7]{MinfacFull}}
To the best of our knowledge, no prior work has provided a general approach for finding the minimal factorization of monotone $k$-partite Boolean formulas given as DNFs,
and we are unaware of prior work that provides an ILP for the problem,
even under restricted settings.

\introparagraph{Factorized Databases and Related Work on Factorization}
Our problem has been studied before in the context of Factorized Databases (FDBs) \cite{olteanu2012factorised,DBLP:journals/sigmod/OlteanuS16,DBLP:conf/tapp/Zavodny11,DBLP:journals/tods/OlteanuZ15}.
Five key differences in focus are:
($i$) The tight bounds provided through that line of work are on ``readability'' i.e. the lowest $k$ such that each variable in the factorized formula is repeated at most $k$ times. 
The work shows that the class of queries with bounded readability is strictly that of hierarchical queries \cite{olteanu2012factorised}.
In contrast, we focus on the minimal number of variable repetitions and show this can be calculated in $\PTIME$ for a strict superset of hierarchical queries.
($ii$) For bounds on the minimal length (as is our focus), FDBs focus only \emph{asymptotic} bounds on the size of query result representations~\cite{DBLP:journals/tods/OlteanuZ15}
whereas we focus on minimizing the exact number of variables
(e.g.\ whether a provenance is read-once or has a factor 2 bigger size is of no relevance in the asymptotic analysis of FDBs).
($iii$) variants of FDBs permit the reuse of intermediate results, i.e.\ they focus on the corresponding \emph{circuit} size, while we focus on \emph{formulas}.
($iv$) Intuitively, FDBs study the trade-offs between applying one of several factorizations (or query plans or variable elimination orders) to the entire query results at once, 
whereas we may \emph{factorize each witness in different ways}.
($v$) Except for \cite{DBLP:conf/tapp/Zavodny11}, the work on FDBs focuses on factorizations in terms of \emph{domain values} whereas provenance formulas are defined in terms of tuple variables 
(e.g, a tuple from an arity-3 relation has 3 different domains, 
but is still represented by a single tuple variable).
These discrepancies lead to different technical questions and answers.
Also related is the very recently studied problem of finding a factorized representation of all the homomorphisms between two finite relational structures~\cite{berkholz_et_al:LIPIcs.ICALP.2023.113}. 
Similar to FDBs, that work also differs from ours in that it focuses on the asymptotic factorization size
(and proves lower bounds and allows a circuit factorization instead of a formula).
Our problem is also different from the problem of calculating a ``$p$-minimal query'' for a given query ~\cite{AmsterdamerDMT:2012:ProvananceMin}:
The solution to our problem depends on the database instance and factorizes a given provenance formula,
whereas the latter problem is posed irrespective of any given database, chooses among alternative polynomials,
and becomes trivial for queries without self-joins.

\introparagraph{Probabilistic Inference, Read-Once Formulas, and Dissociation}
Probabilistic query evaluation ($\prob$) is $\#$P-hard in general~\cite{DBLP:journals/vldb/DalviS07}.
However, if a provenance formula $\varphi$ can be represented in
\emph{read-once form} 
then its marginal probability $\PP{\varphi}$ can be computed in linear time in the number of literals.
Olteanu and Huang~\cite{DBLP:conf/sum/OlteanuH08} showed that the previously known tractable queries called hierarchical queries lead to read-once factorizations.
A query $Q$ is called hierarchical \cite{DBLP:journals/vldb/DalviS07} iff for any two existential variables 
$x,y$, one of the following three conditions holds: 
$\at(x) \subseteq \at(y)$, $\at(x) \supseteq \at(y)$, or $\at(x) \cap \at(y) = \emptyset$, where $\at(x)$ is the set of atoms of $Q$ in which $x$ participates.
Roy et al.\ \cite{DBLP:conf/icdt/RoyPT11} and Sen et al.\ \cite{SenDeshpandeGetoor2010:ReadOnce} independently proposed 
algorithms for identifying read-once provenance for non-hierarchical queries in $\PTIME$.
Notice that finding the read-once form of a formula (if it exists) 
is just an extreme case of 
representing a Boolean function by a minimum length $(\vee,\wedge)$-formula.
Our solution is a \emph{natural generalization} that is \emph{guaranteed} 
to return a read-once factorization in $\PTIME$ should there be one.
We give an interesting connection by proving that  
the tractable queries for our problem are a \emph{strict superset} of hierarchical queries and thus the tractable queries for probabilistic query evaluation.

\introparagraph{Resilience}
The resilience problem~\cite{FreireGIM15,DBLP:conf/pods/FreireGIM20}
is a variant of the deletion propagation problem~\cite{Buneman:2002,Dayal82} focusing on Boolean queries:  
Given $D\models Q$, what is the minimum number of tuples to remove
(called a ``contingency set'')
in order to
make the query false?
We give an interesting connection by proving that  
the tractable queries for our problem are a \emph{strict subset} of the tractable queries for resilience.
Concretely, we show that the structural hardness criterion for resilience also makes the factorization problem hard. We achieve separation by giving a query that is easy for resilience yet hard to factorize.
Additionally, we hypothesize that linear queries 
Addititonally, very recent concurrent work by us on resilience~\cite{MakhijaG:2024} has made a similar observation that the LP relaxation of a natural ILP formulation for resilience solves all $\PTIME$ queries exactly,
which suggests a deeper connection of our problems with general reverse data management problems~\cite{DBLP:journals/pvldb/MeliouGS11}.

\introparagraph{Linear Optimization} 
The question of when an Integer Linear Program (ILP) is tractable has many theoretical and practical consequences~\cite{cornuejols2002ideal}. 
Since we model our problem as an ILP, we can leverage some known results for ILPs to evaluate the complexity of our problem.
We show that for a certain class of queries, the constraint matrices of our ILP are Totally Unimodular~\cite{schrijver1998theory} and hence the ILP is guaranteed to be solvable in $\PTIME$.
Additionally, we find cases that \emph{do not fit known tractable classes}, such as Total Unimodularity~\cite{schrijver1998theory}, Iterative Rounding~\cite{lau2011iterative}, or Balanced Matrices~\cite{conforti2006balanced}. 
We \emph{nevertheless prove that they are in $\PTIME$ and can be solved efficiently by ILP solvers}.
We believe that additional optimization theory will be instrumental in completing the dichotomy. 
From a practical perspective, modeling our problem as an ILP allows us to use highly optimized solvers \cite{gurobi} to obtain exact results even for hard queries.

\introparagraph{Relation to Holistic Join Algorithms}
Our approach  has an interesting conceptual connection 
to ``holistic'' 
join algorithms ~\cite{10.1145/3034786.3056105}
that rely on not just a single tree decomposition (thus one query plan)
but rather multiple tree decompositions (thus multiple plans)
\emph{for different output tuples}. 
Very similarly, our approach also carefully assigns different witnesses to different query plans.

\section{Formal Setup}
\label{sec:background}
\label{SEC:BACKGROUND}

This section introduces our notation and defines the problem $\minfact$ i.e.\  the problem of finding the \emph{minimal factorization} of the provenance of a query.

\subsection{Standard database notations}
\label{sec:background:dbnotation}

We write $D$ for the database, i.e.\ the set of tuples in the relations.
A \emph{conjunctive query} (CQ) is a first-order formula $Q(\vec y)$ $= \exists \vec x\,(g_1 \wedge \ldots \wedge g_m)$ where the variables $\vec x = (x_1, \ldots, x_\ell)$ are called existential variables,
$\vec y$ are called the head
variables and each atom $g_i$ represents a relation $g_i= R_{j_i}(\vec x_i)$ where $\vec x_i \subseteq \vec x \cup \vec y$.\footnote{We follow the conventional notation for boolean CQs that omits writing the existential quantification and that replaces $\wedge$ by a comma.
W.l.o.g., we assume that $\vec x_i$ is a tuple of only variables and don't write the constants.
Selections can always be directly pushed into the database before executing the query.
In other words, for any constant in
the query, we can first apply a selection on each relation and then consider the modified query with
a column removed.
}
W.l.o.g., we discuss only connected queries.\footnote{Results for disconnected queries follow immediately by factorizing each of the query components \emph{independently}.}
We write $\var(X)$ for the set of variables occurring in atom/ relation/ query/ formula $X$ and $\at(x)$ for the set of atoms that contain variable $x$.
We write $[\vec w / \vec x]$ as a valuation (or substitution) of query variables $\vec x$ by constants $\vec w$. These substitutions may be written explicitly by ``domain-annotating'' variables with domain constants as subscripts.
\emph{Domain-annotated tuples} use such domain-annotated variables as subscripts, e.g.\ $r_{x_1,y_2}$ represents a tuple of relation $R(x,y)$ with $x=1$ and $y=2$. 
We sometimes informally omit the variables and use the notation $r_{v_1v_2\ldots v_a}$ where $v_1v_2\ldots v_a$ are the domain values of $\var(R)$ in the order that they appear in atom $R$.
Thus, $r_{12}$ also represents $R(1,2)$.
A \emph{self-join-free (sj-free) CQ} is one where no relation symbol occurs more than once 
and thus every atom represents a different relation. 
Thus, for sj-free CQs, one may refer to atoms and relations interchangeably.
We focus on Boolean queries (i.e., where $\vec y = \emptyset$), since the problem of finding the minimal factorization of the provenance for one particular output tuple of a non-Boolean query immediately
reduces to the Boolean query case
(see e.g.\ \cite{DBLP:series/synthesis/2011Suciu}).\footnote{A solution to Boolean queries immediately also provides an answer to a non-Boolean query $Q(\vec y)$: 
For each output tuple $t \in Q(D)$, solve the problem for a Boolean query $Q'$ that replaces all head variables $\vec y$ with constants of the output tuple $t$.}
Unless otherwise stated, a query in this paper denotes a \emph{sj-free Boolean CQ}.
\cref{sec:sec:appendix:nomenclature} defines further notation.

\subsection{Boolean and Provenance formulas}

The terms provenance and lineage are used in the literature with slightly different 
meanings. 
While lineage was originally formalized in \cite{DBLP:journals/tods/CuiWW00},
we follow the modern treatment of \emph{data provenance} as denoting a proposition formula that corresponds to the \emph{Boolean provenance semiring} of Green et al.\ \cite{GKT07-semirings,Green:2017:SFD:3034786.3056125}, 
which is the commutative semiring of positive Boolean expressions
$(\mathbb{B}[X], \vee, \wedge, 0, 1)$.
We sometimes write 
$\vee$ as semiring-plus ($\oplus$) 
and 
$\wedge$ as times ($\otimes$).

We assign to every tuple $t \in D$ a \emph{provenance token}, i.e. we interpret each tuple as a Boolean variable.
Then the provenance formula (equivalently, provenance expression) $\varphi_p$ of a query $Q \datarule
R_1(\vec x_1),\ldots,\\R_m(\vec x_m)$ on $D$ is the positive Boolean DNF formula 

\begin{align*}
	\Provenance(Q, D) =
	\!
	\bigvee_{\theta: D \models Q[\theta(\vec x)/\vec x]}
	\!
	R_1\big(\theta(\vec x_1)\big) \wedge\cdots\wedge R_m\big(\theta(\vec x_m)\big) 	
\end{align*}
where 
$D \models Q[\theta(\vec x)/\vec x]$ 
denotes that $\theta(\vec x)$ is a \emph{valuation} or assignment of $\vec x$ 
to constants in the active domain that make the query true over database $D$.
Notice that for sj-free queries, this DNF is always $m$-partite as each disjunct contains one tuple from each of the $m$ tables and that the notions of provenance polynomial and provenance formula are interchangeable.

\introparagraph{Read-once}
For a formula $\varphi$, we denote by $\Var(\varphi)$ the set of variables that occur in $\varphi$, 
and by $\len(\varphi)$ its length, i.e., the number of its literals.\footnote{Notice that the length of a Boolean expression $\varphi$ is also at times defined as the total number of symbols (including operators and parentheses, e.g.\ in 	\cite{CramaHammer2010:BooleanFunctions}). 
In our formulation, we only care about the number of variable occurrences.}
A provenance is called \emph{read-once}
if it can be represented in \emph{read-once} form, 
i.e.\ there is an equivalent formula in which each literal appears exactly once~\cite{GolumbicGurvich2010:ReadOnceFunctions,gurvich77,DBLP:conf/sum/OlteanuH08}.
This is possible iff that equivalent formula can be built up recursively 
from the provenance tokens by disjunction (and conjunction), s.t.\
whenever 
$\varphi = \varphi_1 \vee \varphi_2$
(or $\varphi = \varphi_1 \wedge \varphi_2$), then 
 $\var{(\varphi_1)} \cap \var{(\varphi_2)} = \emptyset$.

\introparagraph{Witnesses}
We call a \emph{witness} $\vec w$ a valuation of all variables $\vec x$ that is permitted by $D$ and that makes $Q$ \true (i.e.\ $D \models Q[\vec w/\vec x]$).\footnote{Note that our notion of witness slightly differs from the one used in  provenance literature where a ``witness'' refers to a subset of the input database records that is sufficient to ensure that a given output tuple appears in the result of a query \cite{DBLP:journals/ftdb/CheneyCT09}.} 
The set of witnesses $\witnesses(Q,D)$ (shorthand $W$) is then
\[
	\witnesses(Q,D) = \bigset{\vec w}{D \models Q[\vec w /\vec x]}\; .
\]
Since every witness implies exactly one set of $m$ tuples from $D$ 
that make the query true, 
we will slightly abuse the notation and also refer to this set of tuples as a ``witness.'' 
We will also use ``witness'' to refer to a product term in a DNF of the provenance polynomial.

\begin{figure}[]
\begin{subfigure}[b]{.4\columnwidth}
	\begin{align*}
	\begin{minipage}[t]{35mm}			
	\centering
	\setlength{\tabcolsep}{0.4mm}
		\mbox{
				\begin{tabular}[t]{ >{$}c<{$} | >{$}c<{$} >{$}c<{$} }
	 			R	&  x	\\
				\hline
				\color{blue}r_1		& \color{blue}1		\\
				r_2	& 2		\\
				r_3	& 3								
				\end{tabular}			
		}
		\hspace{0mm}
		\mbox{
				\begin{tabular}[t]{ >{$}c<{$} | >{$}c<{$} >{$}c<{$} >{$}c<{$} >{$}c<{$}}
	 			S	& x		& y\\
				\hline
				s_{11}	& 1		& 1		\\
				s_{12}	& 1		& 2		\\
				s_{23}	& 2		& 3		\\
				s_{33}	& 3		& 3		\\	
				\color{orange} s_{13}	& \color{orange}1		& \color{orange}3
				\end{tabular}
		}
		\hspace{0mm}
		\mbox{
				\begin{tabular}[t]{ >{$}c<{$} | >{$}c<{$} >{$}c<{$}}
	 			T	& y		\\
				\hline
				t_1	& 1	\\	
				t_2	& 2	\\	
				\color{blue}t_3		& \color{blue}3							
				\end{tabular}			
		}
	\end{minipage}
	\end{align*}
\caption{$D$ with and $D'$ without ${\color{orange}s_{13}}$}\label{RST_intro_D}	
\end{subfigure}
\begin{subfigure}[b]{.4\columnwidth}
	\centering
	\begin{tikzpicture}[scale=.5,
		every circle node/.style={fill=white, minimum size=4mm, inner sep=0, draw}]
	\draw (0, 3.5) node { $R$};
	\draw (3.5, 3.5) node[color=black] { $S$};	
	\draw (7, 3.5) node[color=black] { $T$};	
	\node (r1) at  (0,2) [circle]  { $r_1$};
	\node (r2) at  (0,0) [circle]  { $r_2$};
	\node (r3) at  (0,-2) [circle] { $r_3$};
	\node (t1) at  (7,2) [circle]   { $t_1$};
	\node (t2) at  (7,0) [circle]   { $t_2$};
	\node (t3) at  (7,-2) [circle]  { $t_3$};
	\foreach \from/\to/\l in {r1/t1/$s_{11}$, r1/t2/$s_{12}$, r2/t3/$s_{23}$, r3/t3/$s_{33}$}
		\draw[line width=1pt, color=black, >=latex] (\from) -- (\to) node[midway, draw=none, yshift=1.5mm] {\l} ; 
	\draw[line width=1pt, color=orange, dashed, >=latex]  (r1) -- (t3) node[midway, yshift=1.5mm]{$s_{13}$} ;  
	\end{tikzpicture}
\caption{Bipartite join graph}	
\label{fig:RST_intro:b} 
\end{subfigure}
\caption{\Cref{ex:prov,ex:fact}:
(a): Database instance with provenance tokens to the left of each tuple, e.g.\ $s_{12}$ for $S(1,2)$.
(b): $\Provenance(\qtwostar, D)$ for $\qtwostar \datarule R(x), S(x,y), T(y)$ represented as bipartite graph.
$D$ denotes the database with the orange tuple ${\color{orange}s_{13}}$ and $D'$ denotes the database without it.
}
\label{fig:RST_intro} 
\end{figure}

\begin{example}[Provenance]
    \label{ex:prov}
    Consider the Boolean 2-star query $\qtwostar \datarule R(x),$ $S(x,y), T(y)$
    over the database $D'$ in \cref{fig:RST_intro}
    (ignore the tuple  $ \color{orange}s_{13}$ for now).
    Each tuple is annotated with a Boolean variable (or provenance token)
	$ r_1, r_2,  \ldots$. 
    The \emph{provenance} $ \varphi_p$ is the Boolean expression 
	about
	which tuples need to be present for 
	$\qtwostar$ to be true:
    \begin{align}
    	\varphi_p= {\color{blue}r_1} s_{11} t_1 \vee {\color{blue}r_1} s_{12} t_2 \vee r_2 s_{23} {\color{blue}t_3} \vee r_3 s_{33} {\color{blue}t_3}
    	\label{eq:phiDNF}
    \end{align}
	This expression contains $| \Var(\varphi_p) | = 10$ variables, 
	however has a length of $\len(\varphi_p) = 12$ 
	because  variables ${\color{blue}r_1}$ and ${\color{blue}t_3}$ are repeated 2 times each.	
	The witnesses are 
	$\witnesses(Q^\star_2, D') = \{(1,1), (1,2),$ $(2,3), (3,3)\}$
	and their respective tuples are 
	$\set{r_{1},s_{11}, t_1}$, 
	$\set{r_{1},s_{12}, t_2}$, 	
	$\set{r_{2},s_{23}, t_3}$, 		
	and 
	$\set{r_{3},s_{33}, t_3}$.

    The provenance can be re-factored into a read-once factorization $\varphi'$
    which is a factorized representation of the provenance polynomial in which 
    every variable occurs once, 
	and thus $\len(\varphi') | = | \Var(\varphi') | = 10$.
	It can be found in $\PTIME$ in the size of the database~\cite{DBLP:journals/dam/GolumbicMR06}:
    \begin{align*}
    	\varphi' &= {\color{blue}r_1} (s_{11} t_1 \vee s_{12} t_2) \vee (r_2 s_{23} \vee r_3 s_{33}) {\color{blue}t_3}
    \end{align*}
	\end{example}

\subsection{Minimal factorization $\minfact$}

For a provenance $\varphi_p = \Provenance(Q, D)$ as DNF, we want to find an equivalent formula $\varphi' \equiv \varphi_p$ with the minimum number of literals.

\begin{definition}[$\fact$]
	Given a query $Q$ and database $D$, 
we say that $(D, k) \in \fact(Q)$ if there is a formula $\varphi'$ of length $\len(\varphi')\leq k$ that is equivalent to the expression $\varphi_p = \Provenance(Q, D)$.
\end{definition}

Our focus is to determine the difficulty of this problem in terms of data complexity~\cite{DBLP:conf/stoc/Vardi82}, i.e., we treat the query size $|Q|$ as a constant.
We are interested in the optimization version of this decision problem:
given $Q$ and $D$, find the \emph{minimum} $k$ such that $(D,k) \in \fact(Q)$. 
We refer to this optimization variant as the $\minfact$ problem and use $\minfact(Q,D)$ to refer to the length of the minimal size factorization for the provenance of database $D$ under query $Q$.

\begin{example}[$\fact$]
\label{ex:fact}
Now consider the provenance of $\qtwostar$ over 
the modified database $D$ with tuple ${\color{orange}s_{13}}$ from \cref{fig:RST_intro}.
It has no read-once form
and a minimal size 
formula
is
\begin{align*}
	\varphi'' &=r_1 (s_{11} t_1 \vee s_{12} t_2 \vee {\color{orange}s_{13}} {\color{blue}t_3}) 
	\vee (r_2 s_{23} \vee r_3 s_{33}) {\color{blue}t_3}
\end{align*}
We see that $\len(\varphi'') = 12$.
It follows that 
$(D,12) \in \fact(Q_2^\star)$.
At the same time,
$(D,11) \not \in \fact(Q_2^\star)$ and thus $\minfact(Q_2^\star, D)=12$. 
\end{example}

\section{Search Space for $\minfact$}
\label{sec:factorizationspace}
\label{SEC:FACTORIZATIONSPACE}

\begin{figure*}[]
	\begin{subfigure}{4cm}
		\begin{tikzpicture}[grow=right,<-, level distance=3mm, baseline=-1ex, 
			level 1/.style={sibling distance=15mm, level distance=8mm},
			level 2/.style={sibling distance=10mm},
			level 3/.style={sibling distance=3mm}]
			\node [xshift=0mm] {$\oplus$}
			child {node {$\otimes$}
				child {node {$t_1$}
				}
				child {node {$\oplus$}
					child {node {$\otimes$}
						child {node {$s_{11}$}
						}
						child {node {$r_{1}$}
						}
					}
				}
			}
			child {node {$\otimes$}
			  child {node {$r_2$}
			  }
			  child{node {$\oplus$}
				child{node [yshift=-3mm] {$\otimes$}
					child{node {$s_{22}$}
					}
					child{node {$t_{2}$}
					}
				}
				child{node [yshift=3mm] {$\otimes$}
					child{node {$s_{23}$}
					}
					child{node {$t_{3}$}
					}
				}
			  }
			};
	\end{tikzpicture}
	\caption{Factorization Tree ($\factorizationtree$)}
	\end{subfigure}
	\begin{subfigure}{6cm}
		\begin{tikzpicture}[grow=right,<-, 
			level 1/.style={sibling distance=15mm, level distance=8mm},
			level 2/.style={sibling distance=8mm, level distance=12mm},
			level 3/.style={sibling distance=8mm, level distance=12mm},
			level 4/.style={sibling distance=3mm, level distance=14mm}]
			\node [xshift=0mm] {$\pi_{-x_2,y_1}$}
			child {node {$\bowtie_{\textcolor{blue}{y_1}}$}
				child {node {$t_{\substack{1_{y_1}}}$}
				}
				child {node {$\pi_{-x_1}$}
					child {node {$\bowtie_{\textcolor{blue}{x_1}, \textcolor{orange}{y_1}}$}
						child {node {$s_{\substack{{11}_{x_1y_1}}}$}
						}
						child {node {$r_{\substack{1_{x_1}}}$}
						}
					}
				}
			}
			child {node {$\bowtie_{\textcolor{blue}{x_2}}$}
			  child {node {$r_{\substack{2_{x_2}}}$}
			  }
			  child{node {$\pi_{-y_2, y_3}$}
				child{node [yshift=-3mm] {$\bowtie_{\textcolor{orange}{x_2}, \textcolor{blue}{y_2}}$}
					child{node {$s_{\substack{{22}_{x_2y_3}}}$}
					}
					child{node {$t_{\substack{2_{y_2}}}$}
					}
				}
				child{node [yshift=3mm] {$\bowtie_{\textcolor{orange}{x_2}, \textcolor{blue}{y_3}}$}
					child{node {$s_{\substack{{23}_{x_2y_3}}}$}
					}
					child{node {$t_{\substack{3_{y_3}}}$}
					}
				}
			  }
			};
	\end{tikzpicture}
	\caption{Domain-annotated $\factorizationtree$
	}
	\end{subfigure}

	\vspace{3mm}

	\begin{subfigure}{5cm}
		\centering
		\begin{tikzpicture}[grow=right,<-, level distance=3mm, 
			level 1/.style={sibling distance=15mm, level distance=8mm},
			level 2/.style={sibling distance=5mm},
			level 3/.style={sibling distance=10mm},
			level 3/.style={sibling distance=3mm}]
			\node [xshift=0mm] {}
			child {node {$\textcolor{blue}{y_1}$}
					child {node {$\textcolor{blue}{x_1}$}
				}
			edge from parent[draw=none]}
			child {node {$\textcolor{blue}{x_2}$}
				child{node [yshift=-3mm] {$\textcolor{blue}{y_2}$}
				}
				child{node [yshift=3mm] {$\textcolor{blue}{y_3}$}
				}
			edge from parent[draw=none]
			};
	\end{tikzpicture}
	\caption{VEO Factorization Forest (\veoff)}
	\end{subfigure}
	\hspace{3mm}	%
	\begin{subfigure}[b]{3cm}
		\centering
			\begin{tikzpicture}[grow=right,<-, level distance=8mm, baseline=-1ex] 
			\node at (-0.2,0) (w1) {${\{r_1,s_{11},t_1\}}$:};
			\node at (0.8,0) (y) {$y_1$} child {node (x) {$x_1$}};
			\draw [-,opacity=0, line width =10pt, line cap=round] 
			(y.center) -- (x.center) node [below, opacity=0] {S};
			\node at (-0.2,1) (w2) {${\{r_2,s_{22},t_2\}}$:};
			\node at (0.8,1) (x2) {$x_2$} child {node (y2) {$y_2$}};
			\node at (-0.2,2) (w2) {${\{r_2,s_{23},t_3\}}$:};
			\node at (0.8,2) (x2) {$x_2$} child {node (y2) {$y_3$}};
	\end{tikzpicture}
	\vspace{-4mm}
	\caption{$\veo$ instances}
	\end{subfigure}	
\hspace{1mm}
\begin{subfigure}[b]{3cm}
\centering
	\begin{tikzpicture}[grow=right,<-, level distance=8mm] 
	\node at (0,0) (w1) {$v_2$:};
	\node at (0.5,0) (y) {$y$} child {node (x) {$x$}};
	\draw [-,opacity=0.15, line width =10pt, color=dg, line cap=round] (x.center) -- (x.center) node [above, opacity=1] {R} ;
	\draw [-,opacity=0.15, line width =10pt, color=orange, line cap=round] 
	(y.center) -- (x.center) node [below, opacity=1] {S};
	\draw [-,opacity=0.15, line width =10pt, color=blue, line cap=round] (y.center) -- (y.center) node [above, opacity=1] {T};
	\node at (0,1.5) (w3) {$v_1$:};
	\node at (0.5,1.5) (x3) {$x$} child {node (y3) {$y$}};
	\draw [-,opacity=0.15, line width =10pt, color=dg, line cap=round] (x3.center) -- (x3.center) node [above, opacity=1] {R} ;
	\draw [-,opacity=0.15, line width =10pt, color=orange, line cap=round] 
	(x3.center) -- (y3.center) node [below, opacity=1] {S};
	\draw [-,opacity=0.15, line width =10pt, color=blue, line cap=round] (y3.center) -- (y3.center) node [above, opacity=1] {T};
	\end{tikzpicture}
\vspace{-3mm}	
\caption{$\veo$s}
\label{subfig:veos}
\end{subfigure}	
\hfill
\caption{
Representation of a factorization as a mapping of witnesses to $\veo$s for an example database under query $\qtwostar$.
\cref{thm:factorizationveoff} shows the correspondence of (a) via (b) to (c). 
\cref{thm:minfacveo} shows the correspondence of (c) via (d) to (e) for some minimal factorization tree.}
\label{fig:factorizationsveos}
\end{figure*}

\introparagraph{Factorizations}
In order to find the minimal factorization of a provenance formula, 
we first define a search space of all permissible factorizations.
Each factorized formula can be represented as a \emph{factorization tree}
(or $\factorizationtree$),
where each literal of the formula corresponds to a leaf node,\footnote{A variable may appear in multiple leaves just as it can in a factorized formula.} and
internal nodes denote the $\oplus$ and $\otimes$ operators of the commutative provenance semiring.
The length (size) of a $\factorizationtree$ is the number of leaves.
We allow the semiring operations to be $k$-ary (thus even unary) 
and use \emph{prefix notation} for the operators when writing $\factorizationtree$s in linearized text.
Notice that the space of $\factorizationtree$s is strictly larger than the space of factorized expressions: 
E.g., the $\factorizationtree$ $\otimes(r_1, s_1, t_1)$ is \emph{not} equivalent to 
$\otimes(r_1, \otimes(s_1, t_1))$, 
although they represent the same formula $r_1 s_1 t_1$.
We consider $\factorizationtree$s as equivalent under commutativity i.e.\ we treat $\otimes (r_1,s_1,t_1)$ as equivalent to $\otimes (s_1,t_1,r_1)$.
Furthermore, w.l.o.g., we only consider trees in which the operators $\oplus, \otimes$ \emph{alternate}: 
E.g., $\otimes(r_1, \otimes(s_1, t_1))$ is not alternating but represents the same formula as the alternating tree $\otimes(r_1, \oplus(\otimes(s_1, t_1)))$ using unary $\oplus$. 
Henceforth, we use factorization trees or $\factorizationtree$s as a short form for \emph{alternating factorization trees}.

\introparagraph{Variable Elimination Order (\veo)}
$\factorizationtree$s describe tuple-level factorizations, however, they fail to take into account the structure (and resulting join dependencies) of the query producing the provenance. 
For this purpose, we define  query-specific
\emph{Variable Elimination Orders} ($\veo$s).
They are similar to $\veo$s in general reasoning algorithms, such as bucket elimination~\cite{Dechter:1999:Bucket}
and $\veo$s defined in FDBs~\cite{DBLP:journals/sigmod/OlteanuS16} for the case of no caching 
(i.e.\ corresponding to formulas, not circuits). 
However, our formulation allows each node to have \emph{a set of variables} instead of a single variable\iftoggle{fullappendix}{ (see e.g.\ 
\cref{ex:triangleunaryendtoend} and
\cref{Fig:TriangleUnaryEndToEnd} in the appendix)}{}.
This allows $\veo$s to have a 1-to-1 correspondence to the sequence of variables \emph{projected away} in an ``alternating'' query plan, 
in which projections and joins alternate, just as in our $\factorizationtree$s
\iftoggle{fullappendix}{(details in \cref{SEC:APPENDIX:VEOQP})}{(details in ~\cite{MinfacFull})}. 
Furthermore, we show $\veo$s can be ``annotated'' with a data instance and ``merged'' to form forests that describe a minimal factorization tree of any provenance formula.

\begin{definition}[Variable Elimination Order (\veo)]
	A $\veo$ $v$ of a query $Q$ is a rooted tree 
	whose nodes are labeled with non-empty sets of query variables 
	s.t.\ 
	($i$) each variable of $Q$ is assigned to exactly one node of $v$, and 
	($ii$) all variables $\vec x$ for any atom $R(\vec x)$ in $Q$ must occur in the prefix of some 
	node of $v$.
\end{definition}

\begin{definition}[\veo~instance]
Given a $\veo$ $v$ and witness $\w$, a \emph{\veo~instance} $v\langle\w\rangle$ 
is the rooted tree 
resulting from annotating the variables $\vec x$ in $v$ 
with the domain values of $\w$.
\end{definition}

In order to refer to a $\veo$ in-text, we use a linear notation with parentheses representing sets of children. 
To make it a unique serialization, we need to assume an ordering on the children of each parent.
For notational convenience, we leave out the parentheses for nodes with singleton sets. 
For example, $x \!\leftarrow\! y$, 
(instead of $\{x\} \!\leftarrow\! \{y\}$)
and $\{x,y\}$ are two valid $\veo$s of $\qtwostar$.
We refer to the unique path of a node to the root as its \emph{prefix}.

\begin{example}[$\veo$ and $\veo$ instance]
	Consider the 3-chain Query $\qthreechain \datarule R(x,y), S(y,z), T(z,u)$.
	An example $\veo$ is 
$v = z \!\leftarrow\! (u, y \!\leftarrow\! x)$
	\iftoggle{fullappendix}{(\cref{fig:ex:veo1:2star})}{~\cite{MinfacFull}}.
	To make it a unique serialization, we need to assume an ordering on the children of each parent.
	Notice that our definition of $\veo$ also allows sets of variables as nodes. As an extreme example, the legal query plan
	$P' = \proj{xyzu}\join{}{R(x,y),S(y,z),T(z,y)}$ corresponds to a $\veo$ $v'$ with one single node containing all variables.
	In our short notation, we denote nodes with multiple variables in brackets without commas between the variables to distinguish them from children: $v' = \{xyzu\}$. 

	Now consider a witness $\w = (1,2,3,4)$ for $(x,y,z,u)$,
	which we also write as $\w = (x_1,y_2,z_3,u_4)$.
	The $\veo$ instance of $\w$ for $v$ is then
	$v\langle\w\rangle= z_3 \!\leftarrow\! (u_4, y_2 \!\leftarrow\! x_1)$.
	Notice our notation for domain values arranged in a tree: 
	In order to make the underlying $\veo$ explicit
	(and avoiding expressions such as $v\langle\w\rangle= 3 \!\leftarrow\! (4 \!\leftarrow\! 2, 1)$ which would become quickly ambiguous) 
	we include the variable names explicitly in the $\veo$ instance.
	We sometimes refer to them as ``domain-annotated variables.''
	\label{ex:veoinstance}
	\end{example}

\begin{definition}[\veo~table prefix]
Given an atom $R$ in a query $Q$ and a $\veo$ $v$, 
the \emph{table prefix} $v^R$ is the smallest prefix in $v$ that contains all the variables $\vec x \in \var(R)$.
\end{definition}

Similarly to $v \langle\w\rangle $ denoting an instance of a given $\veo$ $v$ for a specific witness $\w$, 
we also define a \emph{table prefix instance} $v^R \langle\w\rangle$ 
for a given table prefix $v^R$ and witness $\w$.

\begin{example}[\veo~table prefix and \veo~table prefix instance]
	Consider again the $\veo$ 
	$v = z \!\leftarrow\! (u, y \!\leftarrow\! x)$
	in \cref{ex:veoinstance}.
	The table prefix of table $S(y,z)$ on $v$ is $v^S = z \!\leftarrow\! y$. 
	Assume a set of two witnesses $W= \{(x_1,y_1,z_1,u_1), (x_1,y_1,z_1,u_2)\}$.
	Then for both witnesses $\w_1$ and $\w_2$,
	the table prefix instances for $S$ are identical:
	$v^S\langle\w_1\rangle = v^S \langle\w_2\rangle = z_1 \!\leftarrow\! y_1$
	\iftoggle{fullappendix}{(\cref{fig:ex:veotableprefixinstance1:2star})}{~\cite{MinfacFull}}.		%
	\label{ex:tableprefixinstance}
	\end{example}

\begin{definition}[\veo~factorization forest (\veoff)]
	A $\veoff$ $\mathcal{V}$ of provenance $\varphi_p$ of database $D$ over query $Q$ is a forest whose nodes are labeled with non-empty sets of domain-annotated variables, such that: 
	(1) For every $\vec w \in \witnesses(Q,D)$ there exists exactly one subtree in $\mathcal{V}$ that is a $\veo$ instance of $\vec w$ and $Q$;
	(2) There is no strict sub-forest of $\mathcal{V}$ that fulfills condition (1).
\end{definition}

	\begin{example}[VEO Factorization Forest]
		Continuing with the $\qthreechain$ query, and the witnesses $W= \{(x_1,y_1,z_1,u_1), (x_1,y_1,z_1,u_2)\}$ as in \cref{ex:tableprefixinstance}, we illustrate several valid and invalid $\veoff$s \iftoggle{fullappendix}{}{(with accompanying figures in the online appendix~\cite{MinfacFull})}.
		We represent a forest of VEO instances in-text as a set of trees $\{t_1, t_2, ...\}$. 

		The forest $\mathcal{V}_1 = \{x_1 \!\leftarrow\! (u_1, u_2, (y_1 \leftarrow z_1)) \}$ \iftoggle{fullappendix}{(\cref{fig:ex:veoff1})}{} 
		is a valid $\veoff$ since (1) for both $w_1$ and $w_2$ there is exactly one subtree each in $\mathcal{V}$ that is a $\veo$ instance. These subtrees are $x_1 \!\leftarrow\! (u_1, (y_1 \leftarrow z_1))$ and $x_1 \!\leftarrow\! (u_1, (y_1 \leftarrow z_1))$. This $\veoff$ also satisfies property (2) removing any variable would lead to a $\veoff$ that does not satisfy property (1).
		
		The forest $\mathcal{V}_2 = \{x_1 \!\leftarrow\! (u_1, (y_1 \leftarrow z_1)) , y_1 \!\leftarrow\! (z_1, (x_1 \leftarrow u_1))  \}$ 		\iftoggle{fullappendix}{(\cref{fig:ex:veoff2})}{}
		is also a valid $\veoff$ since (1) for both $w_1$ and $w_2$ there is exactly one subtree each in $\mathcal{V}$ that is a $\veo$ instance. These subtrees are $x_1 \!\leftarrow\! (u_1, (y_1 \leftarrow z_1))$ and $y_1 \!\leftarrow\! (z_1, (x_1 \leftarrow u_2))$. This $\veoff$ also satisfies property (2) removing any variable would lead to a $\veoff$ that does not satisfy property (1).
	
		The forest $\mathcal{V}_3 = \{x_1 \!\leftarrow\! (u_1, (y_1 \leftarrow z_1)) , y_1 \!\leftarrow\! (z_1, (x_1 \leftarrow u_1 \leftarrow u_2))  \}$ 		\iftoggle{fullappendix}{(\cref{fig:ex:veoff3})}{}
 		is not a valid $\veoff$, although it satisfies property (1) with the same subtrees above. 
		It does not satisfy property (2) since removing $u_2$ in the second tree would lead to a $\veoff$ that still satisfies property (1).
		\label{ex:veoff}
	\end{example}

\begin{restatable}[Factorizations and $\veo$s]{theorem}{thmfactorizationveo}
There exist transformations from $\factorizationtree$s to $\veoff$s and back such that the transformations can recover the original $\factorizationtree$ for at least one minimal size  $\factorizationtree$ $\varphi'$ of any provenance formula $\varphi_p$.
\label{thm:factorizationveoff}	
\end{restatable}

\begin{proofintuition*}
We describe a transformation from $\factorizationtree$s to $\veoff$s 
via domain-annotated $\factorizationtree$s as intermediate step (\cref{fig:factorizationsveos}).
A domain-annotated $\factorizationtree$ is constructed as follows:
We first replace the $\otimes$ operator with a join ($\bowtie$) and the $\oplus$ operator with a projection ($\pi$) and label the leaves with the domain-annotated variables. 
We then recursively 
label each join and projection bottom-up as follows:
(1) label each $\bowtie$ by the union of variables of its children, 
and
(2) label each $\pi$ with the subset of variables of its children that are not required for subsequent joins (this can be inferred from the query).
To get the $\veoff$ instance, we remove all variables
on joins that appear in ancestor joins.
We remove the leaves and absorb all non-join (projection) nodes into their parents (eliminating the root projection node).

We show that if this transformation succeeds then it is a bijection and can be reversed.
The only case when this transformation fails is when it results in an empty annotation for a node, i.e.\ when there is a join after which no variable is projected away (since by design $\veo$s do not permit empty nodes). 
In that case, the $\factorizationtree$ can always be simplified by removing a $\oplus$ node and merging two $\otimes$ nodes.
\end{proofintuition*}

\begin{restatable}[$\minfact$ with $\veo$s]{theorem}{thmminfacveo}
	There exists a transformation that constructs $\factorizationtree$s of a provenance $\varphi_p$ from mappings of each witness of $\varphi_p$ to a $\veo$ of $Q$, and there exists a mapping that is transformed into a minimal size factorization tree $\varphi'$ of $\varphi_p$ under this transformation.
	\label{thm:minfacveo}
\end{restatable}

\begin{proofintuition*}

From \cref{thm:factorizationveoff}, we know that for every provenance
formula $\varphi_p$ there exists a minimal size $\factorizationtree$ that has a reversible transformation to a $\veoff$. 
We show all such $\veoff$s can be constructed by assigning a $\veo$ to each witness of $\varphi_p$.
This is constructed by defining a merge operation on $\veo$s that greedily merges common prefixes. 
\end{proofintuition*}

\introparagraph{Minimal Variable Elimination Orders ($\mveo$)}
By reducing the problem of finding the minimal factorization to that of assigning a $\veo$ to each witness, 
we have so far shown that $\fact$ is in $\NP$ with respect to data complexity.\footnote{This follows from the fact that the number of witnesses is polynomial in the size of the database, and the number of $\veo$s only depends on the query size.}
However, we can obtain a more practically efficient result by showing that we need not consider all 
$\veo$s, but only the \emph{Minimal Variable Elimination Orders} of a query or $\mveo(Q)$. 
We can define a partial order $\preceq$ on $\veo$s of a query $Q$ as follows: $v_1 \preceq v_2$ 
if for every relation $R_i \in Q$ the variables in the $R_i$ table prefix of $v_1$ are a subset of the variables of the $R_i$ table prefix of $v_2$, i.e.\ $\forall R_i \in Q: \var(v_1^{R_i}) \subseteq \var(v_2^{R_i})$.
$\mveo(Q)$ then is the set of all $\veo$s of $Q$ that are minimal with respect to this partial order $\preceq$.
For $\qtwostar$, there are only two minimal variable elimination orders $x \!\leftarrow\! y$ and $y \!\leftarrow\! x$, but not $\{x,y\}$, and for $\qthreestar$, there are only $6$, despite $13$ possible $\veo$s in total.
Interestingly, $\mveo(Q)$ corresponds exactly to Minimal Query Plans as defined in work on probabilistic databases~\cite{DBLP:journals/vldb/GatterbauerS17}, and we can use this connection to leverage prior algorithms for computing $\mveo(Q)$.

\begin{restatable}[$\minfact$ with $\mveo$s]{theorem}{thmminfactminveo}
There exists a transformation that constructs $\factorizationtree$s of a provenance $\varphi_p$ from mappings of each witness of $\varphi_p$ to a $\veo$ $v \in \mveo(Q)$, and there exists a mapping that is transformed into a minimal size factorization tree $\varphi'$ of $\varphi_p$ under this transformation.
\label{thm:minfactminveo}
\end{restatable}

\section{ILP Formulation for $\minfact$}
\label{SEC:ILP}
\label{sec:ilp}

Given a set of witnesses $W=\witnesses(Q,D)$ for a query $Q$ over some database $D$, 
we can use the insight of \cref{thm:minfactminveo} to describe a 0-1 Integer Linear Program (ILP) $\minfactilp$ \cref{eqn:ilp} that chooses a $v \in \mveo$ or equivalently a minimal query plan, for each $\w \in W$, s.t.\ the resulting factorization is of minimal size.
The size of the ILP is polynomial in $n= |D|$ and exponential in the query size.

\introparagraph{ILP Decision Variables} The ILP is based on two sets of binary variables: 
Query Plan Variables ($\QPV$) 
$q$ 
use a one-hot encoding for the choice of a minimal $\veo$ (or equivalently minimal query plan) for each witness, 
and Prefix Variables ($\PV$) $p$
encode sub-factorizations that are a consequence of that choice.
Intuitively, shared prefixes encode shared computation through factorization.

\emph{(1) Query Plan Variables} ($\QPV$):
	For each witness $\w \in W$ 
	and each minimal $\veo$ $v \in \mveo(Q)$ we define a binary variable $q[v\langle\w\rangle]$, 
	which is set to $1$ iff $\veo$ $v$ is chosen for witness
	$\w$.\footnote{Notice that we use indexing in brackets 
	$q[v\langle\w\rangle]$ 
	instead of the more common
	subscript notation $q_{v\langle\w\rangle}$ 
	since each $v\langle\w\rangle$ can depict a tree. 
	Our bracket notation is more convenient.}

\emph{(2) Prefix Variables} ($\PV$):
All witnesses must be linked to a set of prefix variables, by creating instances of $\veo$ prefixes that are in a query-specific set called the \emph{Prefix Variable Format} ($\PVF$). 
This set $\PVF$ is composed of all table prefixes of all minimal $\veo$s $v \in \mveo(Q)$. 
Notice that prefix variables can be shared by multiple witnesses, 
which captures the idea of joint factorization.
Additionally, we define a \emph{weight} 
(or cost\footnote{We write $\wt$ for weight (or cost) to avoid confusion with witnesses $\w$.}) 
$c(v^R)$ for each table prefix $v^R \in \PVF$; 
this weight is equal to 
the number of tables that have the same table prefix for a given $\veo$. 
From that $\PVF$ set, then binary prefix variables $p[v^R\langle\w\rangle]$ 
are defined for each table prefix $v^R \in \PVF$ and $\w \in W$.

\introparagraph{ILP Objective}
\label{subsec:ilp-obj}
The ILP should minimize the length of factorization $\texttt{len}$, which can be calculated by counting the number of times each tuple is written. 
If a tuple is a part of multiple witnesses, it may be repeated in the factorization. 
However, if the tuple has \emph{the same table prefix instance across different 
witnesses} (whether as part of the same $\veo$ or not),
then those occurrences are factorized together and the tuple is written just once in the factorization. 
Thus, $\texttt{len}$ is the weighted sum of all selected table prefix instances. 
The weight accounts for tuples of different tables that have the same table prefix under the same VEO.
Since $p[v^R\langle\w\rangle]=0$ for unselected table prefixes,  we can calculate $\texttt{len}$ as:

\begin{equation}
    \texttt{len} = \sum_{v^R\langle\w\rangle \in \PV} \wt(v^R) \cdot p[v^R\langle\w\rangle]
\end{equation}

\begin{figure}
	\begin{equation}
		\label{eqn:ilp}
		\begin{array}{ll@{}ll}
		\min  	& \!\!\!\!\!\sum\limits_{v^R\langle\w\rangle \in \PV} \wt(v^R) \cdot p[v^R\langle\w\rangle] \\
		\textrm{s.t.}	
			& \!\!\!\!\!\!\!\sum\limits_{v \in \mveo(Q)} \!\!\! q[v\langle\w\rangle] \geq 1, \quad &\forall \w \in W\\
			& \!\!\!\!\!\!\!p[v^R\langle\w\rangle] \geq 
			\sum_{
			v^R\langle\w\rangle 
			\textrm{ prefix of } 
			v\langle\w\rangle} 
			q[v\langle\w\rangle], \quad\ \  &\forall p[v^R\langle\w\rangle] \in \PV \\
			& \!\!\!\!\!\!\!p[v^R\langle\w\rangle] \in \{0,1\}, 	&\forall p[v^R\langle\w\rangle] \in \PV \\
			& \!\!\!\!\!\!\!q[v\langle\w\rangle]\in \{0,1\}, 		&\forall q[v\langle\w\rangle] \in \QPV
		\end{array}
	\end{equation}
	\caption{ILP Formulation for $\minfact$}
	\end{figure}

\introparagraph{ILP Constraints}
A valid factorization of $W$ must satisfy three types of constraints:

\emph{(1) Query Plan Constraints}:
    For every witness $\w \in W$, some $v \in \mveo(Q)$ must be selected.\footnote{
    We wish to have exactly one query plan or minimal $\veo$ per witness, but in a minimization problem, it suffices to say that at least one $v \in \mveo$ 
	is selected - if multiple are selected, either one of them arbitrarily still fulfills all constraints.
	}
	For example, for $\w= (x_1,y_1)$ under $\qtwostar$, 
	we enforce that:
    $q[x_1 \!\leftarrow\! y_1]
	+ 
	q[y_1 \!\leftarrow\! x_1] \geq 1$.
      
\emph{(2) Prefix Constraints}:
For any given table prefix $p$, it must be selected if any one of the $\veo$s that has it as a prefix is selected.
Since (under a minimization optimization) only one $\veo$ is chosen per witness, we can say that the value of 
$p[v^R\langle\w\rangle]$ must be at least as much as the sum of all query plan variables $q[v\langle\w\rangle]$ such that $v^R\langle\w\rangle$ is a prefix of $v\langle\w\rangle$. 
For example, we enforce that $p[x_1]$ must have value at least as much as $q[x_1 \leftarrow y_1 \leftarrow z_1]$ + $q[x_1 \leftarrow z_1 \leftarrow y_1]$.
But we cannot enforce $p[x_1] \geq q[x_1 \leftarrow y_1 \leftarrow z_1] + q[x_1 \leftarrow z_1 \leftarrow y_2]$ (as both $\veo$s do not belong to the same witness.)

\emph{(3) Boolean Integer Constraints}:
Since a $\veo$ is either selected or unselected, we set all variables in $\PV$ and $\QPV$ to 0 or 1.

\begin{restatable}[ILP correctness]{theorem}{thmilpcorrectness}\label{thm:ilpcorrectness}
	The objective of $\minfactilp$ 
	for a query $Q$ and database $D$ is always equal to $\minfact(Q,D)$.
\end{restatable}

\begin{corollary}
	$\fact$, the decision variant of $\minfact$,
	 is in \NP.
\end{corollary}

\begin{figure}[t]
	\centering
	\begin{subfigure}[b]{.42\columnwidth}
		\centering
		\hspace{6mm}
		\begin{tikzpicture}[grow=right,<-, level distance=10mm, baseline=-0.5ex] \node(y){$y$}
		    child {node(z){$z$}
		      child {node(u){$u$}
		      }
		    }
		    child {node(x){$x$}
		    }
			; 
		    \draw [-,opacity=0.15, line width =10pt, color=dg, line cap=round] 
				(y.center) -- (x.center) node [above, opacity=1] {R} ;
		    \draw [-,opacity=0.15, line width =10pt, color=brown, line cap=round] 
				(y.center) -- (z.center) node [above, opacity=1] {S};
		    \draw [-,opacity=0.15, line width =10pt, color=blue, line cap=round, label=above:\(z\)] 
				(z.center) -- (u.center) node [above, opacity=1] {T};
			\draw [-,opacity=0, line width =10pt, color=orange, line cap=round] (x.center) -- (x.center) node [below, opacity=1] {$1$};
	        \draw [-,opacity=0, line width =10pt, color=orange, line cap=round] (y.center) -- (y.center) node [below, opacity=1] {$0$};
	        \draw [-,opacity=0, line width =10pt, color=orange, line cap=round] (z.center) -- (z.center) node [below, opacity=1] {$1$};
	        \draw [-,opacity=0, line width =10pt, color=orange, line cap=round] (u.center) -- (u.center) node [below, opacity=1] {$1$};
		\end{tikzpicture}
		\caption{$v_1 = y \!\leftarrow\! (x, z \!\leftarrow\! u)$}
		\label{3Chain:VEO1}
	\end{subfigure}
	\begin{subfigure}[b]{.42\columnwidth}
		\centering
		\hspace{6mm}	
		\begin{tikzpicture}[grow=right,<-, level distance=10mm, baseline=-0.5ex] \node(z){$z$}
		    child {node(y){$y$}
		      child {node(x){$x$}
		      }
		    }
		    child {node(u){$u$}
		    }
			; 
		    \draw [-,opacity=0.15, line width =10pt, color=dg, line cap=round] 
				(y.center) -- (x.center) node [above, opacity=1] {R} ;
		    \draw [-,opacity=0.15, line width =10pt, color=brown, line cap=round] 
				(z.center) -- (y.center) node [above, opacity=1] {S};
		    \draw [-,opacity=0.15, line width =10pt, color=blue, line cap=round] 
				(z.center) -- (u.center) node [above, opacity=1] {T};
			
			\draw [-,opacity=0, line width =10pt, color=orange, line cap=round] (x.center) -- (x.center) node [below, opacity=1] {$1$};
	        \draw [-,opacity=0, line width =10pt, color=orange, line cap=round] (y.center) -- (y.center) node [below, opacity=1] {$1$};
	        \draw [-,opacity=0, line width =10pt, color=orange, line cap=round] (z.center) -- (z.center) node [below, opacity=1] {$0$};
	        \draw [-,opacity=0, line width =10pt, color=orange, line cap=round] (u.center) -- (u.center) node [below, opacity=1] {$1$};
		\end{tikzpicture}
		\caption{$v_2 = z \!\leftarrow\! (u, y \!\leftarrow\! x)$ 
		}
		\label{3Chain:VEO2}
	\end{subfigure}
	\caption{\Cref{ex:3chain}: $\mveo$ for 3-chain query $\qthreechain$.
	}\label{WeightedBipartite}
\end{figure}

\begin{example}[ILP formulation for 3-chain query]
\label{ex:3chain}
Consider the 3-Chain query $\qthreechain \datarule R(x,y),$ $ S(y,z), T(z,u)$ 
with a set of 2 witnesses $W= \{(x_1,y_1,z_1,u_1),$ $(x_1,y_1,z_1,u_2)\}$
and provenance in DNF of $r_{11}s_{11}t_{11} + r_{11}s_{11}t_{12}$.
Using the dissociation based algorithm \cite{DBLP:journals/vldb/GatterbauerS17}, 
we see that this query has $2$ minimal query plans corresponding to the two $\veo$s shown in 
\cref{WeightedBipartite}.
We use these $\veo$s to first build the set $\QPV$ (Query Plan Variables) and enforce a query plan constraint for each of the $2$ witnesses: 
\begin{align*}
    q[y_1 \!\leftarrow\! (x_1, z_1 \!\leftarrow\! u_1)]+ q[z_1 \!\leftarrow\! (u_1, y_1 \!\leftarrow\! x_1)] & \geq 1 \\
    q[y_1 \!\leftarrow\! (x_1, z_1 \!\leftarrow\! u_2)]+ q[z_1 \!\leftarrow\! (u_2, y_1 \!\leftarrow\! x_1)] & \geq 1 
\end{align*}

Then, we calculate the elements of the set $\PVF$ (Prefix Variable Format) as well as their weights. 
For the two $\veo$s from \cref{WeightedBipartite}, and the three tables $R$, $S$, $T$, we get 6 distinct table prefixes: 
\begin{center}
\renewcommand{\arraystretch}{1.2}
\setlength{\tabcolsep}{3mm}
\begin{tabular}[t]{ >{$}l<{$} | >{$}l<{$}}
	\multicolumn{1}{c |}{VEO $v_1$}				& \multicolumn{1}{c}{VEO $v_2$}		\\
	\hline
	v_1^R = y\!\leftarrow\! x					& v_2^R	= z\!\leftarrow\!y\!\leftarrow\!x\\
	v_1^S = y\!\leftarrow\! z					& v_2^S = z\!\leftarrow\!y	\\
	v_1^T = y\!\leftarrow\!z\!\leftarrow\!u		& v_2^T = z \!\leftarrow\!u
\end{tabular}	
\end{center}

\noindent
We add all these table-prefixes to the $\PVF$. 
Since no table prefix is repeated, they all are assigned weight $c=1$.
Notice that prefixes $y$ for $v_1$ and $z$ for $v_2$ are no table prefixes (and thus have weight $c= 0$
and do not participate in the objective).

From the set of table prefixes $\PVF$, 
we then create the set of prefix variables $\PV$, one for each table prefix and each witness $\w \in W$, 
and define their prefix constraints.
The prefix constraints necessary for witness $\w_1=(x_1,y_1,z_1,u_1)$ are as follows:
\begin{align*}
     p[y_1 \!\leftarrow\! x_1] &\geq q[y_1 \!\leftarrow\! (x_1, z_1 \!\leftarrow\! u_1)] &
     p[y_1 \!\leftarrow\! z_1]  & \geq q[y_1 \!\leftarrow\! (x_1, z_1 \!\leftarrow\! u_1)] \\
     p[y_1 \!\leftarrow\! z_1 \leftarrow u_1] &\geq q[y_1 \!\leftarrow\! (x_1, z_1 \!\leftarrow\! u_1)] &
     p[z_1 \!\leftarrow\! y_1 \leftarrow x_1] &\geq q[z_1 \!\leftarrow\! (u_1, y_1 \!\leftarrow\! x_1)] \\
     p[z_1 \!\leftarrow\! y_1] & \geq q[z_1 \!\leftarrow\! (u_1, y_1 \!\leftarrow\! x_1)] &
     p[z_1 \!\leftarrow\! u_1] &\geq q[z_1 \!\leftarrow\! (u_1, y_1 \!\leftarrow\! x_1)] 
\end{align*}
The prefix constraints for witness $\w_2=(x_1,y_1,z_1,u_2)$ are:
\begin{align*}
     p[y_1 \!\leftarrow\! x_1] &\geq q[y_1 \!\leftarrow\! (x_1, z_1 \!\leftarrow\! u_2)] &
     p[y_1 \!\leftarrow\! z_1]  & \geq q[y_1 \!\leftarrow\! (x_1, z_1 \!\leftarrow\! u_2)] \\
     p[y_1 \!\leftarrow\! z_1 \leftarrow u_2] &\geq q[y_1 \!\leftarrow\! (x_1, z_1 \!\leftarrow\! u_2)] &
     p[z_1 \!\leftarrow\! y_1 \leftarrow x_1] &\geq q[z_1 \!\leftarrow\! (u_2, y_1 \!\leftarrow\! x_1)] \\
     p[z_1 \!\leftarrow\! y_1] & \geq q[z_1 \!\leftarrow\! (u_2, y_1 \!\leftarrow\!x_1)] &
     p[z_1 \!\leftarrow\! u_2] &\geq q[z_1 \!\leftarrow\! (u_2, y_1 \!\leftarrow\! x_1)]
\end{align*}
Notice that we have 12 constraints (one for each pair of witness and table prefix), yet only 8 distinct prefix variables
due to common prefixes across the two witnesses 
(which intuitively enables shorter factorizations).
For this query, for every witness, there are $6$ prefix variables in the objective (some of which are used by multiple witnesses),
$1$ Query Plan constraint, and $6$ Prefix constraints.

Finally, we define the objective to minimize the weighted sum of all 8 prefix variables in $PV$ (here all weights are 1):
\begin{align*}
\texttt{len} = \,\,& 	p(y_1 \leftarrow x_1) +
			p[y_1 \leftarrow z_1]  +
			p[y_1 \leftarrow z_1 \leftarrow u_1] + 
			p[z_1 \leftarrow y_1 \leftarrow x_1] + 
			p[z_1 \leftarrow y_1]  +
			\\
			&
			p[z_1 \leftarrow u_1] + 
 		   	p[y_1 \leftarrow z_1 \leftarrow u_2] + 
			p[z_1 \leftarrow u_2]
\end{align*}

In our given database instance, $\texttt{len}$ has an optimal value of $4$ when the prefixes $p[z_1 \leftarrow y_1 \leftarrow x_1] $, $p[z_1 \leftarrow y_1]$, $p[z_1 \leftarrow u_1]$ and $p[z_1 \leftarrow u_2]$ are set to $1$. This corresponds to the minimal factorization $r_{11}s_{11}(t_{11} + t_{12})$.

\end{example}

\section{$\PTIME$ Algorithms}
\label{sec:approx}
\label{SEC:APPROX}

We provide two $\PTIME$ algorithms: 
(1) a Max-Flow Min-Cut (MFMC) based approach and 
(2) an LP relaxation from which we obtain a rounding algorithm that gives a guaranteed $|\mveo|$-approximation for all instances. 
Interestingly, \cref{SEC:EASYCASE} will later show that both algorithms 
(while generally just approximations)
give \emph{exact answers} for all currently known $\PTIME$ cases of $\minfact$.

\subsection{MFMC-based Algorithm for $\minfact$}
\label{SEC:MINCUT}
\label{sec:mincut}

Given witnesses $W$ and $\mveo(Q)$,
we describe the construction of a \emph{factorization flow graph} $F$
s.t.\ any \emph{minimal cut} of $F$ corresponds to a factorization of $W$.
A minimal cut of a flow graph is the smallest set of nodes whose removal disconnects the source ($\bot$) and target ($\top$) nodes~\cite{williamson2019network}.
Since minimal cuts of flow graphs can be found in $\PTIME$~\cite{Dasgupta:2008le}, 
we obtain is a $\PTIME$ approximation for 
$\minfact$.

\subsubsection{Construction of a factorization flow graph}
\label{sec:constructingFlow}
\label{SEC:CONSTRUCTINGFLOW}

We construct a  flow graph $F$ 
s.t.\ there exists a valid factorization of $W$ of length $\leq c$
if the graph has a cut of size $c$.
$F$ is constructed by transforming decision variables in the ILP into "cut" nodes in the flow graph 
that may be cut at a penalty equal to their weight. 
Any valid cut of the graph selects nodes that fulfill all constraints of the ILP. 
We prove this by describing the construction of $F$ (\cref{fig:mincut-general}). 

\emph{(1) $\mveo$ order}: 
	We use $\Omega$ to describe a total order on $\mveo$ (i.e.\,a total order on the set of minimal query plans). 
	In \cref{fig:mincut-general}, $\mveo$ is ordered by 
	$\Omega= (v_1, v_2, \ldots, v_k)$ where $k= |\mveo|$.

\emph{(2) $\QPV$}:
	For each witness $\vec w$, connect the query plan variables ($\QPV$) as defined by $\Omega$. 
	Since all paths from source to target must be disconnected, 
	at least one $\mveo$ must be cut from this path.
	Thus, $F$ enforces the \emph{Query Plan Constraints}.

\emph{(3) $\PV$}:
	For each witness $\vec w$ and prefix variable $p$, 
	identify the first and last query variable for which $p$ is a prefix 
	and connect the corresponding prefix variable node to connector nodes before and after these query variables.
	For example, in \cref{fig:mincut-general}, for $\w_2$, 
	$p_1$ starts at $q[{v_1\angle{\w_2}}]$ and ends at $q[{v_2\angle{\w_2}}]$
	implying that $p_1$ is a prefix for $q[{v_1\angle{\w_2}}]$ and $q[{v_2\angle{\w_2}}]$, but no query plan after that.
	Now if either of $q[{v_1\angle{\w_2}}]$ or $q[{v_2\angle{\w_2}}]$ are in the minimal cut, the graph is not disconnected until $p_1$ is added to the cut as well.
	These nodes guarantee that $F$ enforces \emph{Prefix Constraints}.

\emph{(4) Weights}:
	Assign each $q$ and $p$ node in $F$ the same weight as in the ILP objective. 
	Recall that this weight is the number of tables with the same table prefix 
	under the same $\veo$ and that it helps calculate the correct factorization length.

Thus, a min-cut of $F$ contains at least one plan for each witness, along with all the prefixes that are necessary for the plan. 
This guarantees a valid (although not necessarily minimal) factorization.\footnote{If a min-cut contains more than one plan for a witness, then one can pick either of the plans arbitrarily to obtain a valid factorization.}

\begin{figure}
    \centering
    \includegraphics[scale=0.35]{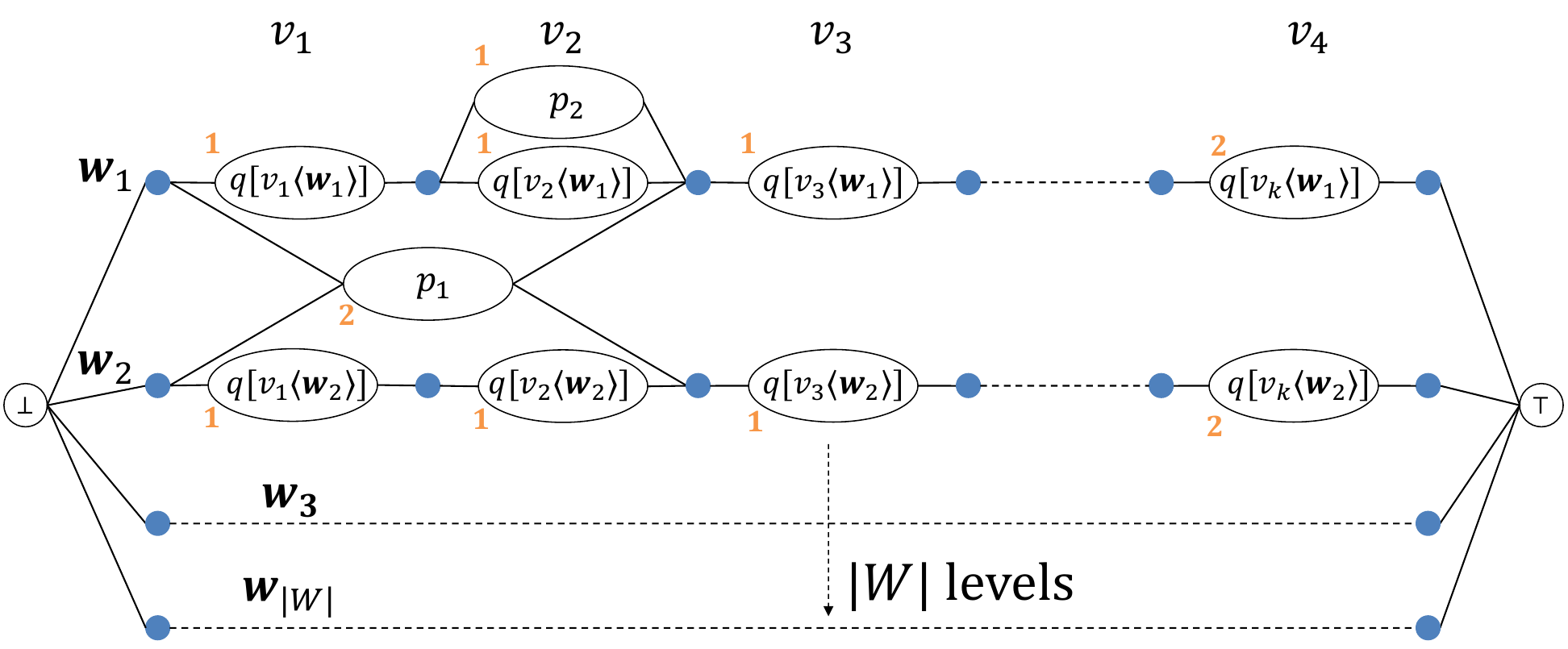}
    \caption{
	A flow graph $F$ for $\minfact$.
	The goal is to disconnect the source and the target nodes with minimum cuts.
	White $q$ and $p$ nodes can be cut and have capacities (in orange) equal to the weights of the corresponding variables in the ILP objective.
	Edges and connector nodes (in blue) have infinite capacity and cannot be cut.
} 
    \label{fig:mincut-general}
\end{figure}

\subsubsection{When is the MFMC-based algorithm optimal?}
\label{sec:heuristic-when-optimal}

In the previous subsection we saw that a min-cut of $F$ always represents a valid factorization.
However, the converse is not true:  there can be \emph{factorizations that do not correspond to a cut}.
The reason is that \emph{spurious constraints} might arise 
by the interaction of paths;
those additional constraints no longer permit the factorization.
There are two types of spurious paths:

\introparagraph{(1) Spurious Prefix Constraints} 
Spurious prefix constraints arise when a prefix node $p$ is in parallel with a query node $q$ of which it is not a prefix. This happens when a $q$ is not prefixed by $p$, 
but other query plans before and after are. 
To avoid this, the ordering $\Omega$ must be a \emph{Running-Prefixes (RP) ordering}.\footnote{Notice that this concept is reminiscent of the \emph{running intersection} property~\cite{Beeri+83,DBLP:journals/siamcomp/BernsteinG81} and the \emph{consecutive ones} property~\cite{ICDE2024:HITSnDIFFs}.}
		
\begin{definition}[Running-Prefixes (RP) ordering]
An ordering 
$\Omega= (q_1, q_2, \ldots,q_k)$
is an RP Ordering
and satisfies the RP-Property
iff for any $p$ that is a prefix for both
$q_i$ and $q_j$ ($i < j$), 
$p$ is a prefix for all $q_k$ with $i \leq k \leq j$.
\end{definition}

\begin{example}[RP ordering]
	\label{ex:rpordering}
Assume $\mveo\!=\!\{(x\!\leftarrow\!y\!\leftarrow\!z), (x\!\leftarrow\!z\!\leftarrow\!y), (z\!\leftarrow\!y\!\leftarrow\!x) \}$.
Then
$\Omega_1\!=\!((x\!\leftarrow\!y\!\leftarrow\!z), (z\!\leftarrow\!y\!\leftarrow\!x), (x\!\leftarrow\!z\!\leftarrow\!y ))$ is not an RP ordering 
since  
the 1st and 3rd
VEO share prefix $x$, however 
the 2nd starts with $z$.
In contrast, $\Omega_2 = ((x\!\leftarrow\!y\!\leftarrow\!z),(x\!\leftarrow\!z\!\leftarrow\!y),(z\!\leftarrow\!y\!\leftarrow\!x))$ is an RP ordering. 
\end{example}

It turns out that for some queries 
RP-Orderings are impossible
(such as $\qsixcyclewe$\iftoggle{fullappendix}{, see \cref{fig:6CycleWE-mveo} in \cref{sec:nestingRP}}{~\cite{MinfacFull}}). 
However, we are able to adapt our algorithm for such queries with a simple extension called \emph{nested orderings}. 
We first define two query plans as \emph{nestable} if each query plan can be ``split'' into paths from root to leaf such that they have an equal number of resulting paths, and that the resulting paths can be mapped to each other satisfying the property that corresponding paths use the same set of query variables.
\emph{Nested orderings} then are partial orders of query plans such that the pair of query plans may be uncomparable iff they are nestable. 
Finally, we define Nested RP-orderings as those such that all paths in the partial nested order satisfy the RP property.
Intuitively, nested orderings add parallel paths for a single witness to model independent decisions. 
We can now prove
that there always is an ordering that avoids spurious prefix constraints. 

\begin{restatable}[Running Prefixes (RP) Property]{theorem}{thmrpproperty}
	\label{lem:RunningPrefix}
	For any query, there is a simple or nested ordering $\Omega$ that satisfies the RP Property.	
\end{restatable}

\introparagraph{(2) Spurious Query  Constraints} 
Query Plan constraints are enforced by paths from source to target 
such that at least one node \emph{from each path} must be chosen for a valid factorization. 
Due to sharing of prefix variables, 
these paths can interact and can lead to additional \emph{spurious paths}
that place additional \emph{spurious constraints} 
on the query nodes. 
The existence of spurious query constraints does not necessarily imply that the algorithm is not optimal. 
In fact, \cref{sec:easycase} shows that  $\qtriangleunary$ and $\qfourchain$ have spurious query paths, 
yet \emph{the min-cut is guaranteed to correspond to the minimal factorization}.
However, if the presence of spurious query paths prevents any of the minimal factorizations to be a min-cut for $F$, 
then we say that the factorization flow graph $F$ has ``\emph{leakage}'' \iftoggle{fullappendix}{(\cref{ex:leakage} in appendix)}{~\cite{MinfacFull}}. 
Since all paths along one witness are cut by construction, 
a leakage path must contain nodes from at least two different witnesses.

\begin{definition}[Leakage] 
Leakage exists in a factorization flow graph 
if no minimal factorization is a valid cut of the graph. 
A leakage path is a path from source to target such that
a valid minimal factorization is possible without using any node on the path. 
\end{definition}

\introparagraph{Optimality of algorithm}
Thus, a solution found by the MFMC-based algorithm is guaranteed to be optimal
if it has two properties:
\begin{enumerate}[nosep]
    \item The ordering $\Omega$ is a \emph{Running-Prefixes ordering} or a \emph{nested Running-Prefixes ordering} (always possible).
    \item There is no \emph{leakage} in the flow graph (not always possible).
\end{enumerate}

We use these properties in \cref{sec:read:once} and \cref{sec:easycase} 
to prove a number of queries 
to be in $\PTIME$.
In fact, all currently known $\PTIME$ cases can be solved exactly with the MFMC-based algorithm via a query-dependent ordering of the $\mveo$s.

\begin{figure}
	\centering
	\begin{subfigure}[b]{.18\linewidth}
		\centering
		\begin{tikzpicture}[grow=right,<-, level distance=11mm, baseline=-0.5ex] \node(xy){$xy$} child {node(z){$z$} };
		\draw [-,opacity=0.15, line width =20pt, color=brown, line cap=round] (xy.center) -- (z.center) node [above, opacity=1] {S};
		\draw [-,opacity=0.15, line width =20pt, color=blue, line cap=round] (xy.center) -- (z.center) node [below, opacity=1] {T};
		\draw [-,opacity=0.5, line width =20pt, color=dg, line cap=round] (xy.center) -- (xy.center) node [above, opacity=1] {R} ;
		\draw [-,opacity=0, line width =40pt, color=orange, line cap=round] (xy.center) -- (xy.center) node [below, opacity=1] {$1$};
		\draw [-,opacity=0, line width =40pt, color=orange, line cap=round] (z.center) -- (z.center) node [below, opacity=1] {$2$};
		\end{tikzpicture}
		\caption{$v_1 = xy \!\leftarrow\! z$} \label{Tri:VEO1}
	\end{subfigure}
	\begin{subfigure}[b]{.18\linewidth}
		\centering
		\begin{tikzpicture}[grow=right,<-, level distance=11mm, baseline=-0.5ex] \node(yz){$yz$} child {node(x){$x$} };
		\draw [-,opacity=0.15, line width =20pt, color=blue, line cap=round] (yz.center) -- (x.center) node [below, opacity=1] {T};
		\draw [-,opacity=0.15, line width =20pt, color=dg, line cap=round] (yz.center) -- (x.center) node [above, opacity=1] {R} ;
		\draw [-,opacity=0.5, line width =20pt, color=brown, line cap=round] (yz.center) -- (yz.center) node [above, opacity=1] {S};
		\draw [-,opacity=0, line width =40pt, color=orange, line cap=round] (yz.center) -- (yz.center) node [below, opacity=1] {$1$};
		\draw [-,opacity=0, line width =40pt, color=orange, line cap=round] (x.center) -- (x.center) node [below, opacity=1] {$2$};
		\end{tikzpicture}
		\caption{$v_2 = yz \!\leftarrow\! x$} \label{Tri:VEO2}
	\end{subfigure}
	\begin{subfigure}[b]{.18\linewidth}
		\centering
		\begin{tikzpicture}[grow=right,<-, level distance=11mm, baseline=-0.5ex] \node(zx){$zx$} child {node(y){$y$} };
		\draw [-,opacity=0.15, line width =20pt, color=dg, line cap=round] (zx.center) -- (y.center) node [above, opacity=1] {R} ;
		\draw [-,opacity=0.15, line width =20pt, color=brown, line cap=round] (zx.center) -- (y.center) node [below, opacity=1] {S};
		\draw [-,opacity=0.5, line width =20pt, color=blue, line cap=round] (zx.center) -- (zx.center) node [below, opacity=1] {T};
		\draw [-,opacity=0, line width =40pt, color=orange, line cap=round] (zx.center) -- (zx.center) node [below, opacity=1] {$1$};
		\draw [-,opacity=0, line width =40pt, color=orange, line cap=round] (y.center) -- (y.center) node [below, opacity=1] {$2$};
		\end{tikzpicture}
		\caption{$v_3 = zx \!\leftarrow\! y$} \label{Tri:VEO3}
	\end{subfigure}
	\begin{subfigure}[b]{0.3\linewidth}
		\centering
		\begin{align*}
		\begin{minipage}[t]{50mm}			
		\centering
		\setlength{\tabcolsep}{0.4mm}
			\mbox{
					\begin{tabular}[t]{ >{$}c<{$} | >{$}c<{$} >{$}c<{$} }
					 R	& x & y	\\
					\hline
					r_{00}	& 0	& 0		\\
					r_{01}	& 0 & 1									
					\end{tabular}			
			}
			\hspace{2mm}
			\mbox{
					\begin{tabular}[t]{ >{$}c<{$} | >{$}c<{$} >{$}c<{$} >{$}c<{$} >{$}c<{$}}
					 S		& y		& z\\
					\hline
					s_{00}	& 0		& 0		\\
					s_{10}	& 1		& 0		
					\end{tabular}
			}
			\hspace{2mm}
			\mbox{
					\begin{tabular}[t]{ >{$}c<{$} | >{$}c<{$} >{$}c<{$}}
					 T	& z	& x		\\
					\hline
					t_{00}	& 0	& 0	\
					\end{tabular}			
			}
		\end{minipage}
		\end{align*}
	\caption{Database instance $D$
	}\label{tab:Triangle-mincut-example-database}
	\end{subfigure}
	
	\vspace{3mm}
	
	\begin{subfigure}[b]{.99\linewidth}
		\centering
		 \includegraphics[scale=0.4]{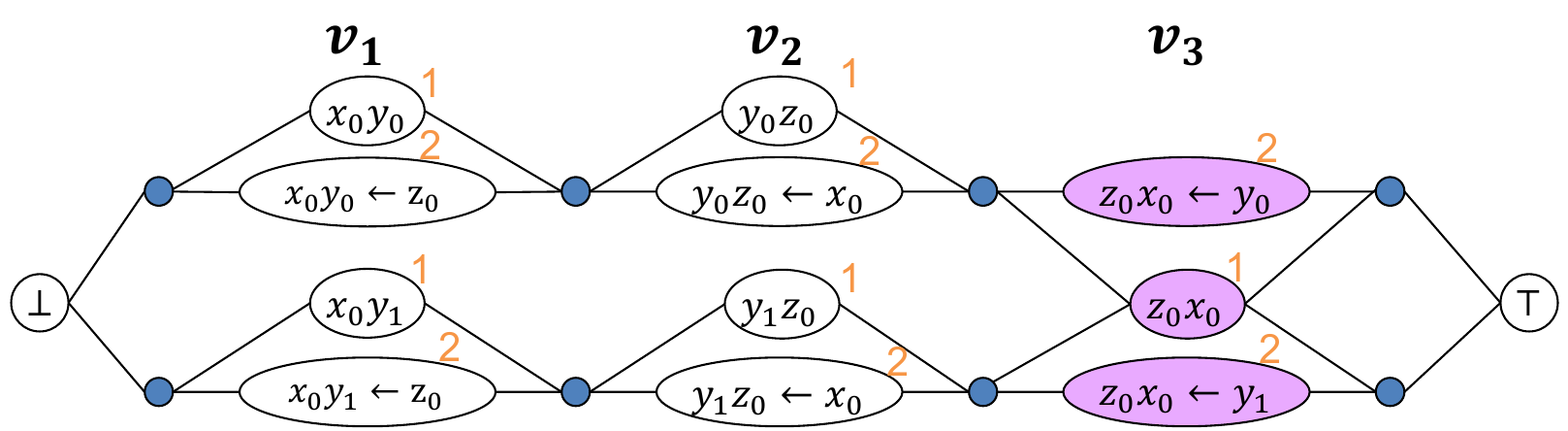}
			\caption{Flow graph $F$.
			}
		\label{fig:mincut-triangle}
	\end{subfigure}
	
	\caption{\Cref{ex:triangleILP}: 
		Three $\mveo$'s for triangle query $\qtriangle \datarule R(x,y),$ $S(y,z),T(z,x)$
		(a)-(c),
		example database instance $D$ (d),
		and constructed flow graph for $\mveo$ order $\Omega = [v_1, v_2, v_3]$ (e).
		Notice that several variables may appear \emph{in the same} node of a $\mveo$ (e.g., $x$ and $y$ in $xy \leftarrow z$).
		}
	\label{Triangle-VEOs}
	\end{figure}

	\vspace{5mm}

	\begin{example}[Flow graph construction for Triangle Query]

	\label{ex:triangleILP}
	Consider the triangle query 
	$\qtriangle \datarule R(x,y), S(y,z), T(z,x)$. 
	The query has $3$ minimal Query Plans corresponding to $\mveo$s	%
	shown in \cref{Tri:VEO1,Tri:VEO2,Tri:VEO3}.
	The provenance
	of $\qtriangle$ over the database shown in \cref{tab:Triangle-mincut-example-database},
	has $2$ witnesses: $W = \{r_{00}s_{00}t_{00}, r_{01}s_{10}t_{00}\}$.
	We build a flow graph to find a factorization.
	(1) We choose 
	$\Omega=(v_1, v_2, v_3)$ 
	as linear order for the $\mveo$. 
	(2) For each witness, we connect their three query plan variables 
	$q[v \langle \w \rangle]$ in this order serially from source to target. 
	(3) In $\qtriangle$, each $\mveo$ has a single prefix. We attach these variables in parallel to their corresponding query variables. Notice that the prefix $z_0x_0$ is shared by both $\w_1$ and $\w_2$, and therefore is attached in parallel to both corresponding query variables.
	(4) Finally, we add weights corresponding to the number of tables having each prefix (see \cref{Tri:VEO1,Tri:VEO2,Tri:VEO3}).
	
	The resulting flow graph is shown in \cref{fig:mincut-triangle}. 
	The min-cut (highlighted in purple) consists of the nodes 
	$\{z_0x_0 \!\leftarrow\! y_0, z_0x_0, z_0x_0 \!\leftarrow\! y_1\}$.
	The corresponding factorization using the selected query plans is $t_{00}(r_{00}s_{00} \vee r_{01}s_{10})$. 
	The weighted cut-value ($5$) is equal to the length of the factorization.
	This factorization is minimal.
	\end{example}

\subsection{LP relaxation for $\minfact$ and an LP relaxation-based approximation}
\label{SEC:LP-ROUNDING}
\label{sec:lp-rounding}

Linear Programming relaxation and rounding is a commonly-used technique to find $\PTIME$ approximations for $\npc$ problems \cite{vazirani2001approximation}. 
The LP relaxation for $\minfact$ simply removes the integrality constraints on all the problem variables.
The LP relaxation may pick multiple query plans for a given witness, each with fractional values.
We present a rounding scheme for $\minfact$ and show it to be a $|\mveo|$-factor approximation of the optimal solution.
The rounding algorithm simply picks the maximum fractional value of query plan variables for each witness, breaking ties arbitrarily. 
Finally, it counts only the prefix variables necessitated by the chosen query plans. 

\begin{restatable}{theorem}{thmlprounding}
	\label{thm:lp-rounding}
	The described rounding scheme gives a $\PTIME$, $|\mveo|$-factor approximation for $\minfact$.
\end{restatable}

\subsubsection{\textbf{When is the LP relaxation optimal?}}
\label{sec:lp-when-optimal}

Experimentally, we observed that the LP solution of many queries 
are equal to the integral ILP solution.
This is surprising since the ILP does not satisfy any of the known requirements for tractable ILPs such as Total Unimodularity, Balanced Matrices, or even Total Dual Integrality~\cite{schrijver1998theory}.
Interestingly, we next prove that the LP relaxation has the same objective value as the original ILP whenever the MFMC-based algorithm is optimal.

\begin{restatable}{lemma}{lemlpeasywhenmincuteasy}\label{prop:lp-easy-mincut}
If all database instances can be solved exactly by the MFMC-based algorithm 
for a given query $Q$
(i.e.\ for each database instance there exists an ordering that generates a leakage-free graph), then the LP relaxation of $\minfact$ always has the same objective as the original ILP.
\end{restatable}

This result is important as it 
exposes cases for which the optimal objective value of $\minfactilp$ is identical to the optimal objective value of a simpler LP relaxation.
The only example we know of where this has been shown using flow graphs is in recent work on resilience \cite{MakhijaG:2024}. 
In such cases, \emph{standard ILP solvers return the optimal solution to the original ILP in $\PTIME$}.
This is due to ILP solvers using an LP-based branch and bound approach which starts by computing the LP relaxation bound and then exploring the search space to find integral solutions that move closer to this bound.
If an integral solution is encountered that is equal to the LP relaxation optimum, then the solver is done~\cite{gurobi_working}.
We use this knowledge in \cref{sec:easycase} to show that ILP solvers can solve all known $\PTIME$ cases in $\PTIME$.

\section{Recovering Read-Once instances}
\label{sec:read:once}
\label{SEC:READONCE}

It is known that the read-once factorization of a read-once instance can be found in $\PTIME$ with specialized algorithms \cite{DBLP:journals/dam/GolumbicMR06,DBLP:conf/icdt/RoyPT11,SenDeshpandeGetoor2010:ReadOnce}. 
We prove that our more general MFMC based algorithm and LP relaxation 
are \emph{always guaranteed to find read-once formulas} when they exist, even though they are not specifically designed to do so.

\begin{restatable}[Read-Once]{thm}{thmreadonce}\label{prop:read-ONCE}
	$\minfact$ can be found in \PTIME by
	\begin{thmlist}
	\item the MFMC based algorithm, and \label{prop:readonce:Mincut}
	\item the LP relaxation\label{prop:readonce:LP},
	\end{thmlist}
	for any query and database instance that permits a read-once factorization.
\end{restatable}

\section{Tractable Queries for $\minfact$}
\label{SEC:EASYCASE}
\label{sec:easycase}

We now go beyond cases when $\prob$ is in \PTIME (i.e.\ read-once instances). We first prove that $\minfact(Q)$ is $\PTIME$ for the large class of queries with $2$ minimal query plans. 
We then show examples of queries with $3$ and $5$ minimal query plans that are $\PTIME$ as well.
All these newly recovered $\PTIME$ cases, along with the previously known read-once cases, 
can be solved \emph{exactly} with both our $\PTIME$ algorithms from \cref{sec:approx}.
Finally, we hypothesize that $\minfact$ is in $\PTIME$
for any linear query.

\subsection{All queries with $\leq$2 minimal query plans}

We prove that our MFMC-based algorithm has no leakage and thus always finds the minimal factorization for queries with at most $2$ minimal $\veo$s ($2$-MQP) queries such as $\qtwostar$ and $\qthreechain$.
We also give an alternative proof that shows that any ILP generated by such a query
is guaranteed to have a \emph{Totally Unimodular (TU)} constraint matrix, and thus is $\PTIME$ solvable \cite{schrijver1998theory}.

\begin{restatable}[2-MQP Queries]{thm}{thmtwoqueryplans}\label{prop:2queryplans}
	$\minfact$ can be found in \PTIME for any query with max $2$ minimal $\veo$s by
	\begin{thmlist}
	\item the MFMC based algorithm, and \label{prop:2queryplans:Mincut}
	\item the LP relaxation\label{prop:2queryplans:LP}.
	\end{thmlist}
\end{restatable}

The theorem recovers the hierarchical queries which are equivalent to $1$-MQP queries since they have one ``safe plan'' \cite{DBLP:journals/vldb/DalviS07}.
The $\PTIME$ nature of $1$-MQP queries also follows from \cref{prop:read-ONCE}, as all hierarchical queries have read-once formulations. 

\begin{restatable}{corollary}{corrhierachicalqueries}
	\label{corr:hierachicalquerieseasy}
	$\minfact$ for Hierarchical Queries is in $\PTIME$.
\end{restatable}

\begin{restatable}{corollary}{corrminfacprobinference}
	\label{corr:minfacharderthanprob}
	The classes of queries for which $\minfact$ is in $\PTIME$ is a strict super-class of those for which probabilistic query evaluation is in $\PTIME$ (if $P \neq NP$).
\end{restatable}
	
\subsection{Two queries with $\geq 3$ minimal query plans}	
	
\introparagraph{Triangle-unary $\qtriangleunary$} 
$\qtriangleunary$ is structurally similar to
$\qtriangle$ (\cref{fig:triangle:variants})
and both have $|\mveo|=3$. 
However, while $\qtriangle$ contains an ``active triad''~\cite{MakhijaG:2024}
and is hard, we show that $\qtriangleunary$ is in $\PTIME$ 
by proving that the factorization flow graph has no leakage.
Interestingly, $\qtriangleunary$'s ILP is \emph{not guaranteed to have a TU} constraint matrix, yet the MFMC algorithm is optimal, and the LP relaxation recovers the minimal ILP objective,
showing that $\PTIME$ cases extend beyond Total Unimodularity of the ILP constraint matrix.

\begin{restatable}[$\qtriangleunary$ is easy]{thm}{thmtriangleunary}\label{prop:triangleUNARY}
	$\minfact(\qtriangleunary, D)$ can be found in \PTIME for any database $D$ by 
	\begin{thmlist}
		\item the MFMC based algorithm, and \label{prop:triangleUNARY:Mincut}
		\item the LP relaxation.\label{prop:triangleUNARY:LP}
	\end{thmlist}
\end{restatable}

\introparagraph{4-chain $\qfourchain$}
This is arguably the most involved proof in the paper. 
$\qfourchain$ has $|\mveo|=5$. 
Yet in a similar proof to $\qtriangle$, we can show that the MFMC-based algorithm and the LP are both optimal.
This surprising result leads to the conjecture that $\minfact$ for longer chains, and all linear queries are in $\PTIME$.

\begin{restatable}[$\qfourchain$ is easy]{thm}{thmfourchain}\label{prop:fourchain}
	$\minfact(\qfourchain, D)$ can be found in \PTIME for any database $D$ by
	\begin{thmlist}
	\item the MFMC based algorithm, and \label{prop:fourchain:Mincut}
	\item the LP relaxation. \label{prop:fourchain:LP}
	\end{thmlist}
\end{restatable}

\begin{figure}
\begin{subfigure}[b]{.45\linewidth}
	\centering
	\begin{tikzpicture}[scale=.5,
				every circle node/.style={fill=white, minimum size=7mm, inner sep=0, draw}]
	\node (4) at (3,4) [] {$x$};	
	\node (1) at (0,0) [] {$y$};
	\node (3) at (6,0) [] {$z$};
	
	\node (10) at (0.5,2.5) [blue] {$R$};
	\node (11) at (3,-1) [blue] {$S$};
	\node (12) at (5.5,2.5) [blue] {$T$};

	\path[->, line width=1pt, auto]
		(4) edge[color=black] (1)
		(1) edge[color=black] (3)
		(3) edge[color=black] (4)	
		;
	\end{tikzpicture}
	\vspace{-2mm}
	\caption{Triangle Query $q_{\triangle}$}
	\label{fig:triangle-query}
\end{subfigure}
~	
\begin{subfigure}[b]{.45\linewidth}
	\centering
	\begin{tikzpicture}[scale=.5,
				every circle node/.style={fill=white, minimum size=7mm, inner sep=0, draw}]
	\node (4) at (3,4) [] {$x$};	
	\node (1) at (0,0) [] {$y$};
	\node (3) at (6,0) [] {$z$};
	\node (10) at (0.5,2.5) [blue] {$R$};
	\node (11) at (3,-1) [blue] {$S$};
	\node (12) at (5.5,2.5) [blue] {$T$};
	\node (12) at (2,5) [blue] {$U$};

	\draw[transform canvas={yshift=+0.0ex},->, line width=1pt, out=120, in=60, distance=1.5cm](4) to (4);

	\path[->, line width=1pt, auto]
		(4) edge[color=black] (1)
		(1) edge[color=black] (3)
		(3) edge[color=black] (4)	
		;
	\end{tikzpicture}
	\vspace{-2mm}
	\caption{Triangle-Unary Query $q_{\triangle U}$}
	\label{fig:triangle-unary-query}
\end{subfigure}
\caption{We show that $\qtriangle$ is hard in \cref{sec:hardcase} because it contains an ``active triad.'' 
Surprisingly, for $\qtriangleunary$, a query that differs by a single unary relation, 
the minimal factorization can always be found in $\PTIME$ by either
using our MFMC based algorithm from \cref{SEC:MINCUT}
or our LP relaxation from \cref{sec:lp-rounding}.
}
\label{fig:triangle:variants}
\end{figure}

\subsection{Conjecture for Linear Queries}
\label{sec:dichotomy}
\label{sec:linearcase}

A query is acyclic if it has a join tree, i.e.\ it permits a placement of its atoms into a tree s.t. 
for any two atoms, the intersection of variables is contained in the union of the variables of the atoms on the unique path between them.\footnote{The concept is alternatively called coherence, the running intersection property, connected subgraph property~\cite{Beeri+83,DBLP:journals/siamcomp/BernsteinG81,DBLP:journals/siamcomp/TarjanY84},
and is used
in the definition of
the junction tree algorithm \cite{10.2307/2345762} and tree decompositions~\cite{DBLP:journals/jal/RobertsonS86,Dechter:2003aa}.}
A query is \emph{linear} if it permits a join path.\footnote{This definition, introduced in \cite{MakhijaG:2024}, is more restrictive than linear queries defined in the original work on resilience~\cite{FreireGIM15} as it does not allow \emph{linearizable queries} (those that can be made linear by ``making \emph{dominated} atoms exogenous'').}
We have spent a lot of time trying to prove the hardness of such queries without success. 
Based on our intuition we hypothesize that \emph{all linear queries are in $\PTIME$}.
Our intuition is strengthened by the fact that over many experimental evaluations, 
the LP relaxation of the $\minfactilp$ was always integral and optimal, 
thus being able to solve the problem in $\PTIME$.

\begin{hypothesis}[$\PTIME$ conjecture]
	If $Q$ is a linear query, then
	$\minfact(Q,D)$ can be found in $\PTIME$ for any database $D$.	
\end{hypothesis}	

We think that additional insights from optimization theory are needed to explain the integrality of the solution to the LP relaxation
and to thus prove this conjecture.
We leave open the structural criterion that separates the easy and hard cases.

\section{Hard Queries for $\minfact$}
\label{SEC:HARDCASE}
\label{sec:hardcase}

In this section, we first prove that all queries that contain a structure called ``\emph{an active triad}'' 
(e.g. $\qthreestar$ and $\qtriangle$) are $\npc$.
We then prove another query to be $\npc$
that does not contain an active triad, but a ``\emph{co-deactivated triad}.'' 
We thus show that while active triads are sufficient for $\fact$ of a query to be $\npc$,
they are not necessary, and 
$\minfact$ is a strictly harder problem than $\res$.

\introparagraph{Queries with Active triads}
We repeat here the necessary definitions introduced in the context of resilience under bag semantics~\cite{MakhijaG:2024}.
A \emph{triad} is a set of three atoms, $\mathcal{T} = \set{R_1, R_2, R_3}$  
s.t.\ for every pair $i \neq j$, there is a path from $R_i$ to $R_j$ that uses no variable occurring in the third atom of $\mathcal{T}$.
Here a \emph{path} 
is an alternating sequence of relations and variables
$R_1-\vec x_1-R_2- \cdots \vec x_{p-1}-R_p$
s.t.\ all adjacent relations 
$R_i, R_{i+1}$ 
share variables $\vec x_i$.
In a query $Q$ with atoms $R$ and $S$, we say $R$ \emph{dominates} $S$ iff $\var(R) \subset \var(S)$.
We call an atom $g$ in a query \emph{independent} iff there is no other atom in the query that contains a strict subset of its variables (and hence it is not dominated).
A triad is \emph{active} iff none of its atoms are dominated.

\begin{restatable}[Active Triads are hard]{theorem}{thmactivetriadshard}
	\label{prop:triadS}
	$\fact(Q)$ for a query $Q$ with an active triad is $\npc$.
\end{restatable}

\introparagraph{Separation between $\res$ and $\minfact$}	
A triad is \emph{deactivated} if any of the three atoms is dominated.
A triad is \emph{co-deactivated} if all three atoms are dominated only by the same (non-empty) set of atoms.
The co-deactivated triangle query
$\qdominatedtriangle \datarule A(w),
R(w,x,y),$ $S(w, y, z), $ $T(w, z, x)$
contains no active triads: notice that the tables $R$, $S$ and $T$ are not independent and have no independent paths to each other. 
Thus, $\res(\qdominatedtriangle)$ is $\PTIME$. 
However, $\qdominatedtriangle$ contains a co-deactivated triad since $R$, $S$ and $T$ are all dominated only by atom $A$.
We next prove that $\fact(\qdominatedtriangle)$ is $\npc$, thus showing a strict separation in the complexities of the two problems.

\begin{restatable}[Co-Deactivated Triads are hard]{theorem}{thmsemideactivatedtriadhard}
	\label{prop:dominatedtrianglehard}
	$\fact(Q)$  for a query $Q$ with a co-deactivated triad is $\npc$.
\end{restatable}

\section{Conclusion}

We propose an ILP framework for minimizing the size of provenance polynomials for sj-free CQs.
We show that our problem is $\npc$
and thus in a lower complexity class than the general Minimum Equivalent Expression (MEE) problem.
Key to our formulation is a way to systematically constrain a space of possible minimum factorizations thus allowing us to build an ILP, 
and connecting minimal variable elimination orders to minimal query plans developed in the context of probabilistic databases.
We complement our hardness results with two unified $\PTIME$ algorithms that can recover \emph{exact solutions} to a strict superset of \emph{all prior known tractable cases}.

\section*{Acknowledgements}

This work was supported in part by the National Science Foundation (NSF) under award numbers IIS-1762268 and IIS-1956096, 
and conducted in part while the authors were visiting the Simons Institute for the Theory of Computing.

\bibliographystyle{ACM-Reference-Format}
\bibliography{BIB/propagation.bib}


\begin{thebibliography}{72}


\ifx \showCODEN    \undefined \def \showCODEN     #1{\unskip}     \fi
\ifx \showDOI      \undefined \def \showDOI       #1{#1}\fi
\ifx \showISBNx    \undefined \def \showISBNx     #1{\unskip}     \fi
\ifx \showISBNxiii \undefined \def \showISBNxiii  #1{\unskip}     \fi
\ifx \showISSN     \undefined \def \showISSN      #1{\unskip}     \fi
\ifx \showLCCN     \undefined \def \showLCCN      #1{\unskip}     \fi
\ifx \shownote     \undefined \def \shownote      #1{#1}          \fi
\ifx \showarticletitle \undefined \def \showarticletitle #1{#1}   \fi
\ifx \showURL      \undefined \def \showURL       {\relax}        \fi
\providecommand\bibfield[2]{#2}
\providecommand\bibinfo[2]{#2}
\providecommand\natexlab[1]{#1}
\providecommand\showeprint[2][]{arXiv:#2}

\bibitem[Abo~Khamis et~al\mbox{.}(2017)]%
        {10.1145/3034786.3056105}
\bibfield{author}{\bibinfo{person}{Mahmoud Abo~Khamis}, \bibinfo{person}{Hung~Q. Ngo}, {and} \bibinfo{person}{Dan Suciu}.} \bibinfo{year}{2017}\natexlab{}.
\newblock \showarticletitle{What Do Shannon-Type Inequalities, Submodular Width, and Disjunctive Datalog Have to Do with One Another?}. In \bibinfo{booktitle}{\emph{{PODS}}}. \bibinfo{pages}{429--444}.
\newblock
\showISBNx{9781450341981}
\urldef\tempurl%
\url{https://doi.org/10.1145/3034786.3056105}
\showDOI{\tempurl}


\bibitem[Allender et~al\mbox{.}(2008)]%
        {Allender:MinimizingDNF:2008}
\bibfield{author}{\bibinfo{person}{Eric Allender}, \bibinfo{person}{Lisa Hellerstein}, \bibinfo{person}{Paul McCabe}, \bibinfo{person}{Toniann Pitassi}, {and} \bibinfo{person}{Michael Saks}.} \bibinfo{year}{2008}\natexlab{}.
\newblock \showarticletitle{Minimizing Disjunctive Normal Form Formulas and AC$^0$ Circuits Given a Truth Table}.
\newblock \bibinfo{journal}{\emph{SIAM J. Comput.}} \bibinfo{volume}{38}, \bibinfo{number}{1} (\bibinfo{year}{2008}), \bibinfo{pages}{63--84}.
\newblock
\urldef\tempurl%
\url{https://doi.org/10.1137/060664537}
\showDOI{\tempurl}


\bibitem[Amsterdamer et~al\mbox{.}(2012)]%
        {AmsterdamerDMT:2012:ProvananceMin}
\bibfield{author}{\bibinfo{person}{Yael Amsterdamer}, \bibinfo{person}{Daniel Deutch}, \bibinfo{person}{Tova Milo}, {and} \bibinfo{person}{Val Tannen}.} \bibinfo{year}{2012}\natexlab{}.
\newblock \showarticletitle{On Provenance Minimization}.
\newblock \bibinfo{journal}{\emph{ACM Trans. Database Syst.}} \bibinfo{volume}{37}, \bibinfo{number}{4}, Article \bibinfo{articleno}{30} (\bibinfo{date}{dec} \bibinfo{year}{2012}), \bibinfo{numpages}{36}~pages.
\newblock
\showISSN{0362-5915}
\urldef\tempurl%
\url{https://doi.org/10.1145/2389241.2389249}
\showDOI{\tempurl}


\bibitem[Beeri et~al\mbox{.}(1983)]%
        {Beeri+83}
\bibfield{author}{\bibinfo{person}{Catriel Beeri}, \bibinfo{person}{Ronald Fagin}, \bibinfo{person}{David Maier}, {and} \bibinfo{person}{Mihalis Yannakakis}.} \bibinfo{year}{1983}\natexlab{}.
\newblock \showarticletitle{On the Desirability of Acyclic Database Schemes}.
\newblock \bibinfo{journal}{\emph{J. ACM}} \bibinfo{volume}{30}, \bibinfo{number}{3} (\bibinfo{date}{July} \bibinfo{year}{1983}), \bibinfo{pages}{479--513}.
\newblock
\showISSN{0004-5411}
\urldef\tempurl%
\url{https://doi.org/10.1145/2402.322389}
\showDOI{\tempurl}


\bibitem[Berkholz and Vinall-Smeeth(2023)]%
        {berkholz_et_al:LIPIcs.ICALP.2023.113}
\bibfield{author}{\bibinfo{person}{Christoph Berkholz} {and} \bibinfo{person}{Harry Vinall-Smeeth}.} \bibinfo{year}{2023}\natexlab{}.
\newblock \showarticletitle{{A Dichotomy for Succinct Representations of Homomorphisms}}. In \bibinfo{booktitle}{\emph{{ICALP}}} \emph{(\bibinfo{series}{LIPIcs}, Vol.~\bibinfo{volume}{261})}, \bibfield{editor}{\bibinfo{person}{Kousha Etessami}, \bibinfo{person}{Uriel Feige}, {and} \bibinfo{person}{Gabriele Puppis}} (Eds.). \bibinfo{pages}{113:1--113:19}.
\newblock
\showISBNx{978-3-95977-278-5}
\showISSN{1868-8969}
\urldef\tempurl%
\url{https://doi.org/10.4230/LIPIcs.ICALP.2023.113}
\showDOI{\tempurl}


\bibitem[Bernstein and Goodman(1981)]%
        {DBLP:journals/siamcomp/BernsteinG81}
\bibfield{author}{\bibinfo{person}{Philip~A. Bernstein} {and} \bibinfo{person}{Nathan Goodman}.} \bibinfo{year}{1981}\natexlab{}.
\newblock \showarticletitle{Power of Natural Semijoins}.
\newblock \bibinfo{journal}{\emph{{SIAM} J. Comput.}} \bibinfo{volume}{10}, \bibinfo{number}{4} (\bibinfo{year}{1981}), \bibinfo{pages}{751--771}.
\newblock
\urldef\tempurl%
\url{https://doi.org/10.1137/0210059}
\showDOI{\tempurl}


\bibitem[Boulos et~al\mbox{.}(2005)]%
        {DBLP:conf/sigmod/BoulosDMMRS05}
\bibfield{author}{\bibinfo{person}{Jihad Boulos}, \bibinfo{person}{Nilesh~N. Dalvi}, \bibinfo{person}{Bhushan Mandhani}, \bibinfo{person}{Shobhit Mathur}, \bibinfo{person}{Christopher R{\'{e}}}, {and} \bibinfo{person}{Dan Suciu}.} \bibinfo{year}{2005}\natexlab{}.
\newblock \showarticletitle{{MYSTIQ:} a system for finding more answers by using probabilities}. In \bibinfo{booktitle}{\emph{{SIGMOD}}}. \bibinfo{pages}{891--893}.
\newblock
\urldef\tempurl%
\url{https://doi.org/10.1145/1066157.1066277}
\showDOI{\tempurl}


\bibitem[Buchfuhrer and Umans(2011)]%
        {buchfuhrer2011complexity}
\bibfield{author}{\bibinfo{person}{David Buchfuhrer} {and} \bibinfo{person}{Christopher Umans}.} \bibinfo{year}{2011}\natexlab{}.
\newblock \showarticletitle{The complexity of Boolean formula minimization}.
\newblock \bibinfo{journal}{\emph{J. Comput. System Sci.}} \bibinfo{volume}{77}, \bibinfo{number}{1} (\bibinfo{year}{2011}), \bibinfo{pages}{142--153}.
\newblock
\urldef\tempurl%
\url{https://doi.org/10.1016/j.jcss.2010.06.011}
\showDOI{\tempurl}


\bibitem[Buneman et~al\mbox{.}(2002)]%
        {Buneman:2002}
\bibfield{author}{\bibinfo{person}{Peter Buneman}, \bibinfo{person}{Sanjeev Khanna}, {and} \bibinfo{person}{Wang-Chiew Tan}.} \bibinfo{year}{2002}\natexlab{}.
\newblock \showarticletitle{On Propagation of Deletions and Annotations Through Views}. In \bibinfo{booktitle}{\emph{{PODS}}}. \bibinfo{pages}{150--158}.
\newblock
\showISBNx{1-58113-507-6}
\urldef\tempurl%
\url{https://doi.org/10.1145/543613.543633}
\showDOI{\tempurl}


\bibitem[Callegaro et~al\mbox{.}(2013)]%
        {callegaro2013read}
\bibfield{author}{\bibinfo{person}{Vinicius Callegaro}, \bibinfo{person}{Mayler~GA Martins}, \bibinfo{person}{Renato~P Ribas}, {and} \bibinfo{person}{Andr{\'e}~I Reis}.} \bibinfo{year}{2013}\natexlab{}.
\newblock \showarticletitle{Read-polarity-once Boolean functions}. In \bibinfo{booktitle}{\emph{2013 26th Symposium on Integrated Circuits and Systems Design (SBCCI)}}. IEEE, \bibinfo{pages}{1--6}.
\newblock
\urldef\tempurl%
\url{https://doi.org/10.1109/SBCCI.2013.6644862}
\showDOI{\tempurl}


\bibitem[Chen et~al\mbox{.}(2024)]%
        {ICDE2024:HITSnDIFFs}
\bibfield{author}{\bibinfo{person}{Zixuan Chen}, \bibinfo{person}{Subhodeep Mitra}, \bibinfo{person}{R Ravi}, {and} \bibinfo{person}{Wolfgang Gatterbauer}.} \bibinfo{year}{2024}\natexlab{}.
\newblock \showarticletitle{HITSNDIFFS: From Truth Discovery to Ability Discovery by Recovering Matrices with the Consecutive Ones Property}. In \bibinfo{booktitle}{\emph{ICDE}}.
\newblock
\urldef\tempurl%
\url{https://arxiv.org/pdf/2401.00013}
\showURL{%
\tempurl}


\bibitem[Cheney et~al\mbox{.}(2009)]%
        {DBLP:journals/ftdb/CheneyCT09}
\bibfield{author}{\bibinfo{person}{James Cheney}, \bibinfo{person}{Laura Chiticariu}, {and} \bibinfo{person}{Wang~Chiew Tan}.} \bibinfo{year}{2009}\natexlab{}.
\newblock \showarticletitle{Provenance in Databases: Why, How, and Where}.
\newblock \bibinfo{journal}{\emph{Foundations and Trends in Databases}} \bibinfo{volume}{1}, \bibinfo{number}{4} (\bibinfo{year}{2009}), \bibinfo{pages}{379--474}.
\newblock
\urldef\tempurl%
\url{https://doi.org/10.1561/9781601982339}
\showDOI{\tempurl}


\bibitem[Conforti et~al\mbox{.}(2006)]%
        {conforti2006balanced}
\bibfield{author}{\bibinfo{person}{Michele Conforti}, \bibinfo{person}{G{\'e}rard Cornu{\'e}jols}, {and} \bibinfo{person}{Kristina Vu{\v{s}}kovi{\'c}}.} \bibinfo{year}{2006}\natexlab{}.
\newblock \showarticletitle{Balanced matrices}.
\newblock \bibinfo{journal}{\emph{Discrete Mathematics}} \bibinfo{volume}{306}, \bibinfo{number}{19-20} (\bibinfo{year}{2006}), \bibinfo{pages}{2411--2437}.
\newblock
\urldef\tempurl%
\url{https://doi.org/10.1016/j.disc.2005.12.033}
\showDOI{\tempurl}


\bibitem[Cornu{\'e}jols and Guenin(2002)]%
        {cornuejols2002ideal}
\bibfield{author}{\bibinfo{person}{G{\'e}rard Cornu{\'e}jols} {and} \bibinfo{person}{Bertrand Guenin}.} \bibinfo{year}{2002}\natexlab{}.
\newblock \showarticletitle{Ideal clutters}.
\newblock \bibinfo{journal}{\emph{Discrete Applied Mathematics}} \bibinfo{volume}{123}, \bibinfo{number}{1-3} (\bibinfo{year}{2002}), \bibinfo{pages}{303--338}.
\newblock
\urldef\tempurl%
\url{https://doi.org/10.1016/S0166-218X(01)00344-4}
\showDOI{\tempurl}


\bibitem[Crama and Hammer(2011)]%
        {CramaHammer2010:BooleanFunctions}
\bibfield{author}{\bibinfo{person}{Yves Crama} {and} \bibinfo{person}{Peter~L. Hammer}.} \bibinfo{year}{2011}\natexlab{}.
\newblock \bibinfo{booktitle}{\emph{Boolean Functions: Theory, Algorithms, and Applications}}.
\newblock \bibinfo{publisher}{Cambridge University Press}.
\newblock
\urldef\tempurl%
\url{https://doi.org/10.1017/cbo9780511852008.003}
\showDOI{\tempurl}


\bibitem[Cui et~al\mbox{.}(2000)]%
        {DBLP:journals/tods/CuiWW00}
\bibfield{author}{\bibinfo{person}{Yingwei Cui}, \bibinfo{person}{Jennifer Widom}, {and} \bibinfo{person}{Janet~L. Wiener}.} \bibinfo{year}{2000}\natexlab{}.
\newblock \showarticletitle{Tracing the lineage of view data in a warehousing environment}.
\newblock \bibinfo{journal}{\emph{ACM TODS}} \bibinfo{volume}{25}, \bibinfo{number}{2} (\bibinfo{year}{2000}), \bibinfo{pages}{179--227}.
\newblock
\urldef\tempurl%
\url{https://doi.org/10.1145/357775.357777}
\showDOI{\tempurl}


\bibitem[Dalvi and Suciu(2004)]%
        {DBLP:conf/vldb/DalviS04}
\bibfield{author}{\bibinfo{person}{Nilesh~N. Dalvi} {and} \bibinfo{person}{Dan Suciu}.} \bibinfo{year}{2004}\natexlab{}.
\newblock \showarticletitle{Efficient Query Evaluation on Probabilistic Databases}. In \bibinfo{booktitle}{\emph{VLDB}}. \bibinfo{pages}{864--875}.
\newblock
\urldef\tempurl%
\url{https://doi.org/10.1016/b978-012088469-8.50076-0}
\showDOI{\tempurl}


\bibitem[Dalvi and Suciu(2007)]%
        {DBLP:journals/vldb/DalviS07}
\bibfield{author}{\bibinfo{person}{Nilesh~N. Dalvi} {and} \bibinfo{person}{Dan Suciu}.} \bibinfo{year}{2007}\natexlab{}.
\newblock \showarticletitle{Efficient query evaluation on probabilistic databases}.
\newblock \bibinfo{journal}{\emph{{VLDB} J.}} \bibinfo{volume}{16}, \bibinfo{number}{4} (\bibinfo{year}{2007}), \bibinfo{pages}{523--544}.
\newblock
\urldef\tempurl%
\url{https://doi.org/10.1007/s00778-006-0004-3}
\showDOI{\tempurl}


\bibitem[Dalvi and Suciu(2012)]%
        {DBLP:journals/jacm/DalviS12}
\bibfield{author}{\bibinfo{person}{Nilesh~N. Dalvi} {and} \bibinfo{person}{Dan Suciu}.} \bibinfo{year}{2012}\natexlab{}.
\newblock \showarticletitle{The dichotomy of probabilistic inference for unions of conjunctive queries}.
\newblock \bibinfo{journal}{\emph{J. ACM}} \bibinfo{volume}{59}, \bibinfo{number}{6} (\bibinfo{year}{2012}), \bibinfo{pages}{30}.
\newblock
\urldef\tempurl%
\url{https://doi.org/10.1145/2395116.2395119}
\showDOI{\tempurl}


\bibitem[Dasgupta et~al\mbox{.}(2008)]%
        {Dasgupta:2008le}
\bibfield{author}{\bibinfo{person}{Sanjoy Dasgupta}, \bibinfo{person}{Christos~H Papadimitriou}, {and} \bibinfo{person}{Umesh~Virkumar Vazirani}.} \bibinfo{year}{2008}\natexlab{}.
\newblock \bibinfo{booktitle}{\emph{Algorithms}}.
\newblock \bibinfo{publisher}{McGraw-Hill Higher Education}, \bibinfo{address}{Boston}.
\newblock
\showISBNx{9780073523408 (acid-free paper)}
\urldef\tempurl%
\url{http://www.loc.gov/catdir/enhancements/fy0665/2006049014-t.html}
\showURL{%
\tempurl}


\bibitem[Dayal and Bernstein(1982)]%
        {Dayal82}
\bibfield{author}{\bibinfo{person}{Umeshwar Dayal} {and} \bibinfo{person}{Philip~A. Bernstein}.} \bibinfo{year}{1982}\natexlab{}.
\newblock \showarticletitle{On the Correct Translation of Update Operations on Relational Views}.
\newblock \bibinfo{journal}{\emph{ACM TODS}} \bibinfo{volume}{7}, \bibinfo{number}{3} (\bibinfo{year}{1982}), \bibinfo{pages}{381--416}.
\newblock
\showISSN{0362-5915}
\urldef\tempurl%
\url{https://doi.org/10.1145/319732.319740}
\showDOI{\tempurl}


\bibitem[Dechter(1999)]%
        {Dechter:1999:Bucket}
\bibfield{author}{\bibinfo{person}{Rina Dechter}.} \bibinfo{year}{1999}\natexlab{}.
\newblock \showarticletitle{Bucket elimination: A unifying framework for reasoning}.
\newblock \bibinfo{journal}{\emph{Artificial Intelligence}} \bibinfo{volume}{113}, \bibinfo{number}{1} (\bibinfo{year}{1999}), \bibinfo{pages}{41--85}.
\newblock
\showISSN{0004-3702}
\urldef\tempurl%
\url{https://doi.org/10.1016/S0004-3702(99)00059-4}
\showURL{%
\tempurl}


\bibitem[Dechter(2003)]%
        {Dechter:2003aa}
\bibfield{author}{\bibinfo{person}{Rina Dechter}.} \bibinfo{year}{2003}\natexlab{}.
\newblock \bibinfo{booktitle}{\emph{Constraint Processing}}.
\newblock \bibinfo{publisher}{Morgan Kaufmann}.
\newblock
\urldef\tempurl%
\url{https://doi.org/10.1016/b978-1-55860-890-0.x5000-2}
\showDOI{\tempurl}


\bibitem[den Heuvel et~al\mbox{.}(2019)]%
        {DBLP:conf/sigmod/HeuvelIGGT19}
\bibfield{author}{\bibinfo{person}{Maarten~Van den Heuvel}, \bibinfo{person}{Peter Ivanov}, \bibinfo{person}{Wolfgang Gatterbauer}, \bibinfo{person}{Floris Geerts}, {and} \bibinfo{person}{Martin Theobald}.} \bibinfo{year}{2019}\natexlab{}.
\newblock \showarticletitle{Anytime Approximation in Probabilistic Databases via Scaled Dissociations}. In \bibinfo{booktitle}{\emph{SIGMOD}}. \bibinfo{pages}{1295--1312}.
\newblock
\urldef\tempurl%
\url{https://doi.org/10.1145/3299869.3319900}
\showDOI{\tempurl}


\bibitem[Fink et~al\mbox{.}(2013)]%
        {DBLP:journals/vldb/FinkHO13}
\bibfield{author}{\bibinfo{person}{Robert Fink}, \bibinfo{person}{Jiewen Huang}, {and} \bibinfo{person}{Dan Olteanu}.} \bibinfo{year}{2013}\natexlab{}.
\newblock \showarticletitle{Anytime approximation in probabilistic databases}.
\newblock \bibinfo{journal}{\emph{{VLDB} J.}} \bibinfo{volume}{22}, \bibinfo{number}{6} (\bibinfo{year}{2013}), \bibinfo{pages}{823--848}.
\newblock
\urldef\tempurl%
\url{https://doi.org/10.1007/s00778-013-0310-5}
\showDOI{\tempurl}


\bibitem[Fink and Olteanu(2011)]%
        {DBLP:conf/icdt/FinkO11}
\bibfield{author}{\bibinfo{person}{Robert Fink} {and} \bibinfo{person}{Dan Olteanu}.} \bibinfo{year}{2011}\natexlab{}.
\newblock \showarticletitle{On the optimal approximation of queries using tractable propositional languages}. In \bibinfo{booktitle}{\emph{{ICDT}}}. \bibinfo{pages}{174--185}.
\newblock
\urldef\tempurl%
\url{https://doi.org/10.1145/1938551.1938575}
\showDOI{\tempurl}


\bibitem[Freire et~al\mbox{.}(2015)]%
        {FreireGIM15}
\bibfield{author}{\bibinfo{person}{Cibele Freire}, \bibinfo{person}{Wolfgang Gatterbauer}, \bibinfo{person}{Neil Immerman}, {and} \bibinfo{person}{Alexandra Meliou}.} \bibinfo{year}{2015}\natexlab{}.
\newblock \showarticletitle{The complexity of resilience and responsibility for self-join-free conjunctive queries}.
\newblock \bibinfo{journal}{\emph{PVLDB}} \bibinfo{volume}{9}, \bibinfo{number}{3} (\bibinfo{year}{2015}), \bibinfo{pages}{180--191}.
\newblock
\urldef\tempurl%
\url{https://doi.org/10.14778/2850583.2850592}
\showDOI{\tempurl}


\bibitem[Freire et~al\mbox{.}(2020)]%
        {DBLP:conf/pods/FreireGIM20}
\bibfield{author}{\bibinfo{person}{Cibele Freire}, \bibinfo{person}{Wolfgang Gatterbauer}, \bibinfo{person}{Neil Immerman}, {and} \bibinfo{person}{Alexandra Meliou}.} \bibinfo{year}{2020}\natexlab{}.
\newblock \showarticletitle{New Results for the Complexity of Resilience for Binary Conjunctive Queries with Self-Joins}. In \bibinfo{booktitle}{\emph{{PODS}}}. \bibinfo{pages}{271--284}.
\newblock
\urldef\tempurl%
\url{https://doi.org/10.1145/3375395.3387647}
\showDOI{\tempurl}


\bibitem[Fuhr and R{\"{o}}lleke(1997)]%
        {DBLP:journals/tois/FuhrR97}
\bibfield{author}{\bibinfo{person}{Norbert Fuhr} {and} \bibinfo{person}{Thomas R{\"{o}}lleke}.} \bibinfo{year}{1997}\natexlab{}.
\newblock \showarticletitle{A Probabilistic Relational Algebra for the Integration of Information Retrieval and Database Systems}.
\newblock \bibinfo{journal}{\emph{{ACM} Trans. Inf. Syst.}} \bibinfo{volume}{15}, \bibinfo{number}{1} (\bibinfo{year}{1997}), \bibinfo{pages}{32--66}.
\newblock
\urldef\tempurl%
\url{https://doi.org/10.1145/239041.239045}
\showDOI{\tempurl}


\bibitem[Garey and Johnson(1979)]%
        {garey1979computers}
\bibfield{author}{\bibinfo{person}{Michael~R Garey} {and} \bibinfo{person}{David~S Johnson}.} \bibinfo{year}{1979}\natexlab{}.
\newblock \bibinfo{booktitle}{\emph{Computers and intractability}}. Vol.~\bibinfo{volume}{174}.
\newblock \bibinfo{publisher}{W. H. Freeman \& Co.}
\newblock
\showISBNx{978-0-7167-1045-5}
\urldef\tempurl%
\url{https://dl.acm.org/doi/10.5555/578533}
\showURL{%
\tempurl}


\bibitem[Gatterbauer and Suciu(2014)]%
        {gatterbauer2014oblivious}
\bibfield{author}{\bibinfo{person}{Wolfgang Gatterbauer} {and} \bibinfo{person}{Dan Suciu}.} \bibinfo{year}{2014}\natexlab{}.
\newblock \showarticletitle{Oblivious bounds on the probability of {Boolean} functions}.
\newblock \bibinfo{journal}{\emph{TODS}} \bibinfo{volume}{39}, \bibinfo{number}{1} (\bibinfo{year}{2014}), \bibinfo{pages}{1--34}.
\newblock
\urldef\tempurl%
\url{https://doi.org/10.1145/2532641}
\showDOI{\tempurl}


\bibitem[Gatterbauer and Suciu(2017)]%
        {DBLP:journals/vldb/GatterbauerS17}
\bibfield{author}{\bibinfo{person}{Wolfgang Gatterbauer} {and} \bibinfo{person}{Dan Suciu}.} \bibinfo{year}{2017}\natexlab{}.
\newblock \showarticletitle{Dissociation and propagation for approximate lifted inference with standard relational database management systems}.
\newblock \bibinfo{journal}{\emph{{VLDB} J.}} \bibinfo{volume}{26}, \bibinfo{number}{1} (\bibinfo{year}{2017}), \bibinfo{pages}{5--30}.
\newblock
\urldef\tempurl%
\url{https://doi.org/10.1007/s00778-016-0434-5}
\showDOI{\tempurl}


\bibitem[Goldsmith et~al\mbox{.}(2008)]%
        {goldsmith2008complexity}
\bibfield{author}{\bibinfo{person}{Judy Goldsmith}, \bibinfo{person}{Matthias Hagen}, {and} \bibinfo{person}{Martin Mundhenk}.} \bibinfo{year}{2008}\natexlab{}.
\newblock \showarticletitle{Complexity of DNF minimization and isomorphism testing for monotone formulas}.
\newblock \bibinfo{journal}{\emph{Information and Computation}} \bibinfo{volume}{206}, \bibinfo{number}{6} (\bibinfo{year}{2008}), \bibinfo{pages}{760--775}.
\newblock
\urldef\tempurl%
\url{https://doi.org/10.1016/j.ic.2008.03.002}
\showDOI{\tempurl}


\bibitem[Golumbic and Gurvich(2010)]%
        {GolumbicGurvich2010:ReadOnceFunctions}
\bibfield{author}{\bibinfo{person}{Martin~Charles Golumbic} {and} \bibinfo{person}{Vladimir Gurvich}.} \bibinfo{year}{2010}\natexlab{}.
\newblock \bibinfo{booktitle}{\emph{Read-once functions}}.
\newblock \bibinfo{publisher}{Cambridge University Press}, Chapter~10.
\newblock
\urldef\tempurl%
\url{https://doi.org/10.1017/cbo9780511852008.011}
\showDOI{\tempurl}


\bibitem[Golumbic et~al\mbox{.}(2006)]%
        {DBLP:journals/dam/GolumbicMR06}
\bibfield{author}{\bibinfo{person}{Martin~Charles Golumbic}, \bibinfo{person}{Aviad Mintz}, {and} \bibinfo{person}{Udi Rotics}.} \bibinfo{year}{2006}\natexlab{}.
\newblock \showarticletitle{Factoring and recognition of read-once functions using cographs and normality and the readability of functions associated with partial k-trees}.
\newblock \bibinfo{journal}{\emph{Discrete Applied Mathematics}} \bibinfo{volume}{154}, \bibinfo{number}{10} (\bibinfo{year}{2006}), \bibinfo{pages}{1465--1477}.
\newblock
\urldef\tempurl%
\url{https://doi.org/10.1016/j.dam.2005.09.016}
\showURL{%
\tempurl}


\bibitem[Golumbic et~al\mbox{.}(2008)]%
        {golumbic2008improvement}
\bibfield{author}{\bibinfo{person}{Martin~Charles Golumbic}, \bibinfo{person}{Aviad Mintz}, {and} \bibinfo{person}{Udi Rotics}.} \bibinfo{year}{2008}\natexlab{}.
\newblock \showarticletitle{An improvement on the complexity of factoring read-once Boolean functions}.
\newblock \bibinfo{journal}{\emph{Discrete Applied Mathematics}} \bibinfo{volume}{156}, \bibinfo{number}{10} (\bibinfo{year}{2008}), \bibinfo{pages}{1633--1636}.
\newblock
\urldef\tempurl%
\url{https://doi.org/10.1016/j.dam.2008.02.011}
\showDOI{\tempurl}


\bibitem[Green et~al\mbox{.}(2007)]%
        {GKT07-semirings}
\bibfield{author}{\bibinfo{person}{Todd~J. Green}, \bibinfo{person}{Grigoris Karvounarakis}, {and} \bibinfo{person}{Val Tannen}.} \bibinfo{year}{2007}\natexlab{}.
\newblock \showarticletitle{Provenance semirings}. In \bibinfo{booktitle}{\emph{PODS}}. \bibinfo{pages}{31--40}.
\newblock
\urldef\tempurl%
\url{https://doi.org/10.1145/1265530.1265535}
\showDOI{\tempurl}


\bibitem[Green and Tannen(2017)]%
        {Green:2017:SFD:3034786.3056125}
\bibfield{author}{\bibinfo{person}{Todd~J. Green} {and} \bibinfo{person}{Val Tannen}.} \bibinfo{year}{2017}\natexlab{}.
\newblock \showarticletitle{The Semiring Framework for Database Provenance}. In \bibinfo{booktitle}{\emph{PODS}}. \bibinfo{pages}{93--99}.
\newblock
\showISBNx{978-1-4503-4198-1}
\urldef\tempurl%
\url{https://doi.org/10.1145/3034786.3056125}
\showDOI{\tempurl}


\bibitem[Gurobi~Optimization(2021a)]%
        {gurobi}
\bibfield{author}{\bibinfo{person}{LLC Gurobi~Optimization}.} \bibinfo{year}{2021}\natexlab{a}.
\newblock \bibinfo{title}{Gurobi Optimizer Reference Manual}.
\newblock
\newblock
\urldef\tempurl%
\url{http://www.gurobi.com}
\showURL{%
\tempurl}


\bibitem[Gurobi~Optimization(2021b)]%
        {gurobi_working}
\bibfield{author}{\bibinfo{person}{LLC Gurobi~Optimization}.} \bibinfo{year}{2021}\natexlab{b}.
\newblock \bibinfo{title}{Mixed-Integer Programming (MIP) -- A Primer on the Basics}.
\newblock
\newblock
\urldef\tempurl%
\url{https://www.gurobi.com/resource/mip-basics/}
\showURL{%
\tempurl}


\bibitem[Gurvich(1977)]%
        {gurvich77}
\bibfield{author}{\bibinfo{person}{V.A. Gurvich}.} \bibinfo{year}{1977}\natexlab{}.
\newblock \showarticletitle{Repetition-free {B}oolean functions}.
\newblock \bibinfo{journal}{\emph{Uspekhi Mat. Nauk}}  \bibinfo{volume}{32} (\bibinfo{year}{1977}), \bibinfo{pages}{183--184}.
\newblock
\urldef\tempurl%
\url{http://mi.mathnet.ru/umn3055}
\showURL{%
\tempurl}
\newblock
\shownote{(in Russian)}.


\bibitem[Hammer and Kogan(1993)]%
        {hammer1993optimal}
\bibfield{author}{\bibinfo{person}{Peter~L Hammer} {and} \bibinfo{person}{Alexander Kogan}.} \bibinfo{year}{1993}\natexlab{}.
\newblock \showarticletitle{Optimal compression of propositional Horn knowledge bases: complexity and approximation}.
\newblock \bibinfo{journal}{\emph{Artificial Intelligence}} \bibinfo{volume}{64}, \bibinfo{number}{1} (\bibinfo{year}{1993}), \bibinfo{pages}{131--145}.
\newblock
\urldef\tempurl%
\url{https://doi.org/10.1016/0004-3702(93)90062-G}
\showDOI{\tempurl}


\bibitem[Heller and Tompkins(1957)]%
        {14AnExtensionofaTheoremofDantzigs}
\bibfield{author}{\bibinfo{person}{I. Heller} {and} \bibinfo{person}{C.~B. Tompkins}.} \bibinfo{year}{31 Dec. 1957}\natexlab{}.
\newblock \bibinfo{booktitle}{\emph{14 . An Extension of a Theorem of Dantzig's}}.
\newblock \bibinfo{publisher}{Princeton University Press}, \bibinfo{address}{Princeton}, \bibinfo{pages}{247 -- 254}.
\newblock
\showISBNx{9781400881987}
\urldef\tempurl%
\url{https://doi.org/10.1515/9781400881987-015}
\showDOI{\tempurl}


\bibitem[Hemaspaandra and Schnoor(2011)]%
        {hemaspaandra2011minimization}
\bibfield{author}{\bibinfo{person}{Edith Hemaspaandra} {and} \bibinfo{person}{Henning Schnoor}.} \bibinfo{year}{2011}\natexlab{}.
\newblock \showarticletitle{Minimization for generalized boolean formulas}. In \bibinfo{booktitle}{\emph{21st International Joint Conference on Artificial Intelligence (IJCAI)}}. \bibinfo{pages}{566--571}.
\newblock
\urldef\tempurl%
\url{https://doi.org/10.1109/sfcs.1997.646147}
\showURL{%
\tempurl}


\bibitem[Ilango(2022)]%
        {ilango2022minimum}
\bibfield{author}{\bibinfo{person}{Rahul Ilango}.} \bibinfo{year}{2022}\natexlab{}.
\newblock \showarticletitle{The Minimum Formula Size Problem is (ETH) Hard}. In \bibinfo{booktitle}{\emph{2021 IEEE 62nd Annual Symposium on Foundations of Computer Science (FOCS)}}. IEEE, \bibinfo{pages}{427--432}.
\newblock
\urldef\tempurl%
\url{https://doi.org/10.1109/FOCS52979.2021.00050}
\showDOI{\tempurl}


\bibitem[Lau et~al\mbox{.}(2011)]%
        {lau2011iterative}
\bibfield{author}{\bibinfo{person}{Lap~Chi Lau}, \bibinfo{person}{Ramamoorthi Ravi}, {and} \bibinfo{person}{Mohit Singh}.} \bibinfo{year}{2011}\natexlab{}.
\newblock \bibinfo{booktitle}{\emph{Iterative methods in combinatorial optimization}}. Vol.~\bibinfo{volume}{46}.
\newblock \bibinfo{publisher}{Cambridge University Press}.
\newblock
\urldef\tempurl%
\url{https://doi.org/10.1017/cbo9780511977152.002}
\showDOI{\tempurl}


\bibitem[Lauritzen and Spiegelhalter(1988)]%
        {10.2307/2345762}
\bibfield{author}{\bibinfo{person}{S.~L. Lauritzen} {and} \bibinfo{person}{D.~J. Spiegelhalter}.} \bibinfo{year}{1988}\natexlab{}.
\newblock \showarticletitle{Local Computations with Probabilities on Graphical Structures and Their Application to Expert Systems}.
\newblock \bibinfo{journal}{\emph{Journal of the Royal Statistical Society. Series B (Methodological)}} \bibinfo{volume}{50}, \bibinfo{number}{2} (\bibinfo{year}{1988}), \bibinfo{pages}{157--224}.
\newblock
\showISSN{00359246}
\urldef\tempurl%
\url{http://www.jstor.org/stable/2345762}
\showURL{%
\tempurl}


\bibitem[Makhija and Gatterbauer(2023)]%
        {MakhijaG:2024}
\bibfield{author}{\bibinfo{person}{Neha Makhija} {and} \bibinfo{person}{Wolfgang Gatterbauer}.} \bibinfo{year}{2023}\natexlab{}.
\newblock \showarticletitle{A Unified Approach for Resilience and Causal Responsibility with Integer Linear Programming (ILP) and LP Relaxations}.
\newblock \bibinfo{journal}{\emph{PACMMOD}} \bibinfo{volume}{1}, \bibinfo{number}{4} (\bibinfo{date}{dec} \bibinfo{year}{2023}), \bibinfo{pages}{228:1--228:27}.
\newblock
\urldef\tempurl%
\url{https://doi.org/10.1145/3626715}
\showURL{%
\tempurl}


\bibitem[Martins et~al\mbox{.}(2010)]%
        {martins2010boolean}
\bibfield{author}{\bibinfo{person}{Mayler~GA Martins}, \bibinfo{person}{Leomar Rosa}, \bibinfo{person}{Anders~B Rasmussen}, \bibinfo{person}{Renato~P Ribas}, {and} \bibinfo{person}{Andre~I Reis}.} \bibinfo{year}{2010}\natexlab{}.
\newblock \showarticletitle{Boolean factoring with multi-objective goals}. In \bibinfo{booktitle}{\emph{International Conference on Computer Design (ICCD)}}. IEEE, \bibinfo{pages}{229--234}.
\newblock
\urldef\tempurl%
\url{https://doi.org/10.1109/ICCD.2010.5647772}
\showDOI{\tempurl}


\bibitem[Meliou et~al\mbox{.}(2011)]%
        {DBLP:journals/pvldb/MeliouGS11}
\bibfield{author}{\bibinfo{person}{Alexandra Meliou}, \bibinfo{person}{Wolfgang Gatterbauer}, {and} \bibinfo{person}{Dan Suciu}.} \bibinfo{year}{2011}\natexlab{}.
\newblock \showarticletitle{Reverse Data Management}.
\newblock \bibinfo{journal}{\emph{{PVLDB}}} \bibinfo{volume}{4}, \bibinfo{number}{12} (\bibinfo{year}{2011}), \bibinfo{pages}{1490--1493}.
\newblock
\urldef\tempurl%
\url{https://doi.org/10.14778/3402755.3402803}
\showURL{%
\tempurl}


\bibitem[Mintz and Golumbic(2005)]%
        {mintz2005factoring}
\bibfield{author}{\bibinfo{person}{Aviad Mintz} {and} \bibinfo{person}{Martin~Charles Golumbic}.} \bibinfo{year}{2005}\natexlab{}.
\newblock \showarticletitle{Factoring Boolean functions using graph partitioning}.
\newblock \bibinfo{journal}{\emph{Discrete Applied Mathematics}} \bibinfo{volume}{149}, \bibinfo{number}{1-3} (\bibinfo{year}{2005}), \bibinfo{pages}{131--153}.
\newblock
\urldef\tempurl%
\url{https://doi.org/10.1016/j.dam.2005.02.007}
\showDOI{\tempurl}


\bibitem[Moerkotte(2009)]%
        {Moerkotte:BuildingQueryCompilers}
\bibfield{author}{\bibinfo{person}{Guido Moerkotte}.} \bibinfo{year}{Sept. 2009}\natexlab{}.
\newblock \bibinfo{title}{Building Query Compilers}.
\newblock \bibinfo{howpublished}{Draft version 03.03.09}.
\newblock
\urldef\tempurl%
\url{http://theo.cs.ovgu.de/lehre/lehre09w/Anfrageoptimierung/querycompiler.pdf}
\showURL{%
\tempurl}


\bibitem[OEIS(2022)]%
        {oeis}
OEIS \bibinfo{year}{2022}\natexlab{}.
\newblock \bibinfo{title}{{OEIS A000169}: The {O}n-{L}ine {E}ncyclopedia of {I}nteger {S}equences}.
\newblock
\newblock
\urldef\tempurl%
\url{https://oeis.org/A000169}
\showURL{%
\tempurl}


\bibitem[Olteanu and Huang(2008)]%
        {DBLP:conf/sum/OlteanuH08}
\bibfield{author}{\bibinfo{person}{Dan Olteanu} {and} \bibinfo{person}{Jiewen Huang}.} \bibinfo{year}{2008}\natexlab{}.
\newblock \showarticletitle{Using {OBDDs} for Efficient Query Evaluation on Probabilistic Databases}. In \bibinfo{booktitle}{\emph{{SUM}}}. \bibinfo{pages}{326--340}.
\newblock
\urldef\tempurl%
\url{https://doi.org/10.1007/978-3-540-87993-0_26}
\showDOI{\tempurl}


\bibitem[Olteanu et~al\mbox{.}(2010)]%
        {DBLP:conf/icde/OlteanuHK10}
\bibfield{author}{\bibinfo{person}{Dan Olteanu}, \bibinfo{person}{Jiewen Huang}, {and} \bibinfo{person}{Christoph Koch}.} \bibinfo{year}{2010}\natexlab{}.
\newblock \showarticletitle{Approximate confidence computation in probabilistic databases}. In \bibinfo{booktitle}{\emph{{ICDE}}}. \bibinfo{pages}{145--156}.
\newblock
\urldef\tempurl%
\url{https://doi.org/10.1109/icde.2010.5447826}
\showDOI{\tempurl}


\bibitem[Olteanu and Schleich(2016)]%
        {DBLP:journals/sigmod/OlteanuS16}
\bibfield{author}{\bibinfo{person}{Dan Olteanu} {and} \bibinfo{person}{Maximilian Schleich}.} \bibinfo{year}{2016}\natexlab{}.
\newblock \showarticletitle{Factorized Databases}.
\newblock \bibinfo{journal}{\emph{{SIGMOD} Rec.}} \bibinfo{volume}{45}, \bibinfo{number}{2} (\bibinfo{year}{2016}), \bibinfo{pages}{5--16}.
\newblock
\urldef\tempurl%
\url{https://doi.org/10.1145/3003665.3003667}
\showDOI{\tempurl}


\bibitem[Olteanu and Z{\'{a}}vodn{\'{y}}(2011)]%
        {DBLP:conf/tapp/Zavodny11}
\bibfield{author}{\bibinfo{person}{Dan Olteanu} {and} \bibinfo{person}{Jakub Z{\'{a}}vodn{\'{y}}}.} \bibinfo{year}{2011}\natexlab{}.
\newblock \showarticletitle{On Factorisation of Provenance Polynomials}. In \bibinfo{booktitle}{\emph{3rd Workshop on the Theory and Practice of Provenance (TaPP'11)}}. \bibinfo{publisher}{{USENIX}}.
\newblock
\urldef\tempurl%
\url{https://www.usenix.org/conference/tapp11/factorisation-provenance-polynomials}
\showURL{%
\tempurl}


\bibitem[Olteanu and Z{\'a}vodn{\`y}(2012)]%
        {olteanu2012factorised}
\bibfield{author}{\bibinfo{person}{Dan Olteanu} {and} \bibinfo{person}{Jakub Z{\'a}vodn{\`y}}.} \bibinfo{year}{2012}\natexlab{}.
\newblock \showarticletitle{Factorised representations of query results: size bounds and readability}. In \bibinfo{booktitle}{\emph{Proceedings of the 15th International Conference on Database Theory}}. \bibinfo{pages}{285--298}.
\newblock
\urldef\tempurl%
\url{https://doi.org/10.1145/2274576.2274607}
\showDOI{\tempurl}


\bibitem[Olteanu and Z{\'{a}}vodn{\'{y}}(2015)]%
        {DBLP:journals/tods/OlteanuZ15}
\bibfield{author}{\bibinfo{person}{Dan Olteanu} {and} \bibinfo{person}{Jakub Z{\'{a}}vodn{\'{y}}}.} \bibinfo{year}{2015}\natexlab{}.
\newblock \showarticletitle{Size Bounds for Factorised Representations of Query Results}.
\newblock \bibinfo{journal}{\emph{{ACM} Trans. Database Syst.}} \bibinfo{volume}{40}, \bibinfo{number}{1} (\bibinfo{year}{2015}), \bibinfo{pages}{2:1--2:44}.
\newblock
\urldef\tempurl%
\url{https://doi.org/10.1145/2656335}
\showDOI{\tempurl}


\bibitem[Pipatsrisawat and Darwiche(2008)]%
        {pipatsrisawat2008new}
\bibfield{author}{\bibinfo{person}{Knot Pipatsrisawat} {and} \bibinfo{person}{Adnan Darwiche}.} \bibinfo{year}{2008}\natexlab{}.
\newblock \showarticletitle{New Compilation Languages Based on Structured Decomposability}. In \bibinfo{booktitle}{\emph{Proceedings of the 23rd National Conference on Artificial Intelligence - Volume 1}} (Chicago, Illinois) \emph{(\bibinfo{series}{AAAI'08})}. \bibinfo{publisher}{AAAI Press}, \bibinfo{pages}{517--522}.
\newblock
\showISBNx{9781577353683}
\urldef\tempurl%
\url{https://dl.acm.org/doi/10.5555/1619995.1620079}
\showURL{%
\tempurl}


\bibitem[R\'e et~al\mbox{.}(2007)]%
        {re2007efficient}
\bibfield{author}{\bibinfo{person}{Christopher R\'e}, \bibinfo{person}{Nilesh Dalvi}, {and} \bibinfo{person}{Dan Suciu}.} \bibinfo{year}{2007}\natexlab{}.
\newblock \showarticletitle{Efficient top-k query evaluation on probabilistic data}. In \bibinfo{booktitle}{\emph{{ICDE}}}. \bibinfo{pages}{886--895}.
\newblock
\urldef\tempurl%
\url{https://doi.org/10.1109/icde.2007.367934}
\showDOI{\tempurl}


\bibitem[Research(2022)]%
        {Mathematica}
\bibfield{author}{\bibinfo{person}{Wolfram Research}.} \bibinfo{year}{2022}\natexlab{}.
\newblock \bibinfo{title}{Mathematica, {V}ersion 13.1}.
\newblock
\newblock
\urldef\tempurl%
\url{https://www.wolfram.com/mathematica}
\showURL{%
\tempurl}


\bibitem[Robertson and Seymour(1986)]%
        {DBLP:journals/jal/RobertsonS86}
\bibfield{author}{\bibinfo{person}{Neil Robertson} {and} \bibinfo{person}{Paul~D. Seymour}.} \bibinfo{year}{1986}\natexlab{}.
\newblock \showarticletitle{Graph Minors. {II.} Algorithmic Aspects of Tree-Width}.
\newblock \bibinfo{journal}{\emph{J. Algorithms}} \bibinfo{volume}{7}, \bibinfo{number}{3} (\bibinfo{year}{1986}), \bibinfo{pages}{309--322}.
\newblock
\urldef\tempurl%
\url{https://doi.org/10.1016/0196-6774(86)90023-4}
\showDOI{\tempurl}


\bibitem[Roy et~al\mbox{.}(2011)]%
        {DBLP:conf/icdt/RoyPT11}
\bibfield{author}{\bibinfo{person}{Sudeepa Roy}, \bibinfo{person}{Vittorio Perduca}, {and} \bibinfo{person}{Val Tannen}.} \bibinfo{year}{2011}\natexlab{}.
\newblock \showarticletitle{Faster query answering in probabilistic databases using read-once functions}. In \bibinfo{booktitle}{\emph{ICDT}}. \bibinfo{pages}{232--243}.
\newblock
\urldef\tempurl%
\url{https://doi.org/10.1145/1938551.1938582}
\showDOI{\tempurl}


\bibitem[Schrijver(1998)]%
        {schrijver1998theory}
\bibfield{author}{\bibinfo{person}{Alexander Schrijver}.} \bibinfo{year}{1998}\natexlab{}.
\newblock \bibinfo{booktitle}{\emph{Theory of linear and integer programming}}.
\newblock \bibinfo{publisher}{John Wiley \& Sons}.
\newblock
\urldef\tempurl%
\url{https://doi.org/10.1137/1030065}
\showDOI{\tempurl}


\bibitem[Sen et~al\mbox{.}(2010)]%
        {SenDeshpandeGetoor2010:ReadOnce}
\bibfield{author}{\bibinfo{person}{Prithviraj Sen}, \bibinfo{person}{Amol Deshpande}, {and} \bibinfo{person}{Lise Getoor}.} \bibinfo{year}{2010}\natexlab{}.
\newblock \showarticletitle{Read-Once Functions and Query Evaluation in Probabilistic Databases}.
\newblock \bibinfo{journal}{\emph{PVLDB}} \bibinfo{volume}{3}, \bibinfo{number}{1} (\bibinfo{year}{2010}), \bibinfo{pages}{1068--1079}.
\newblock
\urldef\tempurl%
\url{https://doi.org/10.14778/1920841.1920975}
\showDOI{\tempurl}


\bibitem[Suciu et~al\mbox{.}(2011)]%
        {DBLP:series/synthesis/2011Suciu}
\bibfield{author}{\bibinfo{person}{Dan Suciu}, \bibinfo{person}{Dan Olteanu}, \bibinfo{person}{Christopher R{\'{e}}}, {and} \bibinfo{person}{Christoph Koch}.} \bibinfo{year}{2011}\natexlab{}.
\newblock \bibinfo{booktitle}{\emph{Probabilistic Databases}}.
\newblock \bibinfo{publisher}{Morgan {\&} Claypool}.
\newblock
\urldef\tempurl%
\url{https://doi.org/10.2200/s00362ed1v01y201105dtm016}
\showDOI{\tempurl}


\bibitem[Tarjan and Yannakakis(1984)]%
        {DBLP:journals/siamcomp/TarjanY84}
\bibfield{author}{\bibinfo{person}{Robert~Endre Tarjan} {and} \bibinfo{person}{Mihalis Yannakakis}.} \bibinfo{year}{1984}\natexlab{}.
\newblock \showarticletitle{Simple Linear-Time Algorithms to Test Chordality of Graphs, Test Acyclicity of Hypergraphs, and Selectively Reduce Acyclic Hypergraphs}.
\newblock \bibinfo{journal}{\emph{{SIAM} J. Comput.}} \bibinfo{volume}{13}, \bibinfo{number}{3} (\bibinfo{year}{1984}), \bibinfo{pages}{566--579}.
\newblock
\urldef\tempurl%
\url{https://doi.org/10.1137/0213035}
\showDOI{\tempurl}


\bibitem[Umans(2001)]%
        {umans2001minimum}
\bibfield{author}{\bibinfo{person}{Christopher Umans}.} \bibinfo{year}{2001}\natexlab{}.
\newblock \showarticletitle{The minimum equivalent DNF problem and shortest implicants}.
\newblock \bibinfo{journal}{\emph{J. Comput. System Sci.}} \bibinfo{volume}{63}, \bibinfo{number}{4} (\bibinfo{year}{2001}), \bibinfo{pages}{597--611}.
\newblock
\urldef\tempurl%
\url{https://doi.org/10.1109/sfcs.1998.743506}
\showDOI{\tempurl}


\bibitem[Vardi(1982)]%
        {DBLP:conf/stoc/Vardi82}
\bibfield{author}{\bibinfo{person}{Moshe~Y. Vardi}.} \bibinfo{year}{1982}\natexlab{}.
\newblock \showarticletitle{The Complexity of Relational Query Languages (Extended Abstract)}. In \bibinfo{booktitle}{\emph{STOC}}. \bibinfo{pages}{137--146}.
\newblock
\urldef\tempurl%
\url{https://doi.org/10.1145/800070.802186}
\showDOI{\tempurl}


\bibitem[Vazirani(2001)]%
        {vazirani2001approximation}
\bibfield{author}{\bibinfo{person}{Vijay~V Vazirani}.} \bibinfo{year}{2001}\natexlab{}.
\newblock \bibinfo{booktitle}{\emph{Approximation algorithms}}. Vol.~\bibinfo{volume}{1}.
\newblock \bibinfo{publisher}{Springer}.
\newblock
\showISBNx{978-3-662-04565-7}
\urldef\tempurl%
\url{https://dl.acm.org/doi/10.5555/500776}
\showURL{%
\tempurl}


\bibitem[Williamson(2019)]%
        {williamson2019network}
\bibfield{author}{\bibinfo{person}{David~P Williamson}.} \bibinfo{year}{2019}\natexlab{}.
\newblock \bibinfo{booktitle}{\emph{Network flow algorithms}}.
\newblock \bibinfo{publisher}{Cambridge University Press}.
\newblock
\urldef\tempurl%
\url{https://doi.org/10.1017/9781316888568}
\showDOI{\tempurl}


\end{thebibliography}

\appendix

\clearpage
\appendix

\section{Nomenclature and Conventions}\label{sec:sec:appendix:nomenclature}

\begin{table}[h]
\centering
\small
\begin{tabularx}{\linewidth}{@{\hspace{0pt}} >{$}l<{$} @{\hspace{2mm}}X@{}} %
\hline
\textrm{Symbol}		& Definition 	\\
\hline
    \hline
	Q			& a self-join free Boolean CQ	\\
	R, S, T, U		& relational tables \\
	r_i, s_i, t_i, u_i & tuple identifiers \\
	x, y, z		& query variables \\	
	m			& number of atoms in a query \\
	N			& size of database $|D|$ \\
	\varphi, \psi	& propositional formulas / expressions   	\\
	\var(X)		& the set of variables in atom / relation / formula $X$\\
	\at(x_j)	& set of atoms that contain variable $x_j$	\\
	W			& set of witnesses $W = \witnesses(Q,D)$\\
	\vec w  	& witness \\
	\veo(Q) & set of all legal $\veo$s for Q\\
	\mveo(Q)		& set of minimal $\veo$s for Q\\
	k=|\mveo(Q)|	& number of minimal $\veo$s \\	
	v\angle{\w}	& a $\veo$ instance of $\veo$ $v$ over witness $\w$\\
	\Var(g_i) 	& set of variables of a query $q$ or atom $g_i$ \\
	P			& query plan	\\
	\mathcal{P}	& set of plans	\\	
	F			& flow graph	\\
	\join{}{\ldots}		& provenance join operator in prefix notation	\\
	\projd{\vec x},\proj{\vec y}	& provenance project operators: onto $\vec x$, or project $\vec y$ away \\
	\vec x		& unordered set or ordered tuple \\
	\vec a/ \vec x		& substitute values $\vec a$ for variables $\vec x$ \\
	$Q[\vec x]$ & indicates that $\vec x$ represents the set of all existentially quantified variables for Boolean query $Q$ \\
	\texttt{len}			& Length of a Factorization \\
	\QPV 			& Query Plan Variables of an ILP \\
	\PV 			& Prefix Variables of an ILP \\
	q[\hdots] 				& a ILP decision query plan variable \\
	p[\hdots] 				& a ILP decision prefix variable \\
	$c$ 			& weight (or cost) of variables in the ILP / nodes in the Factorization Flow Graph \\
	\Omega 			& An Ordering of $\mveo$ chosen for MFMC based algorithm\\
	(v_1,v_2,\hdots,v_k) & An ordered list of $\veo$s, $\veoff$s or any other set of objects \\
\hline

\end{tabularx}
\end{table}

\begin{table}[h]
	\small
	\centering
	\begin{tabularx}{\linewidth}{@{\hspace{0pt}} >{$}l<{$} @{\hspace{2mm}}X@{}} %
	\hline
	\textrm{Query}		& Definition 	\\
		\hline
		\hline
		\qtwochain 		& 2-chain query $R(x,y), S(y,z)$ \\
		\qthreechain 	& 3-chain query $R(x,y), S(y,z), T(z,u)$ \\
		\qfourchain 	& 4-chain query $P(u,x),R(x,y), S(y,z), T(z,v)$ \\
		\qfivechain 	& 5-chain query $L(a,u),P(u,x),R(x,y), S(y,z), T(z,v)$ \\
		\qtwostar 		& 2-star query $R(x)S(y),W(x,y)$ \\
		\qthreestar 	& 3-star query $R(x)S(y),T(z)W(x,y,z)$ \\
		\qtriangle 		& Triangle query $R(x,y)S(y,z),T(z,x)$ \\
		\qtriangleunary & Triangle-unary query $U(x)R(x,y)S(y,z),T(z,x)$ \\
		\qsixcyclewe & 6-cycle query with end points $A(x),$ $R(x,y),$ $B(y),$ $S(y,z),$ $C(z),$ $T(z,u)$ $D(u),$ $U(u,v),$ $E(v),$ $V(v,w),$ $F(w),$ $W(w,x)$ \\
		\qdominatedtriangle & Co-dominated triangle query $A(w),$ $R(w,x,y),$ $S(w, y,z),$ $T(w,z,x)$ \\
		\hline
		\end{tabularx}
	\end{table}

	\begin{table}
		\small
		\centering	
		\textrm{\textbf{Additional Nomenclature used in the Appendix}}\\
		\begin{tabularx}{\linewidth}{@{\hspace{0pt}} >{$}l<{$} @{\hspace{2mm}}X@{}}
		\hline
		\textrm{Symbol}		& Definition 	\\
		\hline
		\hline
			\HVar(P) 	& set of head variables of a query $q$ or a plan $P$ \\	
			\EVar(q) 	& set of existential variables: $\EVar(q) \!=\! \Var(q) \!-\! \HVar(q)$ \\		
			\PP{\phi}	& probability of a Boolean expression		\\
			\Delta		& collection of sets of variables $\Delta = (\vec y_1, \ldots, \vec y_m)$ \\
			R_i^{\vec y_i} & dissociated relation $R_i(\vec x_i)$ on variables $\vec y_i$: $R_i(\vec x_i, \vec y_i)$ \\
			q^{\Delta}	& dissociated query\\
			\join{}{\ldots}		& provenance join operator in prefix notation	\\
			\projd{\vec x},\proj{\vec y}	& provenance project operators: onto $\vec x$, or project $\vec y$ away \\
		\hline
		\end{tabularx}
	\end{table}

We write $[k]$ as short notation for the set $\{1, \ldots, k\}$
and use boldface to denote tuples or ordered sets, 
(e.g., $\vec x = (x_1, \ldots, x_\ell)$).
We fix a relational vocabulary $\vec R = (R_1, \ldots, R_m)$, 
and denote with $\arity(R_i)$ the number of attributes of a relation $R_i$. 
For notational convenience, we assume w.l.o.g. that there are no two atoms $R_i$ and $R_j$ with $\var(R_i)=\var(R_j)$.
A database instance over $\vec R$ is $D = (R_1^D, \ldots, R_m^D)$, where each $R_i^D$ is a finite relation.
We call the elements of $R_i^D$ tuples and
write $R_i$ instead of $R_i^D$ when $D$ is clear from the context.
With some abuse of notation 
we also denote $D$ as the set of all tuples, 
i.e.\ $D = \bigcup_i R_i$.
The active domain $\dom(D)$ is the set of all constants occurring in $D$.
W.l.o.g., we commonly use $\dom(D) \subset \N \cup \{a, b, \ldots, z\}$.
The size of the database instance is $n = |D|$, i.e.\ the number of tuples in the database.\footnote{Notice that other work sometimes uses $|\dom(D)|$ as the size of the database.
Our different definition has no implication on our complexity results but simplifies the discussions of our reductions.
}

\section{Additional details on \cref{SEC:RELATEDWORK}: Related Work}

\subsection{Boolean Factorization}
\label{sec:appendix:rw:booleanfactorization}

\cref{Fig:RelatedWorkBooleanFactorization} illustrates the landscape of known results for the problem of Minimum Equivalent Expressions (MEE) applied to formulas.

The general problem of MEE has been long known to be \np-hard \cite{garey1979computers}.
However, only relatively recently it has been proved to be $\Sigma^2_p$-complete~\cite{buchfuhrer2011complexity}.
Various important classes of this problem have been studied, a fundamental one being the factorization of DNF expressions. 
The MinDNF problem \cite{umans2001minimum}, deals with finding the minimum equivalent DNF expression of an input DNF formula, and is also known to be $\Sigma^2_p$-complete.
However, if the input to the MinDNF is the truth table (or set of all true assignments of the formula) then the problem is $\npc$ \cite{Allender:MinimizingDNF:2008}. 
If we take away the restriction that the factorized formula must be a DNF, then the problem of finding the minimum factorization of an input table is known as the Minimum Formula Size Problem (MFSP) and is shown to be in $\np$ and (ETH)-hard \cite{ilango2022minimum}. 

Another important class of restrictions is over monotone formulas (thus we do not allow negatives in input or output formulas). 
Surprisingly, we do not know of any work that proves the complexity of the general monotone boolean factorization problem.
However, there are many interesting and important restrictions for which complexity results are known.
One such important sub-class is that of read-once formulas, which can be factorized in $\PTIME$ \cite{golumbic2008improvement}.
For Monotone formulas with DNF input and output restrictions, the problem can be solved in logspace by eliminating monomials \cite{goldsmith2008complexity}.
Interestingly the problem monotone formula factorization of an arbitrary formula with a DNF restriction on the output only has differing complexity based on the input encoding of the length of the factorization. 
Checking if the minimum size of a DNF for a monotone
formula is at most k is PP-complete, but for k in unary, the complexity of the problem drops to coNP \cite{goldsmith2008complexity}.
The intuition is that in this problem, (which can be seen as ``dual'' of $\minfact$ since it has a DNF output restriction instead of a DNF input restriction), the optimal output (a DNF) can be exponentially larger than the input (any monotone formula). 

Our problem of $\minfact$ is a further restriction on the MEE problem applied to a monotone DNF. 
Provenance formulas for sj-free CQs are $k$-partite monotone formulas that satisfy join dependencies.
We prove in this paper that the problem is $\npc$, in general, and further identify interesting $\PTIME$ subcases.

\tikzset{
  smallnode/.style={font=\fontsize{8}{10}\selectfont},
  largenode/.style={font=\fontsize{8}{10}\selectfont},
}
\definecolor{darkgreen}{RGB}{34,139,34}
\begin{figure*}
	\centering
	\begin{tikzpicture}[scale=.9]
		\draw[rounded corners=20pt] (0,0) rectangle (16,9);
		\draw[rounded corners=20pt, fill=cyan, draw=black, fill opacity=0.1] (3,0.1) rectangle (9,8);
		\draw[rounded corners=20pt, fill=orange, draw=black, fill opacity=0.1] (6,0.2) rectangle (15.5,8.1);
		\draw[rounded corners=20pt, fill=teal, draw=black, fill opacity=0.2] (6.5,4.5) rectangle (14.5,6.5);
		\draw[rounded corners=20pt, fill=pink, draw=black, fill opacity=0.2] (0.5,0.5) rectangle (14.5,4);
		\draw[rounded corners=20pt, fill=purple, draw=black, fill opacity=0.2] (9.4,1) rectangle (12,2.8);
		\draw[rounded corners=20pt, fill=white, draw=black, fill opacity=0.6, line width=1.5pt] (11.3,1.5) rectangle (14.3,3.5);

		\node[largenode, anchor=north west] at (0.1,5) {\textbf{\textcolor{orange}{$\Sigma^p_2$-complete}}~\cite{buchfuhrer2011complexity} };
		\node[largenode] at (7.5,7) {\textbf{\textcolor{orange}{$\Sigma^p_2$-complete}}~\cite{umans2001minimum}};
		\node[largenode] at (7.8,5.5) {\textbf{\textcolor{blue}{\npc}}~\cite{Allender:MinimizingDNF:2008}};
		\node[largenode] at (11.5, 5.5) {\textbf{\textcolor{blue}{In NP, (ETH) Hard}}~\cite{{ilango2022minimum}}};
		\node[largenode, align=center] at (12.8, 2.5) {\textbf{\textcolor{blue}{\npc}}\\($\minfact$)};
		\node[largenode] at (10.4, 2) {\textbf{\textcolor{darkgreen}{$\PTIME$}}~\cite{golumbic2008improvement}};
		\node[largenode] at (7.5, 2) {\textbf{\textcolor{olive}{$L$}}~\cite{goldsmith2008complexity}};
		\node[largenode, align=center] at (4.5, 2) {
			Unary input:\\ 
			\textbf{\textcolor{blue}{Co-NPC}}~\cite{goldsmith2008complexity}\\
			Binary input:\\ 
			\textbf{\textcolor{teal}{PP-C}}~\cite{goldsmith2008complexity}};

		\node[anchor=south] at (10.5,8) {Input = DNF};	
		\node[anchor=south] at (5,7.9) {Output = DNF};	
		\node[anchor=south] at (10.5, 6.4) {Input = Truth Table};	
		\node[anchor=south] at (10.5, 3.9) {Input = Monotone};	
		\node[anchor=south] at (12.8, 3.4) {Input = Provenance};	
		\node[anchor=north] at (10.8, 1) {Input = Read-Once};	

	\end{tikzpicture}
	\caption{An overview of related work on the Exact, Minimal Equivalent Expression (MEE) problem applied to formulas.}
	\label{Fig:RelatedWorkBooleanFactorization}
\end{figure*}

\subsection{Probabilistic Inference and Dissociation}

Given a provenance that is not read-once, one can still upper and lower bound its probability efficiently via dissociation~\cite{gatterbauer2014oblivious}:
Let $\varphi$ and  $\varphi'$ be two Boo\-lean formulas with variables $\mathbf{x}$ and $\mathbf{x}'$, respectively.  
Then $\varphi'$ is a \emph{dissociation} of $\varphi$ if there exists 
a substitution $\theta: \mathbf{x}' \rightarrow \mathbf{x}$ s.t.\ 
$\varphi'[\theta]=\varphi$. 
If $\theta^{-1}(x) = \set{x_1',\ldots, x_d'}$, then variable $x$ dissociates into $d$ variables $x_1', \ldots, x_d'$. 
Every provenance expression has a unique read-once dissociation up to renaming of variables. 
One application of compiling provenance polynomials into their smallest representation
is motivated by the following known results on ``oblivious bounds" \cite{gatterbauer2014oblivious}:
($i$) lower and upper bounds for intractable expressions can be found very efficiently;
and ($ii$) those bounds work better the fewer times variables are repeated.
Similarly, anytime approximation schemes based on branch-and-bound provenance decomposition methods \cite{DBLP:conf/sigmod/HeuvelIGGT19,DBLP:journals/vldb/FinkHO13} give tighter bounds if Shannon expansions need to be run on fewer variables.

\subsection{Resilience}

The resilience of a Boolean query measures the minimum number of tuples in database $D$, the removal of which makes the query false. 
The optimization version of this decision problem is then:
given $Q$ and $D$, find the \emph{minimum} $k$ so that $(D,k) \in \res(Q)$. 
A
larger $k$ implies that the query is more ``\emph{resilient}'' and requires
the deletion of more tuples to change the query output. 
We know from \cite{FreireGIM15} that all hard queries must have a ``triad''
which is a set of three 
non-dominated
atoms, $\mathcal{T} = \set{R_1,R_2,R_3}$ 
s.t.\ for every pair $i \neq j$, 
there is a path from $R_i$ to $R_j$ that uses no variable occurring in the other 
atom of $\mathcal{T}$. 
In return, the tractable queries for $\res$ are exactly those that are triad-free.

\begin{example}[Resilience]
    \label{ex:res}
	The resilience for our example from \cref{fig:RST_intro} is $2$
	because removing the set $\Gamma=\{{\color{blue}r_1}, {\color{blue}t_3} \}$
	removes one tuple from each witness and there is no smaller set that achieves that.
\end{example}

\iftoggle{fullappendix}
{
\section{Details for \cref{SEC:FACTORIZATIONSPACE}: Search Space of Factorizations}

\subsection{Bijection Between $\veo$s and Query Plans}
\label{SEC:APPENDIX:VEOQP}

Intuitively, $\veo$s capture the sequence of variables \emph{projected away} in a query plan. To make this bijection clear, we define notation around query plans and give an example. 

\introparagraph{Query plans}
A \emph{query plan} $P$ is given by the grammar
$   P 	::=\, R_i(\vec x)
	 \,\,\,|\,\, 	\projd{- \vec x} P
	 \,\,\,|\,\,\! 	\join{}{P_1, \ldots, P_k} 
$
where $R_i(\vec x)$ is an atom containing the variables
$\vec x$, $\projd{- \vec x}$ is the 
\emph{project-away operator} with
duplicate elimination 
($x$ here is the set of variables being removed), 
and $\join{}{\ldots}$ is the \emph{natural join} in prefix notation, 
which we  allow to be $k$-ary ($k \geq 2$).  

We require that joins and projections alternate in every plan. 
This is w.l.o.g.\
because nested joins,
such as $\join{}{\!\!\join{}{R, S}, T}$ 
or $\join{}{R, \join{}{S, T}}$,
can
be rewritten into $\join{}{R, S, T}$ while keeping the
same provenance

by associativity,
e.g., $(r_1 s_2) t_3 = r_1 (s_2 t_3)$.  
For the same reason, we do not distinguish
between different permutations in the joins,
called join orders~\cite{Moerkotte:BuildingQueryCompilers}. 
Thus, in contrast to standard query plans, our query plans do not necessarily form binary trees.

An important type of query plan is the \emph{hierarchical plan}. 
If a query has a hierarchical plan, then this plan can be used to always produce 
read-once factorized provenance polynomials~\cite{DBLP:conf/sum/OlteanuH08} (and thus the minimum size factorization).
Only \emph{hierarchical queries} have hierarchical plans.

\introparagraph{Notation Related to Query Plans}
We write $\Var(P)$ for all variables in a plan $P$ and $\HVar(P)$ for its \emph{head variables},
which are recursively defined as follows:
(1) if $P = R_i(\vec x)$, then $\HVar(P) = \vec x$;
(2) if $P = \projd{\vec x}(P')$, then $\HVar(P) = \vec x$; and
(3) if $P = \,\join{}{P_1, \ldots, P_k}$ 
then $\HVar(P) = \bigcup_{i=1}^k \HVar(P_i)$.
The \emph{existential variables} $\EVar(P)$ are then defined as $\Var(P) - \HVar(P)$.
Every plan $P$ represents a query $Q_P$
defined by taking all atoms mentioned in $P$ as the body and setting
$\HVar(Q_P) = \HVar(P)$.
A plan is called Boolean if $\HVar(P) = \emptyset$.
For notational convenience, we also use the ``\emph{project-away operator}''
$\proj{\vec y}P$ instead of $\projd{\vec x}P$, where
$\vec y$ are the variables being projected away, 
i.e.\ 
$\vec x = \HVar(\proj{\vec y}P) = \HVar(P) - \vec y$.

\introparagraph{Legal Query Plans}
We assume the usual sanity conditions on plans to be satisfied: for a
projection $\projd{\vec x}P$ we assume $\vec x \subseteq
\HVar(P)$, and each variable $y$ is projected away at most once in a
plan, i.e.\ there exists at most one such operator $\proj{\vec x}$ s.t. $y \in \vec x$.
We also assume that at least one variable is projected away at each projection operator, i.e.\ for all $\proj{\vec x}$ s.t. $\vec x \neq \emptyset$.
We call any such plan \emph{legal}.

\begin{figure}
	\centering
	\begin{subfigure}[b]{.4\columnwidth}
	\centering
	\begin{tikzpicture}[grow=right,<-, level distance=14mm, baseline=-0.5ex, 
			level 1/.style={sibling distance=12mm},
			level 2/.style={level distance=10mm}]
			\node [xshift=1mm] (z){$\pi_{-z}\Join$}
		    child {node(u){$\pi_{-u}T(z,u)$}
		    }
		    child {node(y){$\pi_{-y}\Join$}
		      child {node [xshift=5mm](x){$S(y,z)$}
		      }
		      child{node [xshift=5mm] (R){$\pi_{-x}R(x,y)$}
		      }
			};
	\end{tikzpicture}
	\caption{Query Plan $P$}
	\label{WeightedBipartite:VEO1}
	\end{subfigure}
	\begin{subfigure}[b]{.4\columnwidth}
	\centering
	\begin{tikzpicture}[grow=right,<-, level distance=10mm, baseline=-0.5ex, level 1/.style={sibling distance=13mm}] \node(z){$z$}
	    child {node(u){$u$}
	    }
	    child {node(y){$y$}
	      child {node(x){$x$}
	      }
	    }; 
	    \draw [-,opacity=0.15, line width =10pt, color=dg, line cap=round] (y.center) -- (x.center) node [below, opacity=1] {R} ;
	    \draw [-,opacity=0.15, line width =10pt, color=red, line cap=round] (z.center) -- (y.center) node [below, opacity=1] {S};
	    \draw [-,opacity=0.15, line width =10pt, color=blue, line cap=round] (z.center) -- (u.center) node [below, opacity=1] {T};
	\end{tikzpicture}
	\caption{$\veo$ $v$}
	\label{WeightedBipartite:VEO2}
	\end{subfigure}
\caption{\Cref{ex:VEO3chain}: A query plan $P$ and the corresponding $\veo$ $v$.}
\end{figure}

\begin{example}[Query Plans and Variable Elimination Orders]
\label{ex:VEO3chain}
Consider the 3-chain Query
$\qthreechain \datarule R(x,y), S(y,z), T(z,u)$
and the legal query plan $P = \proj{z}\!\join{}{\proj{y}\!\join{}{\proj{x}R(x,y),S(y,z)}, \proj{u}T(z,u)}$, 
also shown in 
\cref{WeightedBipartite:VEO1}.
The corresponding $\veo$ $v$ is shown in \cref{WeightedBipartite:VEO2}. 
In order to refer to a particular $\veo$, 
we use a linear notation with parentheses representing sets of children i.e.\ \cref{WeightedBipartite:VEO2} can be written as
$v = z \!\leftarrow\! (u, y \!\leftarrow\! x)$.
To make it a unique serialization, we only need to assume an ordering on the children of each parent, which we achieve by using the alphabetic ordering on the variables.
Notice that our definition of $\veo$ also allows sets of variables as nodes. As an extreme example, the legal query plan
$P' = \proj{xyzu}\join{}{R(x,y),S(y,z),T(z,y)}$ corresponds to a $\veo$ $v'$ with one single node containing all variables.
In our short notation, we denote nodes with multiple variables in brackets without commas between the variables 
in order to distinguish them from children: $v' = \{xyzu\}$. 
\end{example}

Notice that also minimal $\veo$s can have multiple variable in a node
(e.g.\ $\{yz\} \!\leftarrow\! x$ for query $\qtriangle$).

\subsection{Details and Examples: VEO Instances, Table Prefixes and Factorization Forests}

In this section, we include \cref{ex:veoveoinstanceveoff}, which illustrates the $\veo$s and {\veoff}s mentioned in \cref{ex:veoinstance,ex:tableprefixinstance,ex:veoff} of the main text.

We can also define intuitive ``merging'' and ``splitting'' operations on $\veo$s.
The merge operation (\cref{ex:mergeveos}) takes two $\veo$ instance forests and combines any trees in the forests if they share path(s) from the root.

\begin{definition}[$\veo$ Merge]
	A merge operation on one or more $\veo$ instances or $\veo$ instance forests $\mathcal{V}_1$ and $\mathcal{V}_2$ outputs a new $\veo$ instance forest $\mathcal{V}$ such that:
	(1) All VEO instances in $\mathcal{V}_1$ and $\mathcal{V}_1$ are also present in $\mathcal{V}$.
	(2) There are no two trees $t_1, t_2 \in \mathcal{V}$ such that $t_1$ and $t_2$ share a prefix for any variable that is in $t_1$ and $t_2$.
	\label{def:veomerge}
\end{definition}

\begin{example}[Merging $\veo$ instances]
\label{ex:mergeveos}
Consider the two $\veo$ instances $v_{\w_1}= x_1 \!\leftarrow\! y_1 \!\leftarrow\! z_1$ and $v_{\w_2}= x_1 \!\leftarrow\! y_1 \!\leftarrow\! z_2$. 
Then merging $v_{\w_1}$ and $v_{\w_2}$ would result in the instance $x_1 \!\leftarrow\! y_1 \!\leftarrow\! (z_1, z_2)$, since the common rooted path $x_1 \!\leftarrow\! y_1 $ in both $\veo$s can be combined.

On the other hand, $v_{\w_3}= x_1 \!\leftarrow\! y_1 \!\leftarrow\! z_1$ and $v_{\w_4}= x_2 \!\leftarrow\! y_1 \!\leftarrow\! z_1$ have no common rooted path, and their merger results in simply the forest \{$x_1 \!\leftarrow\! y_1 \!\leftarrow\! z_1$, $x_2 \!\leftarrow\! y_1 \!\leftarrow\! z_1$\}.

Another example is merging the $\veo$ instances $v_{\w_5}= y_1 \!\leftarrow\! (x_1, z_1)$ and $v_{\w_6}= y_1 \!\leftarrow\! (x_2, z_2)$ results in combining the shared path $y_1 \!\leftarrow\! ((x_1, x_2), (z_1, z_2))$. 
Notice that this includes the $\veo$ instances $y_1 \!\leftarrow\! (x_2, z_1)$ and $y_1 \!\leftarrow\! (x_1, z_2)$, 
which must be present in the data anyway due to join dependencies.

We can also merge $\veo$ instances from different $\veo$s. The merger of $v_{\w_7}= x_1 \!\leftarrow\! y_1 \!\leftarrow\! z_1$ and $v_{\w_8}= x_1 \!\leftarrow\! z_2 \!\leftarrow\! y_2$ leads to 
$x_1 \!\leftarrow\! (y_1 \!\leftarrow\! z_1, z_2 \!\leftarrow\! y_2)$.
\end{example}

The splitting operation (\cref{ex:splitveos,fig:veosplit}) splits a $\veo$ with more than one leaf into nested paths from the root to the leaves.
This property is useful to define nested prefix orderings in \cref{SEC:MINCUT} (with an example in \cref{sec:nestingRP}).

\begin{definition}[$\veo$ Split]
	A split operation on a $\veo$ instance outputs the set of root-to-leaf paths in the $\veo$ instance.
	\label{def:veosplit}
\end{definition}

\begin{example}[Splitting $\veo$ instances]
\label{ex:splitveos}
The $\veo$ $x_1 \!\leftarrow\! (y_1, z_1)$ can be split into nested paths $x_1 \!\leftarrow\! y_1$ and $x_1 \!\leftarrow\! z_1$.
\end{example}

\begin{figure}
	\centering
	\begin{subfigure}[b]{.4\columnwidth}
		\centering
		\begin{tikzpicture}[<-, level distance=8mm,sibling distance=10mm,baseline=-0.5ex, grow=right] \node(x){$x$}
		    child {node(u){$u$}
		      child {node(y){$y$}
    		      child {node(z){$z$}
    		      }
		      }
		      child {node(w){$w$}
		        child {node(v){$v$}
    		    }
		      }
		    };
		    \draw [-,opacity=0.15, line width =10pt, color=blue, line cap=round, label=above:\(z\)] 
				(x.center) -- (y.center) node [left, opacity=1] {};
			\draw [-,opacity=0.15, line width =10pt, color=black, line cap=round, label=above:\(z\)] 
				(x.center) -- (w.center) node [left, opacity=1] {};
		    \draw [-,opacity=0.15, line width =10pt, color=green, line cap=round] 
				(u.center) -- (z.center) node [left, opacity=1] {} ;
		    \draw [-,opacity=0.15, line width =10pt, color=red, line cap=round] 
				(y.center) -- (z.center) node [right, opacity=1] {};
			\draw [-,opacity=0.15, line width =10pt, color=teal, line cap=round, label=above:\(z\)] 
				(w.center) -- (v.center) node [left, opacity=1] {};
			\draw [-,opacity=0.20, line width =10pt, color=yellow
			, line cap=round, label=above:\(z\)] 
				(u.center) -- (v.center) node [left, opacity=1] {};
			\draw [-,opacity=0, line width =10pt, color=orange, line cap=round] (x.center) node [above, opacity=1] {$1$};
			\draw [-,opacity=0, line width =10pt, color=orange, line cap=round] (u.center) node [above, opacity=1] {$1$};
			\draw [-,opacity=0, line width =10pt, color=orange, line cap=round] (y.center) node [above, opacity=1] {$2$};
			\draw [-,opacity=0, line width =10pt, color=orange, line cap=round] (z.center) node [above, opacity=1] {$3$};
			\draw [-,opacity=0, line width =10pt, color=orange, line cap=round] (w.center) node [above, opacity=1] {$2$};
			\draw [-,opacity=0, line width =10pt, color=orange, line cap=round] (v.center) node [above, opacity=1] {$3$};
		\end{tikzpicture}
		\caption{$v_1$}
	\end{subfigure}
	\begin{subfigure}[b]{.4\columnwidth}
	\centering
	\begin{minipage}[b][0.1\textheight][s]{0.1\textwidth}
			\centering
			\begin{tikzpicture}[<-, level distance=8mm,sibling distance=10mm, grow=right] 
				\node(x){$x$}
				child {node(u){$u$}
				  child {node(w){$w$}
					  child {node(v){$v$}
					  }
				  }
				};
				\draw [-,opacity=0.15, line width =10pt, color=black, line cap=round, label=above:\(z\)] 
					(x.center) -- (w.center) node [left, opacity=1] {};
				\draw [-,opacity=0.15, line width =10pt, color=teal, line cap=round, label=above:\(z\)] 
					(w.center) -- (v.center) node [left, opacity=1] {};
				\draw [-,opacity=0.20, line width =10pt, color=yellow
				, line cap=round, label=above:\(z\)] 
					(u.center) -- (v.center) node [left, opacity=1] {};
				\draw [-,opacity=0, line width =10pt, color=orange, line cap=round] (x.center) node [above, opacity=1] {$1$};
				\draw [-,opacity=0, line width =10pt, color=orange, line cap=round] (u.center) node [above, opacity=1] {$1$};
				\draw [-,opacity=0, line width =10pt, color=orange, line cap=round] (w.center) node [above, opacity=1] {$2$};
				\draw [-,opacity=0, line width =10pt, color=orange, line cap=round] (v.center) node [above, opacity=1] {$3$};
			\end{tikzpicture}
			\vfill
			\begin{tikzpicture}[<-, level distance=8mm,sibling distance=10mm, grow=right] 
				\node(x){$x$}
				child {node(u){$u$}
				  child {node(y){$y$}
					  child {node(z){$z$}
					  }
				  }
				};
				\draw [-,opacity=0.15, line width =10pt, color=blue, line cap=round, label=above:\(z\)] 
					(x.center) -- (y.center) node [left, opacity=1] {};
				\draw [-,opacity=0.15, line width =10pt, color=green, line cap=round] 
					(u.center) -- (z.center) node [left, opacity=1] {} ;
				\draw [-,opacity=0.15, line width =10pt, color=red, line cap=round] 
					(y.center) -- (z.center) node [right, opacity=1] {};
				\draw [-,opacity=0, line width =10pt, color=orange, line cap=round] (x.center) node [above, opacity=1] {$1$};
				\draw [-,opacity=0, line width =10pt, color=orange, line cap=round] (u.center) node [above, opacity=1] {$1$};
				\draw [-,opacity=0, line width =10pt, color=orange, line cap=round] (y.center) node [above, opacity=1] {$2$};
				\draw [-,opacity=0, line width =10pt, color=orange, line cap=round] (z.center) node [above, opacity=1] {$3$};
			\end{tikzpicture}
	\end{minipage}
	\caption{$v_1$ split into paths}
	\end{subfigure}
	\caption{The $\veo$ split operation in \cref{ex:splitveos}}
	\label{fig:veosplit}
\end{figure}

\begin{figure}
	\begin{subfigure}[t]{.45\columnwidth}
		\centering
		\begin{tikzpicture}[grow=right,<-, level distance=10mm, level 1/.style={sibling distance=13mm}] \node(z){$z$}
			child {node(u){$u$}
			}
			child {node(y){$y$}
			  child {node(x){$x$}
			  }
			}; 
			\draw [-,opacity=0.15, line width =10pt, color=dg, line cap=round] (y.center) -- (x.center) node [below, opacity=1] {R} ;
			\draw [-,opacity=0.15, line width =10pt, color=red, line cap=round] (z.center) -- (y.center) node [below, opacity=1] {S};
			\draw [-,opacity=0.15, line width =10pt, color=blue, line cap=round] (z.center) -- (u.center) node [below, opacity=1] {T};
		\end{tikzpicture}
		\caption{An example $\veo$ $v$ for $\qtwostar$}
		\label{fig:ex:veo1:2star}
	\end{subfigure}
	\hspace{5mm}
	\begin{subfigure}[t]{.45\columnwidth}
			\centering
			\begin{tikzpicture}[grow=right,<-, level distance=10mm, level 1/.style={sibling distance=13mm}] \node(z){$z_3$}
				child {node(u){$u_4$}
				}
				child {node(y){$y_2$}
				  child {node(x){$x_1$}
				  }
				}; 
				\draw [-,opacity=0.15, line width =10pt, color=dg, line cap=round] (y.center) -- (x.center) node [below, opacity=1] {R} ;
				\draw [-,opacity=0.15, line width =10pt, color=red, line cap=round] (z.center) -- (y.center) node [below, opacity=1] {S};
				\draw [-,opacity=0.15, line width =10pt, color=blue, line cap=round] (z.center) -- (u.center) node [below, opacity=1] {T};
			\end{tikzpicture}
			\caption{An example $\veo$ instance $v\angle\w$ for $\qtwostar$ where $\w = (x_1,y_2,z_3,u_4)$}
			\label{fig:ex:veoinstance1:2star}
		\end{subfigure}
	\begin{subfigure}[t]{.45\columnwidth}
		\centering
		\begin{tikzpicture}[grow=right,<-, level distance=10mm, level 1/.style={sibling distance=13mm}] \node(z){$z_1$}
			child {node(y){$y_1$}
			}; 
		\end{tikzpicture}
		\caption{An example $\veo$ table prefix instance $v^S\angle\w$ for $\qtwostar$}
	\label{fig:ex:veotableprefixinstance1:2star}
	\end{subfigure}
	\hspace{5mm}
	\begin{subfigure}[t]{.45\columnwidth}
		\centering
		\begin{tikzpicture}[grow=right,<-, level distance=3mm, 
			level 1/.style={sibling distance=15mm, level distance=8mm},
			level 2/.style={sibling distance=5mm},
			level 3/.style={sibling distance=10mm},
			level 3/.style={sibling distance=3mm}]
			\node [xshift=0mm] {}
			child {node {$\textcolor{blue}{x_1}$}
				child{node  {$\textcolor{blue}{u_1}$}
				}
				child{node {$\textcolor{blue}{u_2}$}
				}
				child{node {$\textcolor{blue}{y_1}$}
					child{node {$\textcolor{blue}{z_1}$}
					}
				}
			edge from parent[draw=none]
			};
	\end{tikzpicture}
	\caption{VEO Factorization Forest (\veoff) $\mathcal{V}_1 =$\\$ [x_1 \!\leftarrow\! (u_1, u_2, (y_1 \leftarrow z_1)) ]$ }
	\label{fig:ex:veoff1}
	\end{subfigure}
	\begin{subfigure}[t]{.45\columnwidth}
		\centering
		\begin{tikzpicture}[grow=right,<-, level distance=3mm, 
			level 1/.style={sibling distance=15mm, level distance=8mm},
			level 2/.style={sibling distance=5mm},
			level 3/.style={sibling distance=10mm},
			level 3/.style={sibling distance=3mm}]
			\node [xshift=0mm] {}
			child {node {$\textcolor{blue}{y_1}$}
					child {node {$\textcolor{blue}{x_1}$}
				}
			edge from parent[draw=none]}
			child {node {$\textcolor{blue}{x_2}$}
				child{node [yshift=-3mm] {$\textcolor{blue}{y_2}$}
				}
				child{node [yshift=3mm] {$\textcolor{blue}{y_3}$}
				}
			edge from parent[draw=none]
			};
	\end{tikzpicture}
	\caption{VEO Factorization Forest (\veoff) $\mathcal{V}_2 =$\\$ [x_1 \!\leftarrow\! (u_1, (y_1 \leftarrow z_1)) ; y_1 \!\leftarrow\! (z_1, (x_1 \leftarrow u_1))  ]$    }
	\label{fig:ex:veoff2}
	\end{subfigure}
	\hspace{5mm}
	\begin{subfigure}[t]{.45\columnwidth}
		\centering
		\begin{tikzpicture}[grow=right,<-, level distance=3mm, 
			level 1/.style={sibling distance=15mm, level distance=8mm},
			level 2/.style={sibling distance=5mm},
			level 3/.style={sibling distance=10mm},
			level 3/.style={sibling distance=3mm}]
			\node [xshift=0mm] {}
			child {node {$\textcolor{blue}{x_1}$}
					child {node {$\textcolor{blue}{y_1}$}
					child {node {$\textcolor{blue}{z_1}$}}
					}
					child {node {$\textcolor{blue}{u_1}$}
				}
			edge from parent[draw=none]}
			child {node {$\textcolor{blue}{y_1}$}
				child{node [yshift=-3mm] {$\textcolor{blue}{z_1}$}
				}
				child{node [yshift=3mm] {$\textcolor{blue}{x_1}$}
					child{node [yshift=3mm] {$\textcolor{blue}{u_1}$}
						child{node {$\textcolor{blue}{u_2}$}
						}
					}
				}
			edge from parent[draw=none]
			};
	\end{tikzpicture}
	\caption{An invalid VEO Factorization Forest (\veoff) \\ $\mathcal{V}_3 = [x_1 \!\leftarrow\! (u_1, (y_1 \leftarrow z_1)) ; y_1 \!\leftarrow\! (z_1, (x_1 \leftarrow u_1 \leftarrow u_2))  ]$   }
	\label{fig:ex:veoff3}
	\end{subfigure}
\caption{\cref{ex:veoinstance,ex:tableprefixinstance,ex:veoff}: Examples of $\veo$ instances, $\veo$ table prefix instances and $\veo$ factorization forests
}
\label{ex:veoveoinstanceveoff}
\end{figure}

\subsection{Connections of $\veo$s to Related Work}

\introparagraph{Connection to FBDs}
Our definition of $\veo$s parallels the definition of variable orders defined in Factorized Databases~\cite{DBLP:journals/sigmod/OlteanuS16,DBLP:journals/tods/OlteanuZ15}, however has two key differences: (i) we do not allow any caching in the variable order, and (ii) we allow multiple variables in each node of the $\veo$.

Thus, we allow the $\{yz\} \!\leftarrow\! x$, while an FDB style definition would force every factorization to choose between $y \!\leftarrow\! z \!\leftarrow\! x$ and $z\!\leftarrow\! y \!\leftarrow\! x$.  

We see that using multiple variables in each node of the $\veo$ allows us to build a more restricted search space, since some $\veo$s with multiple variables in the nodes are minimal $\veo$s (defined through notions of minimal query plans \cite{gatterbauer2014oblivious}, later described in \cref{sec:appendix:minveos}). 
An example of a query with such a minimal $\veo$ is the triangle unary query $\qtriangleunary$ (\cref{ex:triangleunaryendtoend}).

\introparagraph{Connection to Knowledge Representation}
We can draw a parallel between $\veo$s and $\veo$ instances as the notion of v-trees and DNNF formulas in knowledge compilation~\cite{pipatsrisawat2008new}, as the v-trees and $\veo$s both capture the \emph{structure} of the data in the $\veo$ instance or DNNF, respectively.

\subsection{Proof of \cref{thm:factorizationveoff}: Reversible Transformation Between Factorization Trees and $\veo$ Factorization Forests}

\thmfactorizationveo*

\begin{proof}[Proof~\cref{thm:factorizationveoff}]
We first describe in \cref{alg:transformfttoveoff} a partial mapping from a factorization tree (\factorizationtree)
to a $\veo$ factorization forest (\veoff).
Note that such a transformation may result in an error if it does not succeed.
We show that this process is always reversible if it succeeds.
We then show that for at least one minimal $\factorizationtree$ of any provenance formula $\varphi_p$, 
such a transformation succeeds, and is thus also reversible.

\cref{alg:transformfttoveoff} returns a structure consistent with the definition of a $\veo$ factorization forest since 
(1) for each witness there is a unique $\veo$ instance in the $\veoff$ corresponding to sequence of $\oplus$ and $\otimes$ operations that evaluated to the witness in the $\factorizationtree$, 
(2) taking away any variable from the structure would remove at least once $\veo$ instance (corresponding to a witness of the literal the variable is derived from), hence the $\veoff$ is minimal in line with the required definition.

If the \cref{alg:transformfttoveoff} does indeed return a $\veoff$, then we can show that this transformation is always reversible by applying the transformation detailed in \cref{alg:transformveofftoft}.
We can prove step-by-step that the \cref{alg:transformveofftoft} recovers the same $\factorizationtree$ as the one that was used to construct the $\veoff$.
We see that the relationship between $\otimes$ nodes in the factorization trees is never modified in the transformation to the $\veoff$, and hence substituting the nodes of the $\veoff$ with $\otimes$ nodes directly recovers the $\otimes$ relationships. Due to the alternating operators in the factorization trees, we can also recover the $\oplus$ nodes by placing them between $\otimes$ nodes. 
The leaves can be recovered from the table prefix instances. 
Since in the transformation to $\veoff$s we projected away variables ``as soon as possible'', the $\veo$ table prefix instances can recover losslessly where each tuple was joined in the original factorization tree.

However, consider that \cref{alg:transformfttoveoff} returns an error for a given $\factorizationtree$.
Then there exists a join node $n$ in the ``reduced'' annotated $\factorizationtree$ built during \cref{alg:transformfttoveoff} that is not annotated. 
This happens when all the variables at that join node are used in subsequent joins.
Like in \cref{alg:transformfttoveoff}, let $D$ be the relations of the descendant leaves of $n$ and $Q'$ be the query obtained by removing atoms from $Q$ that correspond to relations in $D$.
We also define $Q''$ as the query obtained by keeping only atoms from $Q$ that correspond to relations in $D$.
Thus, $Q$ is split into $Q'$ and $Q''$.
If we have an error condition, there is a set of relations $\vec R$ in $Q'$ that is a superset of $\var(Q'')$.
We see that the factorization length cannot increase if we merge the joins of the relations $\vec R$ with the relations of $Q''$.
We can now try the transformation again on the new $\factorizationtree$ with the merged joins.
We can repeat this process until we have a $\factorizationtree$ that does not result in an error. 
We know that we always make progress since the merging of joins reduces the number of internal nodes of the $\factorizationtree$.
\end{proof}

\begin{algorithm}
\SetKwInOut{Input}{Input}
\Input{A factorization tree $\factorizationtree$ over provenance formula $\varphi_p$ of query $Q$}
\KwResult{A $\veoff$ or an error}
\tcc{Build annotated factorization tree:}

Replace $\otimes$ nodes with a join ($\bowtie$) and $\oplus$ nodes with a projection ($\pi$) \\
\tcc{Annotating the nodes from the leaves to root}
\While{$\exists$ a node $n$ that is not annotated}{
	\If{$n$ is a leaf node }{Annotate with annotated-domain variables of the tuple at the leaf\\}
	\If{Join node}{Annotate with the union of the annotations of the children of $n$ \\}
	\If{Projection node}{
		$v$ = union of annotations of the children of $n$ \\
		$D$ = set of all relations of descendants of $n$ \\
		$Q'$ = query obtained by removing atoms from $Q$ that correspond to relations in $D$\\
		$v'$ = The annotated-domain variables in $v$ that do not correspond to variables of $\var(Q')$ \tcc{These are the variables that can be projected away since they do not feature in any subsequent joins $Q'$}
		Annotate $n$ with $v'$\\}
}
\tcc{Reduce the join annotations:}
	\If{$v$ is a variable that is part of the annotation of a join node $n$}{
		\If{$\exists n'$ that is an ancestor of $n$ and contains $v$}{
			Remove $v$ from $n$\\
		}
	}
\If{A join node has no annotation after minimization}{
	\Return{Error}
}
\tcc{Extract the join annotations:}
From the annotated $\factorizationtree$, build a forest of join nodes where the parent of a join node is its grandparent in the $\factorizationtree$ (thus skipping over $\pi$ nodes)\\
Replace join nodes with the annotations obtained after the minimization step to obtain a $\veoff$ \\
\Return{The constructed $\veoff$}
\caption{Transformation of a factorization tree to a $\veo$ factorization forest}
\label{alg:transformfttoveoff}
\end{algorithm}

\begin{algorithm}
	\SetKwInOut{Input}{Input}
	\Input{A $\veoff$ built from a $\factorizationtree$ via \cref{alg:transformfttoveoff}}
	\KwResult{The original $\factorizationtree$}

	\tcc{Re-add literals to the $\veoff$}
	Identify the paths in the $\veoff$ that correspond to table prefix instances\\
	\If{Node $n$ is the leaf of a table prefix instance}{
		Add a child to $n$ corresponding to the tuple represented by the table prefix instance\\
	}
	\tcc{Re-add $\oplus$ and $\otimes$ nodes}{
		Replace each node in the VEOFF with a $\otimes$ node \\
		Add a $\oplus$ node as a child for every $\otimes$ node\\
		Replace the parent of each $\otimes$ node with the $\oplus$ child of the original parent\\
	}
	\If{there were multiple trees in the $\veoff$}{
		Add an $\oplus$ node as the parent of all the root $\otimes$ nodes\\
	}
	\Return{The reconstructed $\factorizationtree$}
	\caption{Transformation of a factorization tree to a $\veo$ factorization forest}
	\label{alg:transformveofftoft}
\end{algorithm}

\begin{example}[Instance where a $\factorizationtree$ cannot be transformed to a $\veoff$]
	Consider the factorization tree $\otimes(s_{11},\oplus(\otimes(r_1,t_1)))$ for the $\qtwostar$ query over the simple database with $3$ tuples $R(1)$, $S(1,1)$, $T(1)$.

	The join node that is the parent of $r_1$ and $s_1$ is annotated with $x_1$ and $y_1$.
	However, both these variables are also required at the join node with $s_{11}$ and so cannot be projected away.
	Thus, in the step where we reduce the join annotations in \cref{alg:transformfttoveoff}, both $x_1$ and $y_1$ are removed from the deeper join node, leading to an empty annotation. 
	The transformation thus fails.

	We can see intuitively that this is due to the fact a join is performed on two relations whose variables are subsumed by the variables of another relation.
	The joins can thus be reordered without increasing the size of the factorization tree.
\end{example}

\subsection{Proof of \cref{thm:minfacveo}: Construction of Factorization Trees from $\veo$s}

\thmminfacveo*

\begin{proof}[Proof~\cref{thm:minfacveo}]
	Through \cref{thm:factorizationveoff}, we know that for any provenance
	formula $\varphi_p$ there exists a minimal size $\factorizationtree$ that has a correspondence to a $\veoff$. 
	We show all such $\veoff$s can be constructed by assigning a $\veo$ to each witness of $\varphi_p$.
	From the assignment of $\veo$s to each witness, we are able to construct the corresponding $\veo$ instances.
	We use the merge operation (\cref{def:veomerge}) on $\veo$ instances greedily to merge common prefixes. 
	The merge operation is both commutative and associative, and no matter the order of merging, all $|\vec w|$ $\veo$ instances will produce exactly the same forest after merging.
	This resulting forest is the $\veoff$ corresponding to the minimal $\factorizationtree$, since (1) it contains a $\veo$ instance for each witness of $\varphi_p$ (2) it cannot be made smaller since the merge operation maximally merges all shared prefixes.
	Since we can build a $\veoff$ corresponding to a minimal formula of $\varphi_p$ from the assignment of $\veo$s to witnesses, we can also $\factorizationtree$ corresponding to a minimal formula of $\varphi_p$ through the reverse transformation \cref{alg:transformveofftoft}.
\end{proof}

\subsection{Details about Minimal $\veo$s}
\label{sec:appendix:minveos}

We show a connection between minimal $\veo$s and minimal query plans from \cite{gatterbauer2014oblivious}.
To do this, we leverage a correspondence of $\veo$s to query plans and the previously defined notion of dissociations from probabilistic databases \cite{gatterbauer2014oblivious}.

Intuitively, the goal of dissociations in probabilistic inference~\cite{gatterbauer2014oblivious} is to obtain better inference bounds, by converting it to a hierarchical query, for which inference is easy.
Our purpose in using dissociations is in the same spirit - we know that hierarchical queries have read-once factorizations, and hence we would like to use dissociations to obtain a dissociated database on which we can obtain the original provenance by using the dissociated query.

We here recap the key notions of dissociations from \cite{gatterbauer2014oblivious} and provide examples.
But first, we add some more detail on provenance computation from query plans and hierarchical query plans.

\introparagraph{Provenance computation from Query Plans}
We slightly adapt our previous definition of query plans to ``provenance query plans.''
Each sub-plan $P$ now
returns an
intermediate relation of arity $|\HVar(P)|+1$. 
The extra
\emph{provenance attribute} stores an expression $\varphi$ for each output tuple
$t \in P(D)$ returned from a plan $P$.

Given a database $D$ and a plan $P$, 
$\varphi$ is defined
inductively on the structure of $P$ as follows: 
(1) If $t \in R_i(\vec x)$, 
then $\varphi(t)= \varphi_t$, i.e.\ the provenance token of $t$ in $D$;
(2) if $t \in \,
\join{}{P_1(D), \ldots, P_k(D)}$ where 
$t = \join{}{t_1, \ldots,
  t_k}$, then 
  $\varphi(t) = \bigwedge_{i=1}^k\varphi(t_i)$; 
  and
(3) if $t \in
\projd{\vec x}P(D)$, and $t_1, \ldots, t_n \in P(D)$ are all the
tuples that project onto $t$, then 
$\varphi(t) = \big(\bigvee_{i=1}^n \varphi(t_i) \big)$ (notice the required outer parentheses).
For a Boolean plan $P$, we get one single expression,
which we denote $\varphi(P, D)$.  

\introparagraph{Hierarchical Query Plans}
An important type of plan is the \emph{hierarchical plans}. 
Intuitively, if a query has a hierarchical plan, then this plan can be used to always produce read-once factorized provenance polynomial~\cite{DBLP:conf/sum/OlteanuH08} (and thus the minimum size factorization).
The only queries that have hierarchical plans are the \emph{hierarchical queries}.

\begin{definition}[Hierarchical query~\cite{DBLP:journals/vldb/DalviS07}]
\label{def:hierarchicalQ}
A query $Q$ is called hierarchical iff for any two existential variables
$x,y$,
one of the following three conditions holds:
$\at(x) \subseteq \at(y)$, $\at(x) \supseteq \at(y)$, or $\at(x) \cap \at(y) = \emptyset$ where $\at(x)$ is the set of atoms of $Q$ in which $x$ participates.
\end{definition}

\begin{definition}[Hierarchical plan]\label{def:safePlan}
A plan $P$ is called \emph{hierarchical} iff, 
for each join  $\join{}{P_1, \ldots, P_k}$, 
the head variables of each sub-plan $P_i$ contain the same existential variables of plan $P$:
$\HVar(P_i) \cap \EVar(P) = \HVar(P_j) \cap \EVar(P), \forall i, j \in [k]$
\end{definition}

\begin{example}[Hierarchical Queries]
	The $2$-chain query $\qtwochain \datarule$ 
	$R(x,y), S(y,z)$ is a hierarchical query since for each pair of variables $a,b$ either $\at(a)$ and $\at(b)$ are disjoint (like the variables $x$ and $z$, since $\at(x)= \{R\}$ and $\at(z)= \{S\}$) or have a subset relationship (like $\at(x)= \{R\}$ and  $\at(y)= \{R, S\}$).

However, the $3$-chain query $\qthreechain \datarule R(x,y), S(y,z), T(z,u)$ is not hierarchical since $\at(y) = \{R, S\}$ and $\at(z) = \{S,T\}$ 
and these are overlapping sets with neither being a subset of the other.
\end{example}

\begin{example}[Hierarchical Plan]
	\label{ex:hierarchicalplan2chain}
	Consider two query plans for the query $\qtwochain$ in \cref{Fig:ex:twochainqps}.

The query plan $P_2$ is \emph{not} a hierarchical plan since for the sub-plan $\Join(R(x,y), S(y,z))$ 
the head variables for $R(x,y)$ are $\{x, y\}$
and for $S(y,z)$ they are $\{y, z\}$.
Thus, they contain different sets of existential variables of $P_2$.

However, in $P_1$, at the only join $\Join(\pi_{-x}R(x,y), \pi_{-z}S(y,z))$, the head variables of both subplans $\pi_{-x}R(x,y)$ and $\pi_{-z}S(y,z)$ 
are the same set $\{y\}$, so the plan is hierarchical.
\end{example}

\begin{figure}
	\centering
	\begin{subfigure}[b]{.45\columnwidth}
	\centering
	\begin{tikzpicture}[grow=right,<-, level distance=20mm, baseline=-0.5ex, 
			level 1/.style={sibling distance=12mm},
			level 2/.style={level distance=20mm}]
			\node [xshift=1mm] (z){$\pi_{-y}\Join$}
		    child {node(u){$\pi_{-x}R(x,y)$}
		    }
		    child {node(y){$\pi_{-z}S(y,z)$}
		  };
	\end{tikzpicture}
	\caption{Query Plan $P_1$}
	\label{Fig:ex:twochainqp1}
	\end{subfigure}
	\begin{subfigure}[b]{.45\columnwidth}
	\centering
	\begin{tikzpicture}[grow=right,<-, level distance=20mm, baseline=-0.5ex, 
		level 1/.style={sibling distance=12mm},
		level 2/.style={level distance=20mm}]
		\node [xshift=1mm] (z){$\pi_{-xyz}\Join$}
		    child {node(u){$R(x,y)$}
		    }
		    child {node(y){$S(y,z)$}
		  };
	\end{tikzpicture}
	\caption{Query Plan $P_2$}
	\label{Fig:ex:twochainqp2}
	\end{subfigure}
\caption{\Cref{ex:hierarchicalplan2chain}: Two legal
query plans for $\qtwochain$.
$P_1$ is a hierarchical minimal plan while $P_2$ is not.}
\label{Fig:ex:twochainqps}
\end{figure}

\introparagraph{Query Dissociation~\cite[Section 3.1]{DBLP:journals/vldb/GatterbauerS17}}
We leverage the previously defined idea of query dissociation to generate query plans that ``come as close as possible'' to hierarchical plans, to help find the minimal factorization for queries that are not read-once.

A query dissociation is a rewriting of both the data and the query by adding variables to atoms in the query.
In this paper, we care only about hierarchical dissociations, i.e.\ dissociations s.t.\ the rewritten dissociated query plan is hierarchical. 
Dissociations can be compared with a partial order, and for our problem, we only care about minimal hierarchical dissociations and their corresponding plans, known as \emph{minimal query plans}.
There exists an algorithm that can find all such minimal query plans~\cite{DBLP:journals/vldb/GatterbauerS17}, and we assume them as input.

\begin{definition}[Query dissociation]\label{def:Dissociation}
Given a query 
$Q(\vec z) \datarule R_1(\vec x_1), $ $\ldots, R_m(\vec x_m)$.
Let $\Delta = (\vec y_1, \ldots, \vec y_m)$ be a collection of sets
of variables with $\vec y_i \subseteq \EVar(Q)- \vec x_i$ for every relation $R_i$.
The ``\emph{dissociated query}'' defined by $\Delta$ is then
$Q^{\Delta}(\vec z) \datarule R_1^{\vec y_1}(\vec x_1, \vec y_1), \ldots, R_m^{\vec y_m}(\vec x_m, \vec y_m)$
where each $R_i^{\vec y_i}(\vec x_i, \vec
y_i)$ is a new relation of arity $|\vec x_i| + |\vec y_i|$.
\end{definition}

\begin{example}[Query Dissociation]
	\label{ex:querydiss}
	Consider the non-hierarchical query of $\qthreechain \datarule R(x,y), S(y,z), $ $T(z,u)$ again.
	Consider the dissociation $\Delta = (\{z\}, \{ \}, \{ \})$. 
	Then the dissociated query is ${\qthreestar}^\Delta \datarule R(x,y,z),$ $S(y,z), T(z,u)$.
\end{example}

Intuitively, the goal of dissociations in probabilistic inference~\cite{DBLP:journals/vldb/GatterbauerS17} is to obtain better inference bounds, by converting it to a hierarchical query, for which inference is easy.
Our purpose in using dissociations is in the same spirit - we know that hierarchical queries have read-once factorizations, and hence we would like to use dissociations to obtain a dissociated database on which we can obtain the original provenance by using the dissociated query.

\introparagraph{Hierarchical Dissociation}

\begin{definition}[Hierarchical dissociation]
A dissociation $\Delta$ of a query $Q$ is called \emph{hierarchical}
if the dissociated query ${Q^{\Delta}}$ is hierarchical.
We denote the corresponding hierarchical plan applied over the original relations $P^\Delta$.
\end{definition}

For every sf-free CQ, there is an isomorphism between the set of query plans and the set of hierarchical dissociations \cite{DBLP:journals/vldb/GatterbauerS17}. 
Thus we can restrict the dissociations we consider to hierarchical dissociations.

\begin{example}[Hierarchical Dissociation]
	Consider again the non-hierarchical query of $\qthreechain$ 
	and its dissociation $\Delta = (\{z\}, \{ \}, \{ \})$ from \cref{ex:querydiss}.
	Then the dissociated query ${\qthreestar}^\Delta \datarule R(x,y,z), S(y,z),$ $T(z,u)$ is hierarchical
	as all pairs of variables satisfy the constraint for hierarchical queries. 
	Thus $\Delta$ is a hierarchical dissociation.
\end{example}

\introparagraph{Partial Dissociation Orders}
The number of dissociations to be considered can be further reduced by exploiting an ordering on the dissociations.

\begin{definition}[Partial dissociation order]
Given two dissociations
$\Delta = (\vec y_1, \ldots, \vec y_m)$
and
$\Delta' = (\vec y_1', \ldots, \vec y_m')$.
We define the partial order on the dissociations of a query as:
$\Delta \preceq \Delta'   \,\,\Leftrightarrow\,\,  \forall i: \vec y_i \subseteq \vec y_i'$.
\end{definition}

A minimal hierarchical dissociation is one such that no smaller dissociation is hierarchical.

\begin{example}[Partial Dissociation Order]
	For $\qthreechain$, the dissociation $\Delta = (\{z\}, \{ \}, \{ \})$ is a minimal hierarchical dissociation while $\Delta' = (\{z\}, \{ \}, \{y\})$ is hierarchical but not minimal.

	Additionally, we can say that $\Delta \preceq \Delta'$.
\end{example}

\introparagraph{Equivalence of $\mveo(Q)$ and Minimal Query Plans}
In \cref{SEC:FACTORIZATIONSPACE}, we defined minimal $\veo$s based on a partial order $\preceq$ defined via the size of table prefixes.
We show that this partial order corresponds to the partial dissociation order, and hence we can obtain a direct correspondence between $\mveo(Q)$ and Minimal Query Plans (due to the fact that all $\veo$s correspond to query plans through the bijection explained in \cref{SEC:APPENDIX:VEOQP}).
The variables of table-prefix of $R$ always contain the variables of $R$, plus additional variables that are used in joins with $R$.  
If these ``additional'' variables of each table-prefix are added to the respective relations, then the resulting query is hierarchical since every atom already contains all the variables that are used in joins with them.
In other words, the additional variables in the table-prefixes correspond exactly to a hierarchical dissociation of the query.
We can see then that the partial order on $\veo$s is exactly the same as the partial dissociation order since both look at the variables in the table-prefixes / dissociations.

\subsection{Proof of \cref{thm:minfactminveo}: Construction of Factorization Trees from minimal $\veo$s}

\thmminfactminveo*

\begin{proof}[Proof \cref{thm:minfactminveo}]
	We know from \cref{thm:minfacveo} that a minimal factorization tree can be constructed by assigning $\veo$s to each witness.
	It now remains to be shown that using non-minimal $\veo$s (corresponding to non-minimal query plans) cannot make the factorization smaller than if just minimal $\veo$s are used.
	Assume that the factorization uses a non-minimal query plan P', which corresponds to the dissociation $\Delta'$.
	Since it is not minimal there must exist a dissociation $\Delta$ where $\Delta \preceq \Delta'$. 
	This means that the provenance (of $D$, or simply a subset of witnesses) calculated with $P'$ is a dissociation of the provenance calculated with $P$, which means that each variable occurs at least as many times in $\varphi(P^{\Delta'})$
	as in $\varphi(P^{\Delta})$~\cite{gatterbauer2014oblivious}.
\end{proof}

\subsection{End to End Example: Query Plans, $\veo$s and Factorizations}

\cref{ex:endtoendveostu} gives an end-to-end example that shows a query, its minimal hierarchical dissociations~\cite{DBLP:journals/vldb/GatterbauerS17}, the corresponding minimal query plans, the $\veo$s, and the factorizations that each $\veo$ corresponds to. 
This example is also interesting since it shows minimal $\veo$s that have \emph{multiple variables in a node}; this is not permitted in prior definitions of $\veo$s such as \cite{DBLP:journals/sigmod/OlteanuS16,Dechter:1999:Bucket}.

\begin{example}[Triangle unary $\qtriangleunary$]
	\label{ex:triangleunaryendtoend}
	Consider the triangle unary query $\qtriangleunary \datarule U(x), R(x,y), $ $S(y,z), T(z,x)$.
	Notice that this query is not hierarchical: $\at(y) = \{R, S\}$ and $\at(z) = \{S, T\}$, 
	which is not allowed in a hierarchical query.
	Thus, it does not have a hierarchical plan, and we know that a database instance under this query is not guaranteed to have a read-once factorization.
	
	Thus, to obtain the minimal factorization, we consider all its minimal hierarchical dissociations~\cite{DBLP:journals/vldb/GatterbauerS17}, 
	which correspond to the minimal query plans and minimal $\veo$s.
	This query has 3 such minimal dissociations, shown in \cref{Fig:TriangleUnaryEndToEnd} as incidence matrices, query plans, $\mveo$s, and factorizations.
	The objective is now to assign each witness to exactly one of these 3 minimal dissociations.
	Here the incidence matrix is a simple visual representation of the variables, atoms, and variables contained in each atom.
	Recall the intuition for dissociations of adding variables to atoms.

	Notice that $v_1 = \{y,z\} \!\leftarrow\! x$
	would not be permitted in an FDB-style variable ordering, which requires a single variable at each level of the tree~\cite[Definition 3.2]{DBLP:journals/sigmod/OlteanuS16}.
	Thus it would require $v_1$ to be split into $2$ different query plans $y \!\leftarrow\! z \!\leftarrow\! x$ and $z \!\leftarrow\! y \!\leftarrow\! x$. Our approach takes into account concurrent variable elimination and hence can reduce the search space for minimal factorizations in this case.
	\label{ex:endtoendveostu}
\end{example}

\begin{figure}
	\centering
	\begin{subfigure}[b]{.7\linewidth}
		\centering
		\begin{minipage}[t]{40mm}			
		\centering
		\setlength{\tabcolsep}{1mm}
			\mbox{		
					\begin{tabular}[t]{c | c c c  }
						& x & y	& z\\
				\hline
				R	& \circle{5}	& \circle{5} & \\
				S	& 	& \circle{5}	& \circle{5} \\
				T	& \circle{5} &	& \circle{5} \\
				U	& \circle{5}	&  & 			
				\end{tabular}			
			}
		\end{minipage}
	\caption{$\qtriangleunary$ Incidence Matrix
	}\label{tab:TriangleUnaryIncidence}
	\end{subfigure}
	
	\vspace{5mm}

	\begin{subfigure}[b]{.25\linewidth}
		\centering
		\begin{minipage}[t]{30mm}			
		\centering
		\setlength{\tabcolsep}{0.8mm}
			\mbox{
					\begin{tabular}[t]{c | c c c  }
					 	& x & y	& z\\
					\hline
					R	& \circle{5}	& \circle{5} & \circle*{5}\\
					S	& 	& \circle{5}	& \circle{5} \\
					T	& \circle{5} & \circle*{5}	& \circle{5} \\
					U	& \circle{5}	& \circle*{5} & \circle*{5}			
					\end{tabular}			
			}
		\end{minipage}
	\caption{$\Delta_1$
	}\label{tab:TriangleUnaryDelta1}
	\end{subfigure}
	\begin{subfigure}[b]{.37\linewidth}
		\centering
		\begin{minipage}[t]{30mm}			
		\centering
		\setlength{\tabcolsep}{0.8mm}
			\mbox{
					\begin{tabular}[t]{c | c c c  }
					 	& x & y	& z\\
					\hline
					R	& \circle{5}	& \circle{5} & \\
					S	& \circle*{5}	& \circle{5}	& \circle{5} \\
					T	& \circle{5} & \circle*{5}	& \circle{5} \\
					U	& \circle{5}	&  & 			
					\end{tabular}			
			}
		\end{minipage}
	\caption{$\Delta_2$
	}\label{tab:TriangleUnaryDelta2}
	\end{subfigure}
	\begin{subfigure}[b]{.33\linewidth}
		\centering
		\begin{minipage}[t]{30mm}			
		\centering
		\setlength{\tabcolsep}{0.8mm}
			\mbox{
					\begin{tabular}[t]{c | c c c  }
					 	& x & y	& z\\
					\hline
					R	& \circle{5}	& \circle{5} &  \circle*{5}\\
					S	& \circle*{5}	& \circle{5}	& \circle{5} \\
					T	& \circle{5} &	& \circle{5} \\
					U	& \circle{5}	&  & 			
					\end{tabular}			
			}
		\end{minipage}
	\caption{$\Delta_3$
	}\label{tab:TriangleUnaryDelta3}
	\end{subfigure}

	\vspace{5mm}

	\begin{subfigure}[b]{.28\linewidth}
		\centering
     	\includegraphics[scale=0.6]{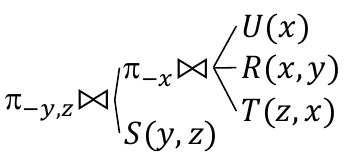}
        \caption{$P_1$}
    \label{fig:TriangleUnaryP1}
	\end{subfigure}
	\begin{subfigure}[b]{.35\linewidth}
		\centering
        \includegraphics[scale=0.6]{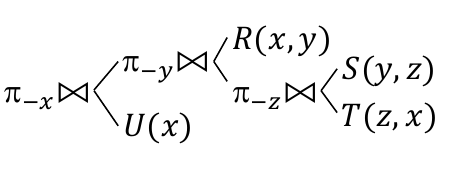}
        \caption{$P_2$}
    \label{fig:TriangleUnaryP2}
	\end{subfigure}
	\begin{subfigure}[b]{.35\linewidth}
		\centering
     	\includegraphics[scale=0.6]{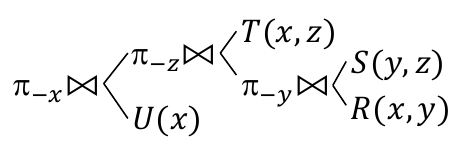}
        \caption{$P_3$}
    \label{fig:TriangleUnaryP3}
	\end{subfigure}

	\vspace{5mm}

	\begin{subfigure}[b]{.25\linewidth}
		\centering
		\begin{tikzpicture}[grow=right,<-, level distance=10mm, baseline=-0.5ex] \node(xy){$yz$} child {node(z){$x$} };
		\draw [-,opacity=0.15, line width =20pt, color=blue, line cap=round] (xy.center) -- (z.center) node [above, opacity=1] {T};
		\draw [-,opacity=0.15, line width =20pt, color=red, line cap=round] (xy.center) -- (z.center) node [above, right, opacity=1] {U};
		\draw [-,opacity=0.15, line width =20pt, color=dg, line cap=round] (xy.center) -- (z.center) node [below, opacity=1] {R};
		\draw [-,opacity=0.5, line width =20pt, color=brown, line cap=round] (xy.center) -- (xy.center) node [above, opacity=1] {S} ;
		\draw [-,opacity=0, line width =40pt, color=orange, line cap=round] (xy.center) -- (xy.center) node [below, opacity=1] {$1$};
		\draw [-,opacity=0, line width =40pt, color=orange, line cap=round] (z.center) -- (z.center) node [below, opacity=1] {$3$};
		\end{tikzpicture}
		\caption{$v_1 = yz \!\leftarrow\! x$} \label{TriU:VEO1}
	\end{subfigure}
	\hspace{1mm}
	\begin{subfigure}[b]{.37\linewidth}
		\centering
		\begin{tikzpicture}[grow=right,<-, level distance=10mm, baseline=-0.5ex] \node(x){$x$} child {node(y){$y$} child {node(z){$z$} }};
		\draw [-,opacity=0.5, line width =20pt, color=brown, line cap=round] (x.center) -- (z.center) node [above, opacity=1] {S};
		\draw [-,opacity=0.15, line width =20pt, color=blue, line cap=round] (x.center) -- (z.center) node [below, opacity=1] {T};
		\draw [-,opacity=0.15, line width =20pt, color=dg, line cap=round] (x.center) -- (y.center) node [above, opacity=1] {R} ;
		\draw [-,opacity=0.15, line width =20pt, color=red, line cap=round] (x.center) -- (x.center) node [above, opacity=1] {U};
		\draw [-,opacity=0, line width =40pt, color=orange, line cap=round] (x.center) -- (x.center) node [below, opacity=1] {$1$};
		\draw [-,opacity=0, line width =40pt, color=orange, line cap=round] (x.center) -- (y.center) node [below, opacity=1] {$1$};
		\draw [-,opacity=0, line width =40pt, color=orange, line cap=round] (x.center) -- (z.center) node [below, opacity=1] {$2$};
	\end{tikzpicture}
		\caption{$v_2 = x \!\leftarrow\! y \!\leftarrow\! z$} \label{TriU:VEO2}
	\end{subfigure}
	\hspace{0mm} 
	\begin{subfigure}[b]{.33\linewidth}
		\centering
		\begin{tikzpicture}[grow=right,<-, level distance=10mm, baseline=-0.5ex] \node(x){$x$} child {node(z){$z$} child {node(y){$y$} }};
			\draw [-,opacity=0.5, line width =20pt, color=brown, line cap=round] (x.center) -- (y.center) node [above, opacity=1] {S};
			\draw [-,opacity=0.15, line width =20pt, color=dg, line cap=round] (x.center) -- (y.center) node [below, opacity=1] {R} ;
			\draw [-,opacity=0.15, line width =20pt, color=blue, line cap=round] (x.center) -- (z.center) node [above, opacity=1] {T};
			\draw [-,opacity=0.15, line width =20pt, color=red, line cap=round] (x.center) -- (x.center) node [above, opacity=1] {U};
			\draw [-,opacity=0, line width =40pt, color=orange, line cap=round] (x.center) -- (x.center) node [below, opacity=1] {$1$};
			\draw [-,opacity=0, line width =40pt, color=orange, line cap=round] (x.center) -- (z.center) node [below, opacity=1] {$1$};
			\draw [-,opacity=0, line width =40pt, color=orange, line cap=round] (x.center) -- (y.center) node [below, opacity=1] {$2$};	
		\end{tikzpicture}
		\caption{$v_3 = x \!\leftarrow\! z \!\leftarrow\! y$}
		\label{TriU:VEO3}
	\end{subfigure}

	\vspace{10mm}

	\begin{subfigure}[b]{.3\columnwidth}
		\centering
		\begin{tikzpicture}[baseline=1ex]
		\node at (0,1.5)  { \small $\bigvee_{\!yz}S(y,z) [\bigvee_{\!x} U(x), R(x,y), T(x,z) ]$};
		\end{tikzpicture}
		\caption{Factorization 1}
		\label{}
	\end{subfigure}
	\begin{subfigure}[b]{.25\columnwidth}
		\centering
		\begin{tikzpicture}[]
		\node at (0,0) { \small $\bigvee_{\!x} U(x) [\bigvee_{\!y} R(x,y) [\bigvee_{\!z} S(y,z), T(x,z) ]]$};
		\end{tikzpicture}
		\caption{Factorization 2}
		\label{}
	\end{subfigure}
	\begin{subfigure}[b]{.4\columnwidth}
		\centering
		\begin{tikzpicture}[baseline=1ex]
		\node at (-3,1.5) {\small $\bigvee_{\!x} U(x) [\bigvee_{\!z} T(x,z) [\bigvee_{\!y} R(x,y), S(y,z) ]]$};
		\end{tikzpicture}
		\caption{Factorization 3}
		\label{}
	\end{subfigure}
	
	\caption{\Cref{ex:triangleunaryendtoend}: 
		For the triangle unary query $\qtriangleunary \datarule U(x), $ $R(x,y), S(y,z), T(z,x)$, 
		we show the incidence graph matrix (a), 
		the three minimal dissociations (b)-(d),
		corresponding three minimal query plans (e)-(g),
		three minimal $\veo$s (h)-(j),
		and factorizations (k)-(m).
	}
\label{Fig:TriangleUnaryEndToEnd}
\end{figure}

\section{Details  for \cref{sec:ilp}: ILP Formulation for $\minfact$}
\label{sec:ilp-examples}

Here we illustrate $\minfactilp$ with examples, including two end-to-end worked-out examples that illustrate the construction of the ILP, and more insights on the value of structural knowledge.

\introparagraph{Size of ILP}
For query $Q$ with 
$m$ relations and
$|\mveo(Q)| = k$, and a set of $n$ witnesses under $Q$: the resulting ILP (without any optimizations) has
$n(1+km)$ constraints.
Thus the size of the ILP is linear in the number of witnesses.

\begin{example}[ILP Decision Variables]
	\label{ex:ILPdecisionvar}
	Consider again the 2-star query 
	$\qtwostar$ 
	with $\mveo(Q)=\{x \!\leftarrow\! y, y \!\leftarrow\! x\}$.
	For a witness $\w = (x_1, y_2)$ we thus have two QPVs:
	$q(x_1 \!\leftarrow\! y_2)$ and
	$q(y_2 \!\leftarrow\! x_1)$.
	The set of prefixes for $v_1 = x \!\leftarrow\! y$ is 
	$\{x, x \!\leftarrow\! y\}$.
	In that set, $x$ is the table prefix for $R(x)$ and $x \!\leftarrow\! y$ the table prefix for both $S(x,y)$ and $T(y)$.
	Thus, the weights for the prefixes are
	$c(x) = 1$ for one table $R$, and
	$c(x \!\leftarrow\! y) = 2$ for two tables $S$ and $T$.
	For that witness $\w$ and $v_1$ we thus get the two prefix variables
	$p(x_1)$ and
	$p(x_1\!\leftarrow\! y_1)$
	with weights for the prefixes
	$c(x) = 1$ and
	$c(x\!\leftarrow\! y) = 2$.
\end{example}

\begin{example}[ILP Objective Function]
	Consider again $\qtwostar$ 
	and two witnesses $W=\{(x_1,y_1),$ $(x_1,y_2)\}$.
	If both witnesses are factorized according to the $\veo$ 
	$x \!\leftarrow\! y$, then the $\veo$ instances used will be encoded by $q(x_1 \!\leftarrow\! y_1)=1$ and $q(x_1 \!\leftarrow\! y_2)=1$. 
	There will be a single table prefix instance of table $R$ under this VEO
	that is assigned true: $v^R_{\w_1} = v^R_{\w_2} = x_1$, 
	and thus $p(x_1)= 1$.
	And there will be two unique table prefix instances for tables $S$ and $T$ each: 
	$v^S_{\w_1} = v^T_{\w_1} = x_1 \!\leftarrow\! y_1$ and $v^S_{\w_2} = v^T_{\w_2} x_1 \!\leftarrow\! y_2$,
	and thus $p(x_1 \!\leftarrow\! y_1)= 1$ and $p(x_1 \!\leftarrow\! y_2)= 1$ each with weight 2.
	Thus, the objective is equal to $5$. 
	This fits with the actual length of the factorization $\texttt{len} = r_1(s_{11} t_1+s_{12} t_2)$.
	\end{example}

\subsection{ILP for 3-Star Query $\qthreestar$}
\label{sec:3star-ilp-example}

This example highlights some optimizations to the ILP.

\begin{example}[ILP formulation for 3-Star query]
\label{ex:3star}
Consider the 3-Star query 
$\qthreestar \datarule R(x), S(y),$ $ T(z), W(x,y,z)$.
This query has $6$ minimal query plans corresponding to the 6 $\veo$s 
shown in \cref{3Star}.
We use these $\veo$s to build the set $\QPV$. 

Note that since all query plans are linear (they all form paths), 
they can alternatively be specified by the first two variables in the $\veo$ (the last one is implied).
We use this shorter notation as it also simplifies constraints later.
Hence, the query plan constraint 
for an example witness $\w_1 = (x_1,y_1,z_1)$ would be:
\begin{align*}
q[x_1 \!\leftarrow\! y_1] +
q[x_1 \!\leftarrow\! z_1] +
q[y_1 \!\leftarrow\! x_1]	& \\ 
\phantom{x} + q[y_1 \!\leftarrow\! z_1] +
q[z_1 \!\leftarrow\! x_1] +
q[z_1 \!\leftarrow\! y_1]  	&\geq 1  \phantom{x}
\end{align*}

Then, we calculate the table-prefixes needed to build set $\PVF$, illustrated here for all $6$ $\veo$s: 

\begin{center}
\renewcommand{\arraystretch}{1.2}
\setlength{\tabcolsep}{1mm}
\begin{tabular}[t]{ >{$}l<{$} | >{$}l<{$}  | >{$}l<{$}  | >{$}l<{$}}
\multicolumn{1}{c |}{VEO $v_1$}				& \multicolumn{1}{c |}{VEO $v_2$}			& \multicolumn{1}{c|}{VEO $v_3$}\\
\hline
\;v_1^R = x									& \;v_2^R	= x                            & \;v_3^R	= x\!\leftarrow\!y \\
\;v_1^S = x \!\leftarrow\!y				& \;v_2^S = x\!\leftarrow\!z\!\leftarrow\!y & \;v_3^S = y	\\
\;v_1^T = x\!\leftarrow\!y\!\leftarrow\!z	& \;v_2^T = x\!\leftarrow\!z				& \;v_3^T = y\!\leftarrow\!x\!\leftarrow\!z\\
v_1^W = x\!\leftarrow\!y\!\leftarrow\!z		& v_2^W = x\!\leftarrow\!z\!\leftarrow\!y	& v_3^W = y\!\leftarrow\!x\!\leftarrow\!z\\	
\end{tabular}	
\end{center}

\begin{center}
\renewcommand{\arraystretch}{1.2}
\setlength{\tabcolsep}{1mm}
\begin{tabular}[t]{ >{$}l<{$} | >{$}l<{$}  | >{$}l<{$}  | >{$}l<{$}}
\multicolumn{1}{c |}{VEO $v_4$}				& \multicolumn{1}{c |}{VEO $v_5$}			& \multicolumn{1}{c|}{VEO $v_6$}\\
\hline
\;v_4^R = y\!\leftarrow\!z\!\leftarrow\!x									& \;v_5^R	= z\!\leftarrow\!x                            & \;v_6^R	= z\!\leftarrow\!y\!\leftarrow\!x\\
\;v_4^S = y					& \;v_5^S = z\!\leftarrow\!x\!\leftarrow\!y & \;v_6^S = z\!\leftarrow\!y	\\
\;v_4^T = y\!\leftarrow\!z	& \;v_5^T = z				& \;v_6^T = z\\
v_4^W = y\!\leftarrow\!z\!\leftarrow\!x		& v_5^W = z\!\leftarrow\!x\!\leftarrow\!y	& v_6^W = z\!\leftarrow\!y\!\leftarrow\!x\\	
\end{tabular}	
\end{center}

Notice that we need to count the paths of length 3 twice, 
once for the table $W$, and once for the last of the unary tables in the respective VEO.
In total, we have $15 = 3 + 6 + 6$ unique prefixes 
($3$ of length 1, $6$ of length 2, $6$ of length 3).

Let $T_i$ be the set of all table prefix instances of length $i$.
The objective can be written as 
$$
\texttt{len} = 
\min \sum_{\pi \in T_1} p(\pi) + 
\sum_{\pi \in T_2} p(\pi) + 
2 \sum_{\pi \in T_3} p(\pi)
$$

We can simplify this objective by observing 
that the prefixes of length 3 are unique to each witness
(this is so as those prefixes contain all variables, and every witness is unique under set semantics).
Thus, the sum of all total orders of length $3$ from the set $\{x,y,z\}$ over all $n$ witnesses is $n$ (one for each witness). 
We can thus replace all projection variables with table-prefix of length $3$ with a constant 
and simplify the objective to 
$$
\texttt{len} = 
\min \sum_{\pi \in T_1} p(\pi) + 
\sum_{\pi \in T_2} p(\pi) + 
2 n
$$

The projection constraints for the witness $w_1=(x_1,y_1,z_1)$ are:
\begin{align*}
    p[x_1] \geq q[x_1 \!\leftarrow\! y_1] + q[x_1 \!\leftarrow\! z_1] \\
    p[y_1] \geq q[y_1 \!\leftarrow\! x_1] + q[y_1 \!\leftarrow\! z_1] \\
    p[z_1] \geq q[z_1 \!\leftarrow\! x_1] + q[z_1 \!\leftarrow\! y_1] 
\end{align*}

In this query, for every witness, there must be $9$ corresponding variables in the objective, $1$ Query Plan constraint, and $3$ Projection constraints.

\end{example}

\begin{figure}
\centering
\begin{subfigure}[b]{.25\columnwidth}
	\centering
	\begin{tikzpicture}[grow=right,<-, level distance=7mm, baseline=-0.5ex] \node(x){$x$} child {node(y){$y$} child {node(z){$z$}}};
	\draw [-,opacity=0.15, line width =10pt, color=dg, line cap=round] (x.center) -- (x.center) node [above, opacity=1] {R} ;
	\draw [-,opacity=0.15, line width =10pt, color=brown, line cap=round] (y.center) -- (y.center) node [above, opacity=1] {S};
	\draw [-,opacity=0.15, line width =10pt, color=blue, line cap=round] (z.center) -- (z.center) node [above, opacity=1] {T};
	\draw [-,opacity=0.15, line width =10pt, color=red, line cap=round] 
		(x.center) -- (y.center) -- (z.center) node [below, opacity=1] {W};
	
	\draw [-,opacity=0, line width =30pt, color=orange, line cap=round] (x.center) -- (x.center) node [below, opacity=1] {$1$};
	\draw [-,opacity=0, line width =30pt, color=orange, line cap=round] (y.center) -- (y.center) node [below, opacity=1] {$1$};
	\draw [-,opacity=0, line width =30pt, color=orange, line cap=round] (z.center) -- (z.center) node [below, opacity=1] {$2$};
	\end{tikzpicture}
	\caption{$v_1 = x \!\leftarrow\! y \!\leftarrow\! z$} \label{3Star:VEO1}
\end{subfigure}
\hspace{3mm}
\begin{subfigure}[b]{.25\columnwidth}
	\centering
	\begin{tikzpicture}[grow=right,<-, level distance=7mm, baseline=-0.5ex] \node(x){$x$} child {node(z){$z$} child {node(y){$y$}}};
	\draw [-,opacity=0.15, line width =10pt, color=dg, line cap=round] (x.center) -- (x.center) node [above, opacity=1] {R} ;
	\draw [-,opacity=0.15, line width =10pt, color=brown, line cap=round] (y.center) -- (y.center) node [above, opacity=1] {S};
	\draw [-,opacity=0.15, line width =10pt, color=blue, line cap=round] (z.center) -- (z.center) node [above, opacity=1] {T};
	\draw [-,opacity=0.15, line width =10pt, color=red, line cap=round] 
		(x.center) -- (z.center) -- (y.center) node [below, opacity=1] {W};
	
	\draw [-,opacity=0, line width =30pt, color=orange, line cap=round] (x.center) -- (x.center) node [below, opacity=1] {$1$};
	\draw [-,opacity=0, line width =30pt, color=orange, line cap=round] (z.center) -- (z.center) node [below, opacity=1] {$1$};
	\draw [-,opacity=0, line width =30pt, color=orange, line cap=round] (y.center) -- (y.center) node [below, opacity=1] {$2$};
	\end{tikzpicture}
	\caption{$v_2 = x \!\leftarrow\! z \!\leftarrow\! y$} \label{3Star:VEO2}
\end{subfigure}
\hspace{3mm}
\begin{subfigure}[b]{.25\columnwidth}
	\centering
	\begin{tikzpicture}[grow=right,<-, level distance=7mm, baseline=-0.5ex] \node(y){$y$} child {node(x){$x$} child {node(z){$z$}}};
	\draw [-,opacity=0.15, line width =10pt, color=dg, line cap=round] (x.center) -- (x.center) node [above, opacity=1] {R} ;
	\draw [-,opacity=0.15, line width =10pt, color=brown, line cap=round] (y.center) -- (y.center) node [above, opacity=1] {S};
	\draw [-,opacity=0.15, line width =10pt, color=blue, line cap=round] (z.center) -- (z.center) node [above, opacity=1] {T};
	\draw [-,opacity=0.15, line width =10pt, color=red, line cap=round] 
		(y.center) -- (x.center) -- (z.center) node [below, opacity=1] {W};
	
	\draw [-,opacity=0, line width =30pt, color=orange, line cap=round] (y.center) -- (y.center) node [below, opacity=1] {$1$};
	\draw [-,opacity=0, line width =30pt, color=orange, line cap=round] (x.center) -- (x.center) node [below, opacity=1] {$1$};
	\draw [-,opacity=0, line width =30pt, color=orange, line cap=round] (z.center) -- (z.center) node [below, opacity=1] {$2$};
	\end{tikzpicture}
	\caption{$v_3 = y \!\leftarrow\! x \!\leftarrow\! z$} \label{3Star:VEO3}
\end{subfigure}

\vspace{5mm}
 
\begin{subfigure}[b]{.25\columnwidth}
	\centering
	\begin{tikzpicture}[grow=right,<-, level distance=7mm, baseline=-0.5ex] \node(y){$y$} child {node(z){$z$} child {node(x){$x$}}};
	\draw [-,opacity=0.15, line width =10pt, color=dg, line cap=round] (x.center) -- (x.center) node [above, opacity=1] {R} ;
	\draw [-,opacity=0.15, line width =10pt, color=brown, line cap=round] (y.center) -- (y.center) node [above, opacity=1] {S};
	\draw [-,opacity=0.15, line width =10pt, color=blue, line cap=round] (z.center) -- (z.center) node [above, opacity=1] {T};
	\draw [-,opacity=0.15, line width =10pt, color=red, line cap=round] 
		(y.center) -- (z.center) -- (x.center) node [below, opacity=1] {W};
	
	\draw [-,opacity=0, line width =30pt, color=orange, line cap=round] (y.center) -- (y.center) node [below, opacity=1] {$1$};
	\draw [-,opacity=0, line width =30pt, color=orange, line cap=round] (z.center) -- (z.center) node [below, opacity=1] {$1$};
	\draw [-,opacity=0, line width =30pt, color=orange, line cap=round] (x.center) -- (x.center) node [below, opacity=1] {$2$};
	\end{tikzpicture}
	\caption{$v_4 = y \!\leftarrow\! z \!\leftarrow\! x$} \label{3Star:VEO4}
\end{subfigure}
\hspace{3mm}
\begin{subfigure}[b]{.25\columnwidth}
	\centering
	\begin{tikzpicture}[grow=right,<-, level distance=7mm, baseline=-0.5ex] \node(z){$z$} child {node(x){$x$} child {node(y){$y$}}};
	\draw [-,opacity=0.15, line width =10pt, color=dg, line cap=round] (x.center) -- (x.center) node [above, opacity=1] {R} ;
	\draw [-,opacity=0.15, line width =10pt, color=brown, line cap=round] (y.center) -- (y.center) node [above, opacity=1] {S};
	\draw [-,opacity=0.15, line width =10pt, color=blue, line cap=round] (z.center) -- (z.center) node [above, opacity=1] {T};
	\draw [-,opacity=0.15, line width =10pt, color=red, line cap=round] 
		(z.center) -- (x.center) -- (y.center) node [below, opacity=1] {W};
	\draw [-,opacity=0, line width =30pt, color=orange, line cap=round] (x.center) -- (x.center) node [below, opacity=1] {$1$};
	\draw [-,opacity=0, line width =30pt, color=orange, line cap=round] (y.center) -- (y.center) node [below, opacity=1] {$2$};
	\draw [-,opacity=0, line width =30pt, color=orange, line cap=round] (z.center) -- (z.center) node [below, opacity=1] {$1$};
	\end{tikzpicture}
	\caption{$v_5 = z \!\leftarrow\! x \!\leftarrow\! y$} \label{3Star:VEO5}
\end{subfigure}
\hspace{3mm}
\begin{subfigure}[b]{.25 \columnwidth}
	\centering
	\begin{tikzpicture}[grow=right,<-, level distance=7mm, baseline=-0.5ex] \node(z){$z$} child {node(y){$y$} child {node(x){$x$}}}; 
	\draw [-,opacity=0.15, line width =10pt, color=dg, line cap=round] (x.center) -- (x.center) node [above, opacity=1] {R} ;
	\draw [-,opacity=0.15, line width =10pt, color=brown, line cap=round] (y.center) -- (y.center) node [above, opacity=1] {S};
	\draw [-,opacity=0.15, line width =10pt, color=blue, line cap=round] (z.center) -- (z.center) node [above, opacity=1] {T};
	\draw [-,opacity=0.15, line width =10pt, color=red, line cap=round] 
		(z.center) -- (y.center) -- (x.center) node [below, opacity=1] {W};
	\draw [-,opacity=0, line width =30pt, color=orange, line cap=round] (x.center) -- (x.center) node [below, opacity=1] {$2$};
	\draw [-,opacity=0, line width =30pt, color=orange, line cap=round] (y.center) -- (y.center) node [below, opacity=1] {$1$};
	\draw [-,opacity=0, line width =30pt, color=orange, line cap=round] (z.center) -- (z.center) node [below, opacity=1] {$1$};
	\end{tikzpicture}
	\caption{$v_6 = z \!\leftarrow\! y \!\leftarrow\! x$} \label{3Star:VEO6}
\end{subfigure}
\caption{\Cref{ex:3star}: $\mveo$ for 3-star query $\qthreestar$.}
\label{3Star}
\end{figure}

\subsection{Correctness of $\minfactilp$}

\thmilpcorrectness*

\begin{proof}[Proof \cref{thm:ilpcorrectness}]
	We can show that $\minfactilp$ always computes the minimal length factorization by proving that:
	(1) A valid assignment of the ILP variables always results in a valid factorization of length equal to the objective value.
	(2) for any valid factorization, there is a valid variable assignment of the ILP.

	From earlier discussion, we know that a factorization can be identified by an assignment of a query plan to each witness.
	This resolves Part (2), since the ILP allows any possible assignment of query plans to witnesses. 
	Part (1) can then be proved by the construction in \cref{sec:ilp}, since the prefix constraints enforce that the length of the factorization takes into account all the prefixes used by a witness and hence is correct. 
\end{proof}

\subsection{Value of structural knowledge}

A key underlying quest in complexity theory and knowledge compilation is determining the 
\emph{value of 
knowledge of the underlying structure} by which a problem decomposes.
Consider an arbitrary $m$-partite positive DNF and the problem of finding its minimal factorization.
What is the value of knowing the query that produced that provenance polynomial?

If we don't have additional information, then each clause (or witness) could be factored in any of the possible factorization orders.
The number of those is equal to the number of labeled rooted trees with $m$ vertices,
which is $m^{m-1}$~\cite{oeis}: \href{https://oeis.org/A000169}{A000169}:
1, 2, 9, 64, 625, 7776, 117649, $\ldots$ 
However, if we know that the DNF is the provenance of a query $Q$ with $k$ variables over a database $D$,
then we only need to consider the 
query-specific number $|\mveo(Q)|$ 
of minimal $\veo$s.

For example, for the $k$-star query $Q^*_k$, 
which has the highest number of minimal $\veo$s (given fixed $k$),
the number of minimal $\veo$s is $k!$. 
Contrast this with knowing the query and just having a $k+1$-partite positive DNF (its provenance has $k+1$ partitions).
The number partitions alone would imply $(k+1)^k$ possible factorizations for each witness.
To illustrate the difference, we contrast the numbers for increasing $k=2, 3, \ldots$:
$\{(k: k! | (k+1)^k) \} = 
\{(2: 2 | 9), (3: 6 | 64), (4: 24 | 625), (5: 120 | 7776), (6: 720 | 117649), \ldots \}$.

\section{Details for \cref{SEC:APPROX}: $\PTIME$ Algorithms}

\subsection{Factorization Flow Graphs with Nested Running-Prefix Orderings}
\label{sec:nestingRP}

RP-orderings (Running-Prefix orderings) are possible for any query with a slight modification to the construction.
We call the modified flow graphs ``nested'' and their orders ``\emph{nested RP-orderings}.''
We show an example before we formalize.

\begin{figure}
	\centering
	\begin{subfigure}[b]{.24\columnwidth}
		\centering
		\begin{tikzpicture}[grow=right,<-, level distance=8mm,sibling distance=10mm, baseline=-0.5ex] \node(x){$x$}
		    child {node(u){$u$}
		      child {node(y){$y$}
    		      child {node(z){$z$}
    		      }
		      }
		      child {node(w){$w$}
		        child {node(v){$v$}
    		    }
		      }
		    };
		    \draw [-,opacity=0.15, line width =10pt, color=blue, line cap=round, label=above:\(z\)] 
				(x.center) -- (y.center) node [left, opacity=1] {};
			\draw [-,opacity=0.15, line width =10pt, color=black, line cap=round, label=above:\(z\)] 
				(x.center) -- (w.center) node [left, opacity=1] {};
		    \draw [-,opacity=0.15, line width =10pt, color=green, line cap=round] 
				(u.center) -- (z.center) node [left, opacity=1] {} ;
		    \draw [-,opacity=0.15, line width =10pt, color=red, line cap=round] 
				(y.center) -- (z.center) node [right, opacity=1] {};
			\draw [-,opacity=0.15, line width =10pt, color=teal, line cap=round, label=above:\(z\)] 
				(w.center) -- (v.center) node [left, opacity=1] {};
			\draw [-,opacity=0.20, line width =10pt, color=yellow, line cap=round, label=above:\(z\)] 
				(u.center) -- (v.center) node [left, opacity=1] {};
			\draw [-,opacity=0, line width =10pt, color=orange, line cap=round] (x.center) node [above, opacity=1] {$1$};
			\draw [-,opacity=0, line width =10pt, color=orange, line cap=round] (u.center) node [above, opacity=1] {$1$};
			\draw [-,opacity=0, line width =10pt, color=orange, line cap=round] (y.center) node [below, opacity=1] {$2$};
			\draw [-,opacity=0, line width =10pt, color=orange, line cap=round] (z.center) node [below, opacity=1] {$3$};
			\draw [-,opacity=0, line width =10pt, color=orange, line cap=round] (w.center) node [above, opacity=1] {$2$};
			\draw [-,opacity=0, line width =10pt, color=orange, line cap=round] (v.center) node [above, opacity=1] {$3$};
		\end{tikzpicture}
		\caption{$v_1$}
		\label{6CycleWE:VEO1}
	\end{subfigure}
	\begin{subfigure}[b]{.24\columnwidth}
		\centering
		\begin{tikzpicture}[grow=right,<-, level distance=8mm,sibling distance=10mm, baseline=-0.5ex] \node(x){$x$}
		    child {node(u){$u$}
		      child {node(y){$y$}
    		      child {node(z){$z$}
    		      }
		      }
		      child {node(v){$v$}
		        child {node(w){$w$}
    		    }
		      }
		    };
		    \draw [-,opacity=0.15, line width =10pt, color=blue, line cap=round, label=above:\(z\)] 
				(x.center) -- (y.center) node [left, opacity=1] {};
			\draw [-,opacity=0.15, line width =10pt, color=black, line cap=round, label=above:\(z\)] 
				(x.center) -- (w.center) node [left, opacity=1] {};
		    \draw [-,opacity=0.15, line width =10pt, color=green, line cap=round] 
				(u.center) -- (z.center) node [left, opacity=1] {} ;
		    \draw [-,opacity=0.15, line width =10pt, color=red, line cap=round] 
				(y.center) -- (z.center) node [right, opacity=1] {};
			\draw [-,opacity=0.15, line width =10pt, color=teal, line cap=round, label=above:\(z\)] 
				(w.center) -- (v.center) node [left, opacity=1] {};
			\draw [-,opacity=0.20, line width =10pt, color=yellow
			, line cap=round, label=above:\(z\)] 
				(u.center) -- (v.center) node [left, opacity=1] {};
			\draw [-,opacity=0, line width =10pt, color=orange, line cap=round] (x.center) node [above, opacity=1] {$1$};
			\draw [-,opacity=0, line width =10pt, color=orange, line cap=round] (u.center) node [above, opacity=1] {$1$};
			\draw [-,opacity=0, line width =10pt, color=orange, line cap=round] (y.center) node [below, opacity=1] {$2$};
			\draw [-,opacity=0, line width =10pt, color=orange, line cap=round] (z.center) node [below, opacity=1] {$3$};
			\draw [-,opacity=0, line width =10pt, color=orange, line cap=round] (v.center) node [above, opacity=1] {$2$};
			\draw [-,opacity=0, line width =10pt, color=orange, line cap=round] (w.center) node [above, opacity=1] {$3$};
		\end{tikzpicture}
		\caption{$v_2$}
		\label{6CycleWE:VEO2}
	\end{subfigure}
	\begin{subfigure}[b]{.24\columnwidth}
		\centering
		\begin{tikzpicture}[grow=right,<-, level distance=8mm, sibling distance=10mm, baseline=-0.5ex] \node(x){$x$}
		    child {node(u){$u$}
		      child {node(z){$z$}
    		      child {node(y){$y$}
    		      }
		      }
		      child {node(v){$v$}
		        child {node(w){$w$}
    		    }
		      }
		    };
		    \draw [-,opacity=0.15, line width =10pt, color=blue, line cap=round, label=above:\(z\)] 
				(x.center) -- (y.center) node [left, opacity=1] {};
			\draw [-,opacity=0.15, line width =10pt, color=black, line cap=round, label=above:\(z\)] 
				(x.center) -- (w.center) node [left, opacity=1] {};
		    \draw [-,opacity=0.15, line width =10pt, color=green, line cap=round] 
				(u.center) -- (z.center) node [left, opacity=1] {} ;
		    \draw [-,opacity=0.15, line width =10pt, color=red, line cap=round] 
				(y.center) -- (z.center) node [right, opacity=1] {};
			\draw [-,opacity=0.15, line width =10pt, color=teal, line cap=round, label=above:\(z\)] 
				(w.center) -- (v.center) node [left, opacity=1] {};
			\draw [-,opacity=0.20, line width =10pt, color=yellow
			, line cap=round, label=above:\(z\)] 
				(u.center) -- (v.center) node [left, opacity=1] {};
			\draw [-,opacity=0, line width =10pt, color=orange, line cap=round] (x.center) node [above, opacity=1] {$1$};
			\draw [-,opacity=0, line width =10pt, color=orange, line cap=round] (u.center) node [above, opacity=1] {$1$};
			\draw [-,opacity=0, line width =10pt, color=orange, line cap=round] (z.center) node [below, opacity=1] {$2$};
			\draw [-,opacity=0, line width =10pt, color=orange, line cap=round] (y.center) node [below, opacity=1] {$3$};
			\draw [-,opacity=0, line width =10pt, color=orange, line cap=round] (v.center) node [above, opacity=1] {$2$};
			\draw [-,opacity=0, line width =10pt, color=orange, line cap=round] (w.center) node [above, opacity=1] {$3$};
		\end{tikzpicture}
		\caption{$v_3$}
		\label{6CycleWE:VEO3}
	\end{subfigure}
	\begin{subfigure}[b]{.24\columnwidth}
		\centering
		\begin{tikzpicture}[grow=right,<-, level distance=8mm, sibling distance=10mm, baseline=-0.5ex] \node(x){$x$}
		    child {node(u){$u$}
		      child {node(z){$z$}
    		      child {node(y){$y$}
    		      }
		      }
		      child {node(w){$w$}
		        child {node(v){$v$}
    		    }
		      }
		    };
		    \draw [-,opacity=0.15, line width =10pt, color=blue, line cap=round, label=above:\(z\)] 
				(x.center) -- (y.center) node [left, opacity=1] {};
			\draw [-,opacity=0.15, line width =10pt, color=black, line cap=round, label=above:\(z\)] 
				(x.center) -- (w.center) node [left, opacity=1] {};
		    \draw [-,opacity=0.15, line width =10pt, color=green, line cap=round] 
				(u.center) -- (z.center) node [left, opacity=1] {} ;
		    \draw [-,opacity=0.15, line width =10pt, color=red, line cap=round] 
				(y.center) -- (z.center) node [right, opacity=1] {};
			\draw [-,opacity=0.15, line width =10pt, color=teal, line cap=round, label=above:\(z\)] 
				(w.center) -- (v.center) node [left, opacity=1] {};
			\draw [-,opacity=0.20, line width =10pt, color=yellow
			, line cap=round, label=above:\(z\)] 
				(u.center) -- (v.center) node [left, opacity=1] {};
			\draw [-,opacity=0, line width =10pt, color=orange, line cap=round] (x.center) node [above, opacity=1] {$1$};
			\draw [-,opacity=0, line width =10pt, color=orange, line cap=round] (u.center) node [above, opacity=1] {$1$};
			\draw [-,opacity=0, line width =10pt, color=orange, line cap=round] (z.center) node [below, opacity=1] {$2$};
			\draw [-,opacity=0, line width =10pt, color=orange, line cap=round] (y.center) node [below, opacity=1] {$3$};
			\draw [-,opacity=0, line width =10pt, color=orange, line cap=round] (v.center) node [above, opacity=1] {$2$};
			\draw [-,opacity=0, line width =10pt, color=orange, line cap=round] (w.center) node [above, opacity=1] {$3$};
		\end{tikzpicture}
		\caption{$v_4$}
		\label{6CycleWE:VEO4}
	\end{subfigure}

	\vspace{1mm}

    \begin{subfigure}[b]{\linewidth}
        \begin{tikzpicture}
        \path[blue, thick, <->] (0,0) edge[] node[ fill=white, anchor=center, pos=0.5] {y} (6.8,0);
        \path[blue, thick, <->] (7.2,0) edge[] node[ fill=white, anchor=center, pos=0.5] {z} (14,0);
        
        \path[blue, thick, <->] (0,-0.5) edge[] node[ fill=white, anchor=center, pos=0.5] {w} (3.9,-0.5);
        \path[blue, thick, <->] (4.1,-0.5) edge[] node[ fill=white, anchor=center, pos=0.5] {v} (10.9,-0.5);
        \path[blue, thick, <->] (11.1,-0.5) edge[] node[ fill=white, anchor=center, pos=0.5] {w} (14,-0.5);
        \end{tikzpicture}
	\end{subfigure}
	\caption{\Cref{ex:6-cycleWE}: 4 $\mveo$ for $\qsixcyclewe$ that prevent an RP-ordering}	
	\label{fig:6CycleWE-mveo}
\end{figure}

\begin{figure}
    \centering
    \includegraphics[scale=0.4]{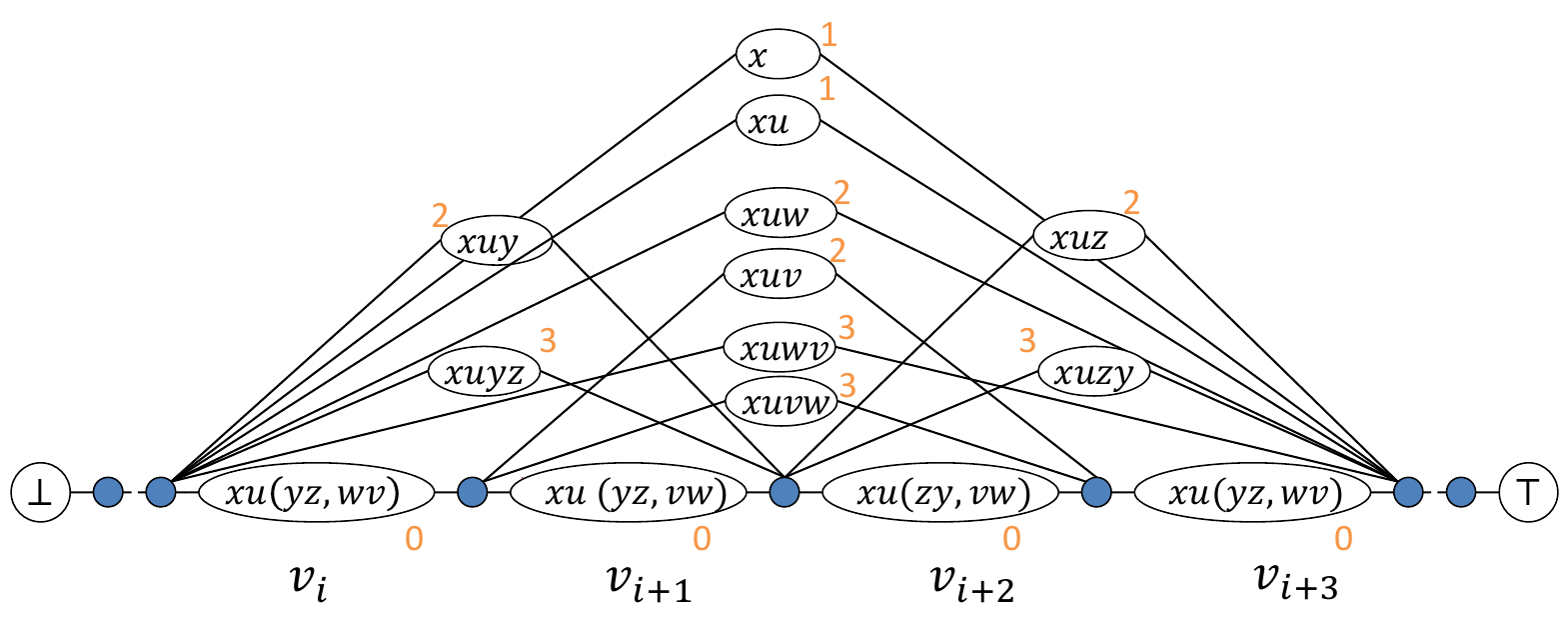}
	\caption{
	\Cref{ex:6-cycleWE} incorrect: Part of factorization flow graph showing $4$ min $\veo$s in Non-RP ordering. 
	Notice that prefixes $xuw$ and $xuwv$ span query variables for which they are not prefixes.}
	\label{fig:6Cycle-NonRP-Flowgraph}
\end{figure}

\begin{figure}
    \centering
    \includegraphics[scale=0.4]{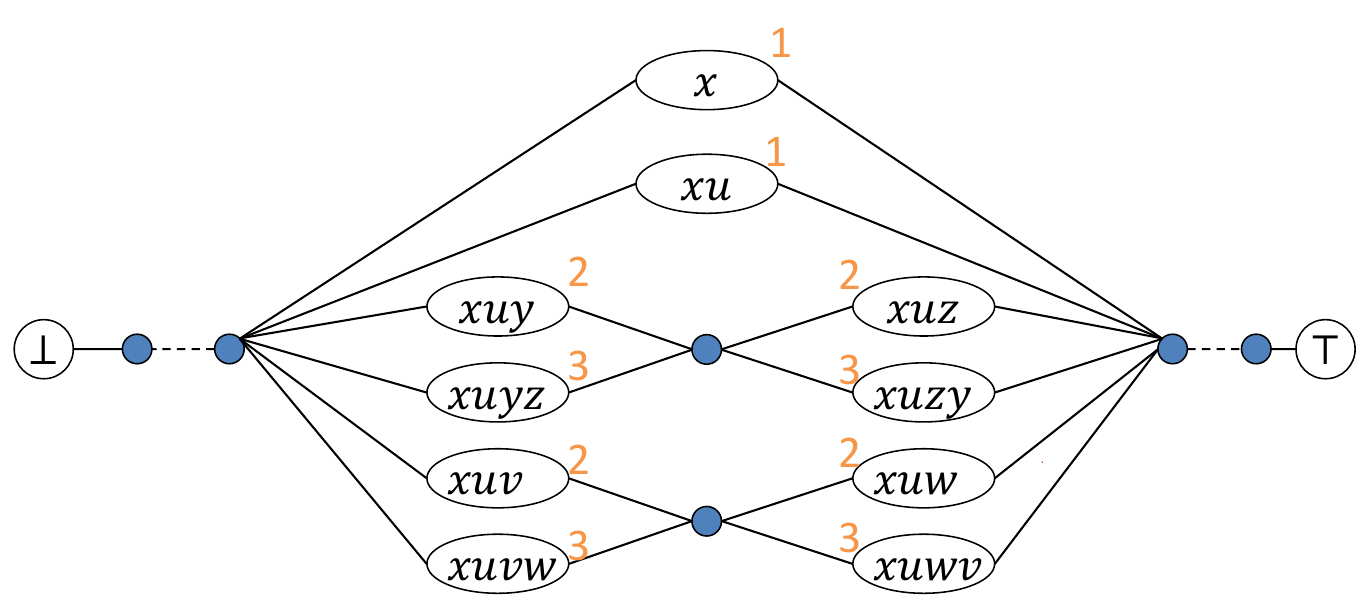}
	\caption{\Cref{ex:6-cycleWE} corrected: Nested RP-ordering that treats each subtree $\{y, z\}$ and $\{v, w\}$ under $\{x, u\}$ as independent subproblem.
}
	\label{fig:6Cycle-Nested-Flowgraph}
\end{figure}

\begin{example}[6-cycle]
	\label{ex:6-cycleWE}
Consider the 6-cycle query with endpoints
shown in \cref{fig:6cyclewe}
and the subset $\{v_1, v_2, v_3, v_4\}$ of $\mveo$ of this query shown in \cref{fig:6CycleWE-mveo}.
Notice that there is no ordering for these 4 $\veo$s to allow the running prefix property. 
For example, if the prefixes of $y$,$z$, and $v$ satisfy the property, 
then we cannot have the prefix of $w$ adjacent as well in a linear ordering. 
\Cref{fig:6Cycle-NonRP-Flowgraph} shows a partial factorization flow graph focused on these four $\veo$s in an arbitrary ordering. 
Notice that the nodes such as $x \!\leftarrow\! u \!\leftarrow\! w $ may be selected even if their corresponding plans are not used.

Our modification to generate an RP-Ordering relies on the observation that the two children of node $a$ in the $\veo$ may be \emph{independently} chosen. 
The left sub-child may be $y\!\leftarrow\!z$ or $z\!\leftarrow\!y$ and the right sub-child can independently be  $v\!\leftarrow\!w$ or $w\!\leftarrow\!v$. 
We encode these independent decisions as parallel paths
and write the nested ordering as  $\Omega = [x \!\leftarrow\! u \!\leftarrow\! ((z\!\leftarrow\!y, y\!\leftarrow\!z) \times (v\!\leftarrow\!w,w\!\leftarrow\!v))]$. 
In the new flow graph \cref{fig:6Cycle-Nested-Flowgraph}, each independent decision is a parallel path in the flow graph. 
Thus, the choice of $xuyz$ vs.\ $xuzy$ is independent of choosing $xuvw$ vs.\ $xuwv$. 
More formally, $\{v_1, v_2, v_3, v_4\}$ are all nestable with each other. 
Thus, we can build a nested ordering where any pair of these are uncomparable.
\Cref{fig:6Cycle-Nested-Flowgraph} shows the factorization flow graph 
for 
$\Omega$.
\end{example}

\subsubsection{Proof \cref{lem:RunningPrefix}}

We prove that for any query $Q$, there is a linear order $\Omega$ on its minimal $\veo$s $\mveo(Q)$ such that it satisfies the RP property:

\thmrpproperty*

\begin{figure}
\begin{subfigure}[b]{.45\linewidth}
	\centering
	\begin{tikzpicture}[scale=0.9,
				every circle node/.style={fill=white, minimum size=7mm, inner sep=0, draw}]
	\node (1) at (-1,0) [] {$x$};	
	\node (2) at (-0.5,0.866) [] {$y$};
	\node (3) at (0.5,0.866) [] {$z$};
	\node (4) at (1,0) [] {$u$};
	\node (5) at (0.5,-0.866) [] {$v$};
	\node (6) at (-0.5,-0.866) [] {$w$};

	\path[->, line width=1pt, auto]
		(6) edge[color=black] (1)
		(1) edge[color=black] (2)
		(2) edge[color=black] (3)	
		(3) edge[color=black] (4)	
		(4) edge[color=black] (5)	
		(5) edge[color=black] (6)
		
		(1) edge [loop left] (1)
	    (2) edge [loop above] (2)
		(3) edge [loop above] (3)
		(4) edge [loop right] (4)
		(5) edge [loop below] (5)
		(6) edge [loop below] (6)
		;
	\end{tikzpicture}
\end{subfigure}
\caption{\Cref{ex:6-cycleWE}: 6-cycle query with endpoints 
$\qsixcyclewe \datarule A(x), R(x,y), B(y), S(y,z), C(z), T(z,u), D(u), U(u,v), E(v),$ $V(v,w), F(w), W(w,x)$.
}
\label{fig:6cyclewe}
\end{figure}

\begin{proof}[Proof \cref{lem:RunningPrefix}]
We prove the theorem by outlining a method of construction of $\Omega$ and showing that this $\Omega$ is guaranteed to satisfy the RP property. 

To construct an RP-Ordering $\Omega$ we first group the $\mveo$ by their root variable $r$. Since the root variables are present in any prefix we know that no two $\veo$s in different groups can share prefixes. Thus arranging the roots in a pre-determined order e.g. a lexicographical order can lead to no violation of the Running Prefixes property. Within a group of $\veo$s with root $r$, we must again decide the ordering. If $r$ disconnects the query incidence matrix, it will have as many children in the $\veo$s as the number of disconnected components. Since $r$ is a min-cut of the query matrix, there will be no node that is shared by its children and the ordering of each component can be decided independently. 
With these new groups of independent prefixes, we again look to order them using the predetermined order, splitting into nested paths when necessary. Hence, we obtain a nested RP-Ordering $\Omega$ for any $Q$.
\end{proof}

\subsubsection{A Flow Graph With Leakage}

The factorization flow graphs are an approximation and can have \emph{leakage} i.e. if no minimal factorization forms a min-cut of the flow graph.

\begin{example}[Leakage in $\qtriangle$ flow graph]
\label{ex:leakage}
Consider a set of 4 witnesses $W = \{r_0s_0t_0, r_1s_1t_0, r_2s_1t_1, $ $r_2s_2t_2\}$
for
$\qtriangle$. 
It has 3 minimal $\veo$s:
$\mveo(\qtriangle) = \{v_1 = xy \!\leftarrow\! z, v_2 = yz \!\leftarrow\! x, v_3 = zx \!\leftarrow\! y\}$.
Notice that any permutation of the $\mveo$ is RP.
We arrange the query plan variables in order 
$\Omega=[v_1, v_2, v_3]$
to construct 
\cref{fig:mincut_triangle_leakage}. 
The min-cut of the graph is $11$.
However, the minimal factorization of $W$ $t_0(r_0s_0 \vee r_1s_1) \vee r_2(s_1t_1 \vee s_2t_2)$ has length $10$ (corresponding to the nodes in blue).
However, this factorization does not cut the graph
because the graph encodes an additional spurious constraint (shown as a red path) that enforces $x_0y_1 + y_1z_0 + z_0x_1 \geq 1$ which need not be true.
\end{example}

\begin{figure}
	\centering
	\includegraphics[scale=0.4]{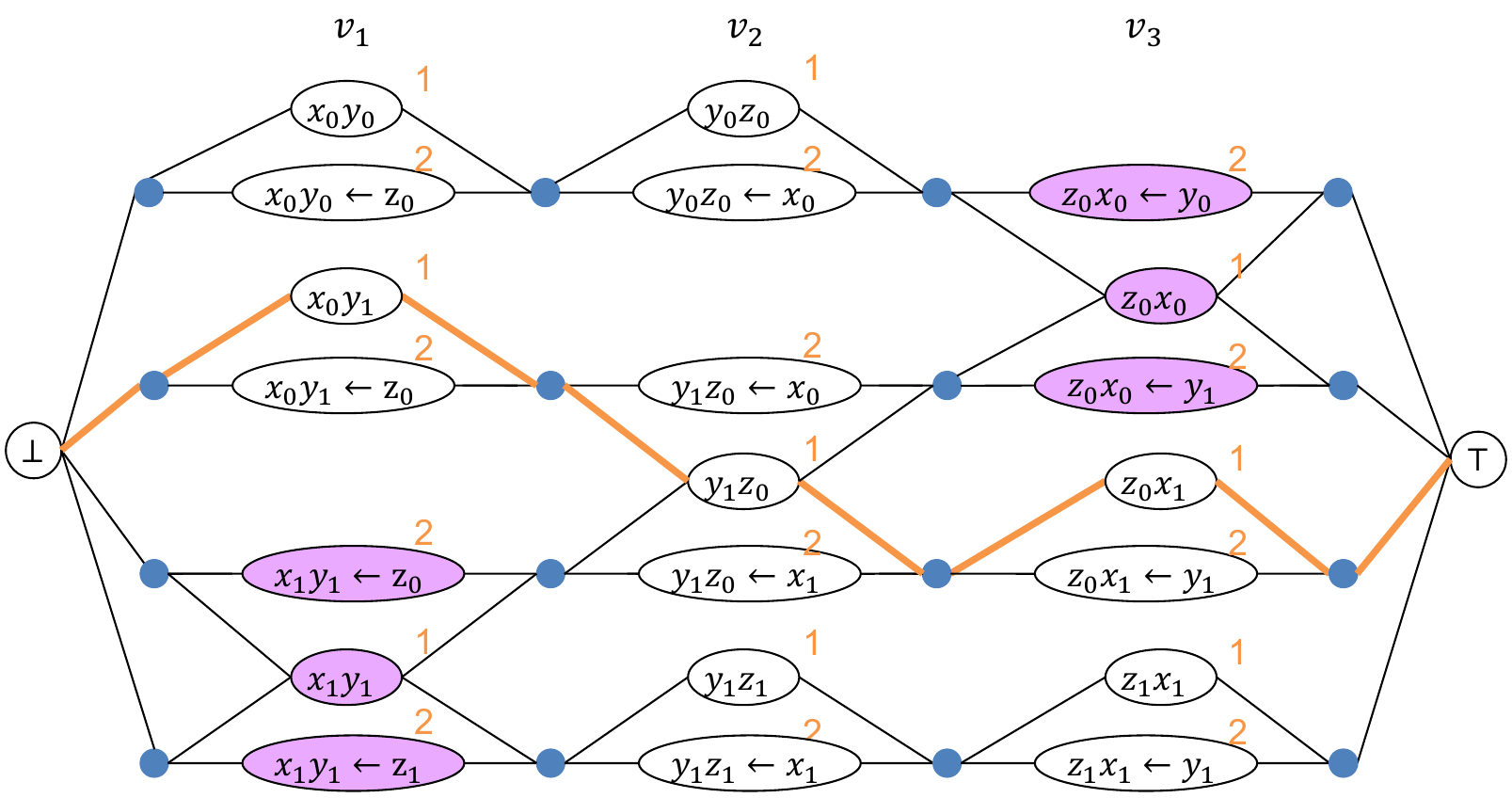}    
	\caption{\Cref{ex:leakage}: A flow graph with leakage for the query $\qtriangle$. 
	}
	\label{fig:mincut_triangle_leakage}
\end{figure}

\subsection{Details for \cref{SEC:LP-ROUNDING}: LP relaxation}

\subsubsection{Details on the Approximation Bound provided by the LP rounding algorithm}

This section details the LP rounding algorithm (\cref{alg:lp-rounding}) and shows a proof of its guaranteed $\mveo$ bound. 
We also show an example comparing it to the natural upper bound of the $DNF$.

\begin{algorithm}

	\SetKwInOut{Input}{Input}
	\Input{List of witnesses $W$ under query $Q$ }
	\KwResult{A factorization for the instance}

	\tcc{Solving a LP:}
	
	Construct the required an ILP for $\minfact(Q)$ over witnesses $W$ \\
	Solve an LP relaxation of the problem by treating all decision variables as continuous- i.e. allowing values [0,1] instead of \{0,1\}. \\
	
	\tcc{LP rounding:}
	Set the value of all prefixes variables to $0$.\\
	\For{$\forall$ $\vec w$ $\in W$}{
		$qp=$ the query plan variable of $\vec w$ with the highest value (break ties arbitrarily) \\
		Set value of $qp = 1$ \\
		Set all other query plan variables of $\vec w$ to $0$ \\
	}
	\For{$\forall$ query plan variables with value $=1$}{
		Set all their prefix variables to $1$.
	}
	
	\caption{LP rounding Algorithm}
	\label{alg:lp-rounding}
	\end{algorithm}

The bound provided by $\cref{alg:lp-rounding}$ is the number of minimal query plans, which may be exponential in query size. 
However, we argue that this simple algorithm provides a useful bound. 
Let us compare it to the trivial upper bound provided by the $DNF$ provenance.
Consider an arbitrary $k$-chain with intermediates query $\datarule R_1(x_1, z_1, x_2), R_2(x_2, z_2, x_3), \hdots R_k(x_k, z_k, x_{k+1})$.
Assume that all variables have the same domain $d$.
A read-once factorization would have $k*d$ literals, while the DNF could have upto $d^k$ literals, with a worst-case approximation of $d^{k-1}/k$.
On the other hand, the number of minimal query plans for a $k$-chain query are given by the Catalan numbers\footnote{OEIS sequence \href{https://oeis.org/A000108}{A000108}} \cite{DBLP:journals/vldb/GatterbauerS17}.
The approximation bound does not depend on domain size and gets better in comparison to the $DNF$ as domain size is increased.
However, even for a small domain size of $10$, all chains up to at least k=$30$ have a better approximation bound with the LP rounding.

\thmlprounding*

\begin{proof}[Proof \cref{thm:lp-rounding}]
	Since for each witness, a valid query plan as well as all associated prefixes are chosen, we see that the approximation algorithm returns a valid factorization.
	
	We see that each query plan is multiplied by at most $|\mveo|$ the value it had in the LP.
	A prefix is chosen only if the associated query plan has a value of at least $1/ |\mveo|$ in the LP. 
	Due to prefix constraints, any chosen prefix too must have had to have a value of at least $1/ |\mveo|$ in the LP.
	Thus all variables are multiplied by at most $|\mveo|$ and we can guarantee the algorithm is an $|\mveo|$-factor approximation.
\end{proof}

\subsubsection{Proof of \cref{prop:lp-easy-mincut}}

\lemlpeasywhenmincuteasy*

\begin{proof}[Proof \cref{prop:lp-easy-mincut}]
We prove that, for such a query, the LP relaxation of the $\minfactilp$ always gives the correct value.
Concretely, we show that the solution of the LP is the same as that of the MFMC-based algorithm in any leakage-free ordering in \cref{proof:tu-flow-optimal}.

First, we notice that the LP relaxation can only provide a lower bound for $\minfact$, while we have the assurance that the MFMC-based algorithm computes the exact $\minfact$ value. 

Next, we show that any solution obtained by the LP must satisfy the MFMC algorithm in a leakage-free ordering. 
We know that the LP solution must satisfy the query plan and prefix constraints, and thus the corresponding paths are cut in the factorization flow graph.
The only remaining paths can only be leakage paths. 
However, since we assumed this to be a leakage-free ordering, 
any additional leakage paths are guaranteed not to increase the value of the min-cut. 
Thus, due to the leakage-free nature, we need only check that the paths corresponding to query plan paths and prefix paths are cut for a valid cut.
In other words, a solution that cuts all query plan and prefix paths must have an equivalent solution that cuts all paths.
Thus there is no path (or ``constraint'') that is not cut (i.e.\ not satisfied) by the LP solution, 
and the LP solution can be represented as a cut of the graph.

Since we have shown that the LP solution is both no greater than and no smaller than the min-cut solution, 
we have also shown that they are identical (in the presence of a leakage-free ordering).
\end{proof}

\section{Proofs for \cref{SEC:READONCE}: Read-once formulas}
\label{sec:appendix:read-once}

\introparagraph{Details on read-once instances}
A formula is read-once if it does not contain any repeated variables. 
Read-once provenance formulas are characterized by the absence of a $P_4$~\cite{CramaHammer2010:BooleanFunctions} co-occurrence pattern,
which is a pattern $(w_1-r-w_2-s-w_3)$ where a witness $w_2$ shares tuple $r$ with $w_1$, and shares tuple $s$ with $w_3$, 
however, $w_1$ and $w_3$ do not share $r$ nor $s$ (\cref{fig:P4}).\footnote{$P_4$ patterns are alternatively defined based on the 4 variables that form a path like ($t-r-s-u)$ in \cref{fig:P4}. The second property of ``normality'' \cite[Sect.~10]{CramaHammer2010:BooleanFunctions} is always fulfilled by partitioned expressions such as the provenance for sj-free queries.} 
We prove that for any read-once formula for any query
the $F$ under any RP-Ordering does not have leakage.

\begin{figure}
    \centering
	\tikzset{>=latex}
    \begin{tikzpicture}[thick]
    \draw (0,0) ellipse (1.5cm and 0.3cm);
    \draw (2,0) ellipse (1.5cm and 0.3cm);
    \draw (4,0) ellipse (1.5cm and 0.3cm);
    \draw (-1,0) node {$t$};
    \draw (1,0) node {$r$};
    \draw (3,0) node {$s$};
    \draw (5,0) node {$u$};
    \draw (0,-0.75) node {$\w_{1}$};
    \draw (2,-0.75) node {$\w_{2}$};
    \draw (4,-0.75) node {$\w_{3}$};
    \end{tikzpicture}
\vspace{-2mm}	
\caption{\Cref{sec:appendix:read-once}: 
$P_4$ co-occurrence pattern $(w_1 \!-\! r \!-\! w_2 \!-\! s \!-\! w_3)$.
}
\label{fig:P4}
\end{figure}

\thmreadonce*

We first prove \cref{prop:readonce:Mincut} and then \cref{prop:readonce:LP}.

\begin{proof}[Proof \cref{prop:readonce:Mincut}]
Consider any RP or nested RP-Ordering $\Omega$.
Notice that the factorization flow graph $F$ represents a database with a $P_4$
if and only if we have three witnesses connected as shown in \cref{fig:flow-graph-P4}.

Let us next assume that we have no $P_4$, yet have leakage in $F$ involving $w_1$ and $w_2$. 
For this to be true, $w_1$ and $w_2$ must share a node $t$ that is not chosen in the minimal factorization. 
It is usually not beneficial to choose a node that belongs exclusively to $w_1$ or $w_2$ over a shared prefix $t$. 

Thus, a shared prefix $t$ is \emph{not} chosen only if there is some other prefix that is shared between witnesses. 
Then there are exactly two possibilities: either this additional shared prefix is 1) shared with one of $w_1$ or $w_2$, or 2) shared with both $w_1$ and $w_2$.
We explore both these cases:
\begin{enumerate}

\item One of the witnesses shares a different prefix with a different witness $w_3$.
This scenario describes a $P_4$, and hence is not possible in a read-once instance.

\begin{figure}
    \centering
    \includegraphics[scale=0.4]{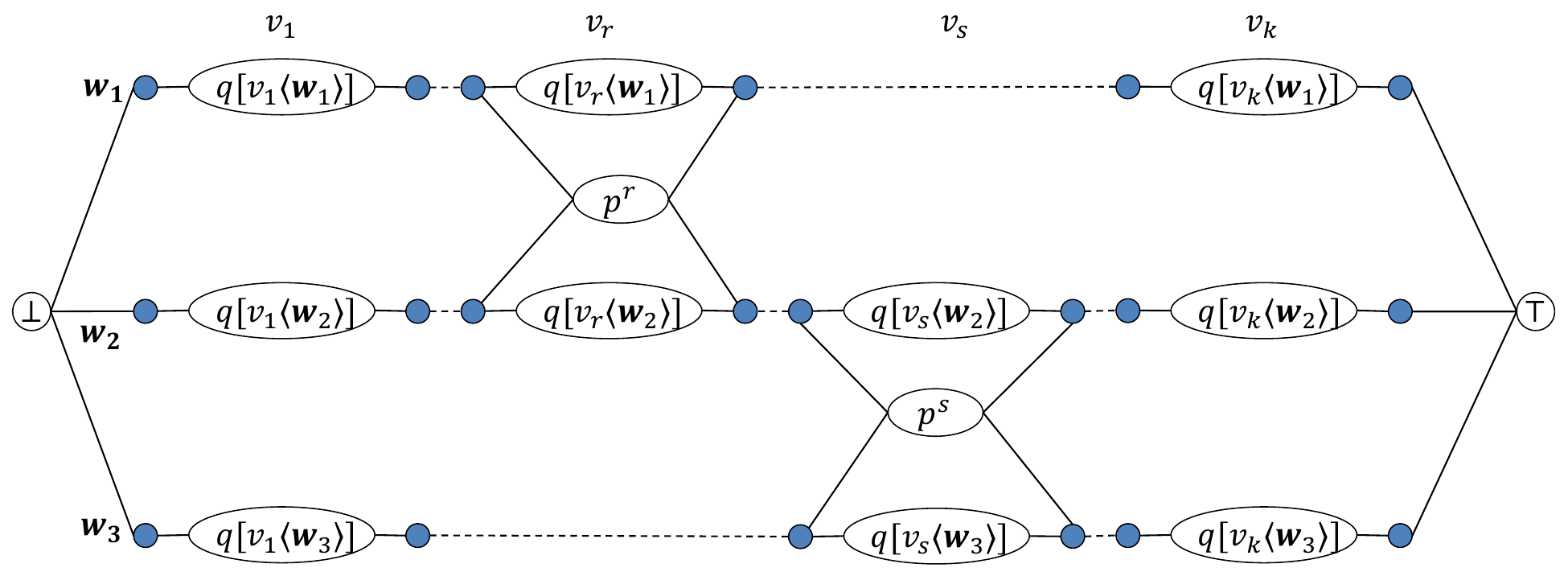}
    \caption{Proof \cref{prop:read-ONCE}: $P_4 (w_1-p_r-w_2-p_s-w_3)$ in factorization flow graph.
	}
    \label{fig:flow-graph-P4}
\end{figure}

\begin{figure}
\centering
\includegraphics[scale=0.4]{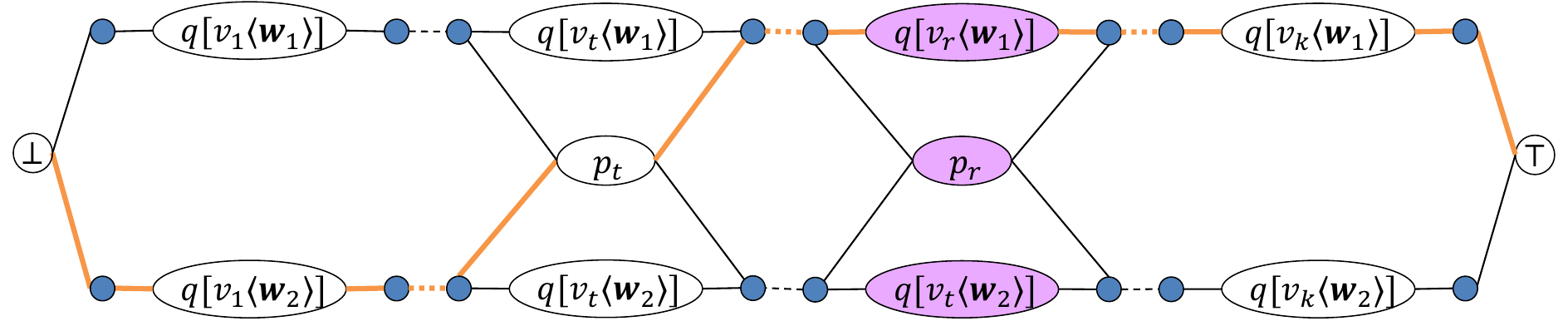}
\caption{Proof \cref{prop:read-ONCE}: Impossibility of leakage in read-once Instances (Case 2).}
\label{fig:flow-graph-2-shared-nodes}
\end{figure}
    
\item The two witnesses $w_1$ and $w_2$ share a different node $r$ as well, and that node is chosen for the minimal factorization.
W.l.o.g, consider that the chosen $\veo$ $v_r$ lies to the right of shared $v_t$ node 
under the ordering $\Omega$ (\cref{fig:flow-graph-2-shared-nodes}). 
Any leakage path that passes through $v_t$ must pass through $v_r$ or a prefix of $v_r$ as well. Since these nodes must be selected, the path cannot "leak".
\end{enumerate}

Thus, we can never have a leakage path for any read-once instance and the factorization flow graph returns the minimal factorization.
In more general terms, our flow graph also correctly identifies the minimal factorization of all known tractable cases of exact probabilistic inference (see \cref{fig:connectionsWithPDBs})
\end{proof}

\begin{proof}[Proof \cref{prop:readonce:LP}]
	We show that the LP relaxation of the $\minfact$ ILP is always correct for read-once instances by showing that the solution of the LP is the same as that of the MFMC-based algorithm (which we proved is correct in \cref{prop:read-ONCE}).
	First, we see that any solution obtained by the MFMC-based algorithm is a valid solution for the LP. 
	Since each path in the factorization flow graph is cut, any solution fulfills all the query plan as well as prefix constraints, thus becoming a valid solution for the LP and thus an upper bound for the optimal LP solution.

	We can next see that the reverse is also true, that a valid LP solution also is a cut for the factorization flow graph. We showed in \cref{prop:read-ONCE} that a read-once instance has no leakage path, therefore each path of the flow graph corresponds exactly to a constraint in the LP. Since the flow graph imposes no additional constraints (``paths''), there is no constraint (``path'') that is not satisfied (not cut) by the LP solution.

	Hence, since valid LP solutions and valid min-cuts have an equivalence for read-once instances, the Linear Program is guaranteed to find the read-once factorization as well. 
\end{proof}

\section{Proofs for \cref{sec:easycase}: Tractable Cases}
\label{sec:appendix:easy-case}

\subsection{All queries with $\leq$2 minimal query plans}

\thmtwoqueryplans*

We first prove \cref{prop:2queryplans:Mincut} and then \cref{prop:2queryplans:LP}. \cref{prop:2queryplans:LP} alternatively follows from \cref{prop:lp-easy-mincut} given \cref{prop:2queryplans:Mincut}.

\begin{proof}[Proof \cref{prop:2queryplans:Mincut}]
Consider an arbitrary query $Q$ with $\mveo=\{v_1, v_2\}$. There are two possible orderings of $\QPV$. 
We can trivially see they are both Running-Prefixes ordering since both query nodes either share prefixes or do not. 
For there to be leakage in the flow graph, there must be a path that connects two cut nodes that are not associated with the same witness. 
Since the query has two minimal plans, every path from source to target passes through at most two cut nodes. 
The two nodes in the path are connected if and only if they are associated with a common witness. 
Thus, no leakage is possible in this graph.

Since the flow graph for any query with $\mveo$ of size $2$ always has an RP ordering and can never have leakage, 
our algorithm is guaranteed to always return the minimal factorization.
\end{proof}

\begin{proof}[Proof \cref{prop:2queryplans:LP}]
\label{proof:TUM}
If a query has $2$ min-$\veo$s, we first show that if it has "equal roots" then the constraint matrix for the ILP generated by this query is Totally Unimodular, 
and then show how this applies to any query with $2$ minimal plans.

\begin{definition}[VEO roots]
The root of a $\veo$ is the set of variables in the $\veo$ that are projected away last in the corresponding query plan. 
\end{definition}

Notice that the roots of two $\veo$ are considered equal if and only if the sets are identical.

\begin{example}[VEO roots]
    The $\veo$s $x\!\leftarrow\!y\!\leftarrow\!z$ 
	and 
	$x\!\leftarrow\!z\!\leftarrow\!y$ have equal root
	$\{x\}$, 
	while $xy\!\leftarrow\!z$ and $xz\!\leftarrow\!y$ do not. 
\end{example}

\begin{proposition}[TU of two-plan queries]
	\label{prop:TUM}
Let $Q$ be a query with precisely two minimal query plans $\{v_1, v_2\} = \mveo(Q)$ 
where $v_1$ and $v_2$ have unequal roots. 
Let $C$ be the constraint matrix of an ILP to find the minimal factorization of any instance with $Q$ constructed as defined in \cref{eqn:ilp}. 
Then $C$ is a Totally Unimodular (TU) matrix.
\end{proposition}

\begin{proof}[Proof \cref{prop:TUM}]
The following \cref{thm:tum-conditions} describes a set of sufficient conditions for Total Unimodularity. We show that $C^T$ meets all conditions. Since the transpose of a TU matrix is also TU, we hence prove that $C$ is a TU matrix as well.

\begin{lemma}[Total Unimodularity \cite{14AnExtensionofaTheoremofDantzigs}]
\label{thm:tum-conditions}
Let $A$  be an $m$ by $n$ matrix whose rows can be partitioned into two disjoint sets $B$ and $C$, with the following properties:
\begin{enumerate}
    \item every entry in A is $0$, $+1$, or $-1$
    \item every column of $A$ contains at most two non-zero entries;
    \item if two non-zero entries in a column of A have the same sign, then the row of one is in $B$, and the other in $C$
    \item if two non-zero entries in a column of $A$ have opposite signs, then the rows of both are in $B$, or both in $C$.
\end{enumerate}
Then every minor determinant of $A$ is $0$, $+1$, or $-1$ i.e. $A$ is a totally unimodular matrix.
\end{lemma}

Each column in $C^T$ represents an ILP constraint and each row represents a decision variable.

\textbf{Condition 1}: Recall the query plan constraint, projection constraint, and integrity constraints. For any variable and any constraint, the coefficient is $1$ if the variable is on the LHS of the constraint, $-1$ if it is on the RHS, and $0$ if it does not participate in the constraint. The coefficient cannot be any other value. 

\textbf{Condition 2}: The query plan constraint enforces a relationship between the variables of each min-VEO. In the case of queries with $2$ plans, this constraint involves $2$ variables. The project constraints always enforce a relationship between $2$ variables- i.e. a query plan instance and a variable-prefix instance, while the integrity constraint always involves a single variable. Hence, there are never more than $2$ non-zero coefficients in a single constraint (and hence, a column of $C^T$).

\textbf{Condition 3}: The only type of constraint that has two coefficients of the same non-zero value is the query plan constraint. 
Each constraint involves a variable for $\veo$ $v_1$, and another for $\veo$ $v_2$. 
We assign all variables corresponding to $v_1$ to set $B$ and those corresponding to $v_2$ to set $C$ to satisfy this property.

\textbf{Condition 4}: The only type of constraint that has two coefficients of the different non-zero values is the projection constraint. The constraint involves a query variable and a projection variable. To satisfy condition 3 we have already assigned the query variable to either set $B$ or $C$. We now assign the projection variable to the same set to which the query variable is assigned. A projection variable is constrained to a query variable if and only if the projection variable has the same table-prefix instances as the query variable. Since $v_1$ and $v_2$ have different roots, no projection variable can be constrained to both the query variables. Hence, every variable is assigned to only one of the two sets, and we fulfill all required properties.
\end{proof}

We can now use the proven-$\PTIME$ case for $2$ $\veo$s with unequal roots to also solve the case where the roots are equal. For query $Q$ with min-$\veo$s $v_1, v_2$ such that they have equal roots, we first check which tables have the same table-prefix in both the $\veo$s. The tuples from these tables can be minimally factorized only in a pre-ordained since they have a single table-prefix. We separately pre-compute the length due to tuples from these tables and remove these table-prefixes from the set $PVF$. After this simplification, no table-prefix will be associated with both $\veo$s. We can now satisfy Condition $4$ by dividing projection variables like in the unequal-root case. Recall that condition $4$ was the only condition where we required the unequal roots property. Thus, we prove that the simplified matrix is TU, and we can obtain $\PTIME$ solutions for all queries with $2$ minimal query plans.
\end{proof}

\subsection{Two queries with $\geq$ 3 minimal query plans}

\subsubsection{Triangle-unary $\qtriangleunary$ is easy}
\label{sec:TU-appendix}

We next prove that $\minfact$ for a query that we call ``Triangle Unary''
$\qtriangleunary \datarule U(x),R(x,y),S(y,z),$ $T(z,x)$ is in $\PTIME$.
This is surprising since
the query has $3$ minimal $\veo$s 
$\mveo =\{x \!\leftarrow\! y \!\leftarrow\! z, x \!\leftarrow\! z \!\leftarrow\! y, yz \!\leftarrow\! x\}$
and appears very similar to the triangle query 
$Q_{\triangle} \datarule R(x,y),S(y,z),T(z,x)$ which 
also has $3$ minimal $\veo$s 
$\mveo=\{xy \!\leftarrow\! z, xz \!\leftarrow\! y, yz \!\leftarrow\! x\}$.
Both queries are illustrated in \cref{fig:triangle:variants}.
However $Q_{\triangle}$ contains a triad, which proves $\fact(Q)$ to be $\npc$,
whereas $\qtriangleunary$ does not have a triad (neither $R(x,y)$ nor $T(z,x)$ are independent due to the relation $U(x)$).
Furthermore, the constraint matrix of the $\qtriangleunary$'s ILP is \emph{not guaranteed to be Totally Unimodular (TU)},
thus our prior TU argument cannot be used anymore for it to be in $\PTIME$.
Then we can use \cref{prop:lp-easy-mincut} to show that the LP relaxation is always easy as well.
\thmtriangleunary*

\begin{proof}[\cref{prop:triangleUNARY:Mincut}]

\cref{fig:TU-one-witness} shows the flow graph for a single witness under order $\Omega = [x \leftarrow y \leftarrow z,x \leftarrow z \leftarrow x,yz \leftarrow x]$. 
However, this proof can be applied to any $\Omega$ of $\mveo(\qtriangleunary)$ that satisfies the Running Prefix property. 
Notice that no matter what $\mveo$ is chosen, a weight of at least $2$ is assigned to the full witness. 
We can preprocess the length due to these nodes as $2*|W|$ and remove the nodes that correspond to the full witness from the graph. 
In addition, we rename the $\mveo$ to names based on the tables that distinguish them i.e.\, $\Omega$ can now be represented as $UR, UT, S\!\leftarrow\!U$. 
The optimized graph for a single witness is shown in \cref{fig:TU-one-witness-simplified}. 
Since all the weights are $=1$, we leave out the weights in subsequent figures.

\begin{figure}
    \centering
    \includegraphics[scale=0.4]{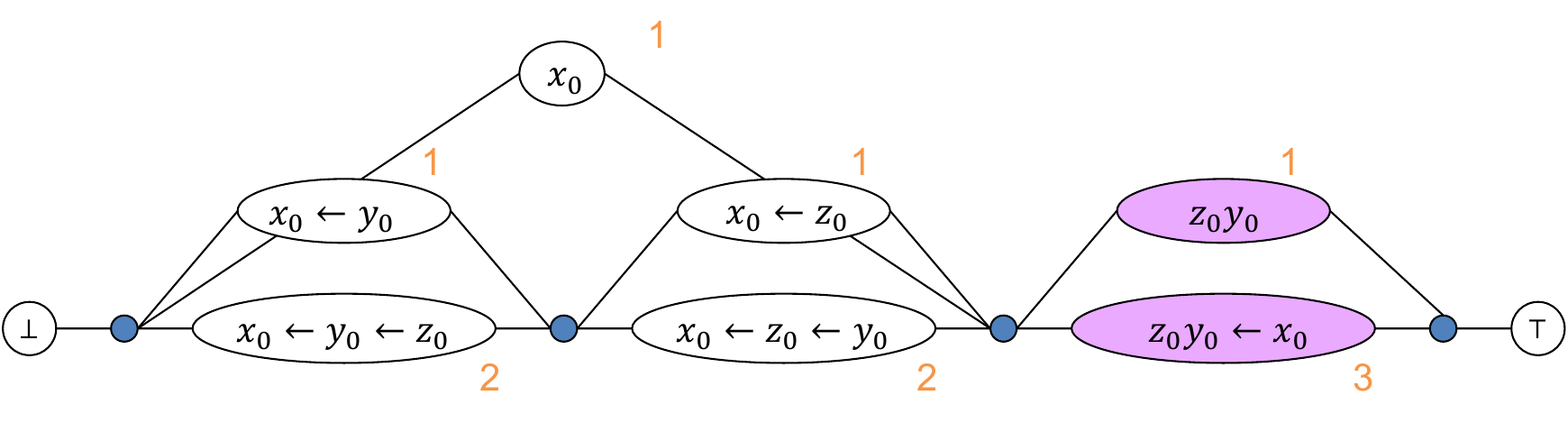}
    \caption{$\qtriangleunary$ graph for a single witness}
    \label{fig:TU-one-witness}
\end{figure}

\begin{figure}
    \centering
    \includegraphics[scale=0.4]{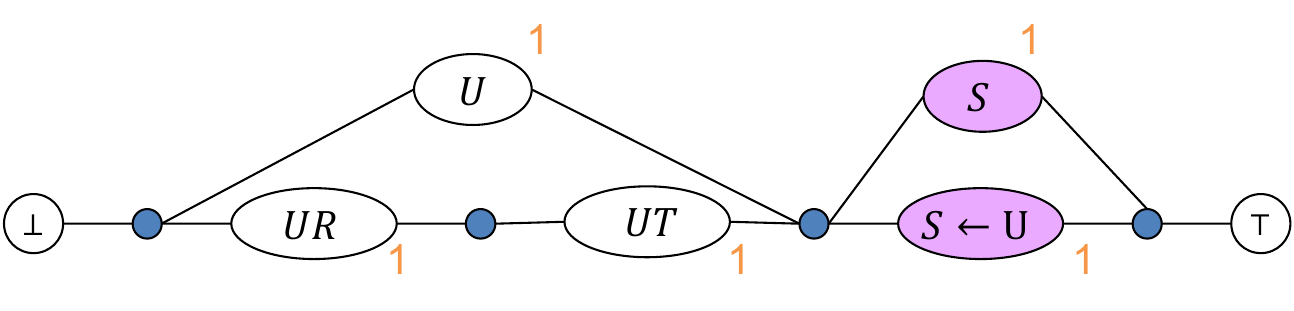}
    \caption{Simplified $\qtriangleunary$ graph for a single witness
	}
    \label{fig:TU-one-witness-simplified}
\end{figure}

We see from the graph that $\Omega$ is a Running-Prefixes Ordering, as the only shared prefix $U$, is shared between adjacent nodes.

We now show that $F(Q, \Omega)$ is leakage-free for any $W$. 
A leakage path must involve at least two witnesses since a path with nodes from one witness only will be cut in any minimal factorization. Let $\w_1$ and $\w_2$ be two witnesses whose nodes lie on the leakage path.

We divide the proof into $7$ parts based on the relationship between $\w_1$ and $\w_2$. If $\w_1$ and $\w_2$ share:
\begin{enumerate}
    \item \emph{No common variable values}:
    There will be no common nodes between the two witnesses, and leakage is not possible.
    
    \item \emph{Only $y$ or $z$ values}:    
    There are still no common nodes between the two witnesses (since there are no $y$ or $z$ nodes), and hence leakage is not possible.
    
    \item \emph{Only $x$ values}:
    In this case (\cref{fig:TU-shared-x}), the $U$ node is shared.
	For the shared node to be on the leakage path, it is not selected in a minimal factorization - thus the witnesses may not choose any $UR$ or $UT$ plans -> thus the witnesses must choose the $S \!\leftarrow\! U$ plans. 
	Since all paths through $U$ pass through $S$ or $S \!\leftarrow\! U$ nodes, leakage is not possible.
    \begin{figure}
        \centering
        \includegraphics[scale=0.4]{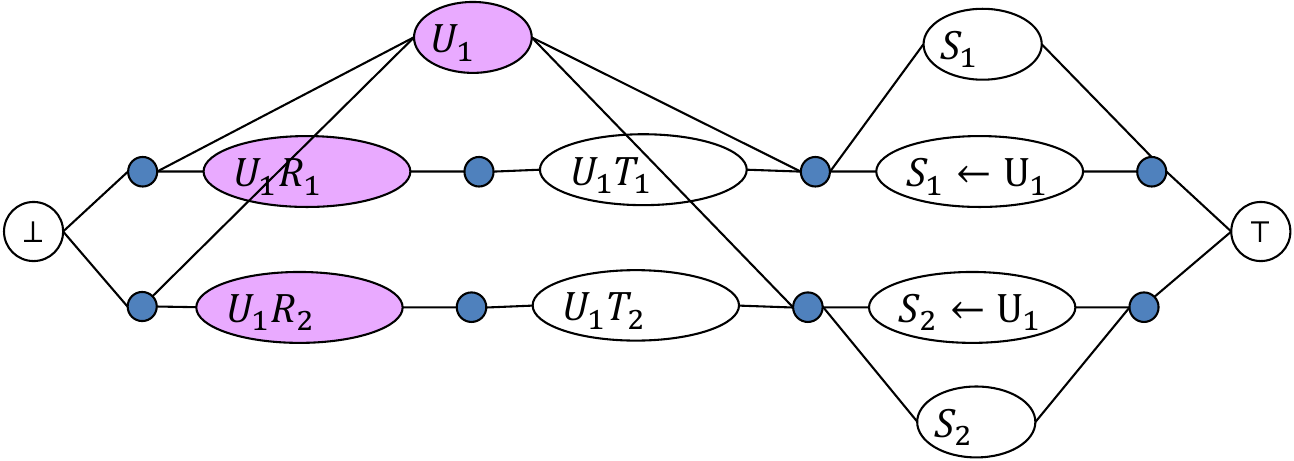}
        \caption{$\qtriangleunary$ instance, with shared $x$ (Case 3)}
        \label{fig:TU-shared-x}
    \end{figure}
    
    \item \emph{$x$ and $y$ values}:    
    In this case (\cref{fig:TU-shared-xy}), we can see that since the shared node $UR$ is at the side of the graph, we see the paths passing through $UR$ must pass through nodes corresponding with the same witness only (either $\w_1$ or $\w_2$). Thus, no leakage is possible.
    \begin{figure}
        \centering
        \includegraphics[scale=0.4]{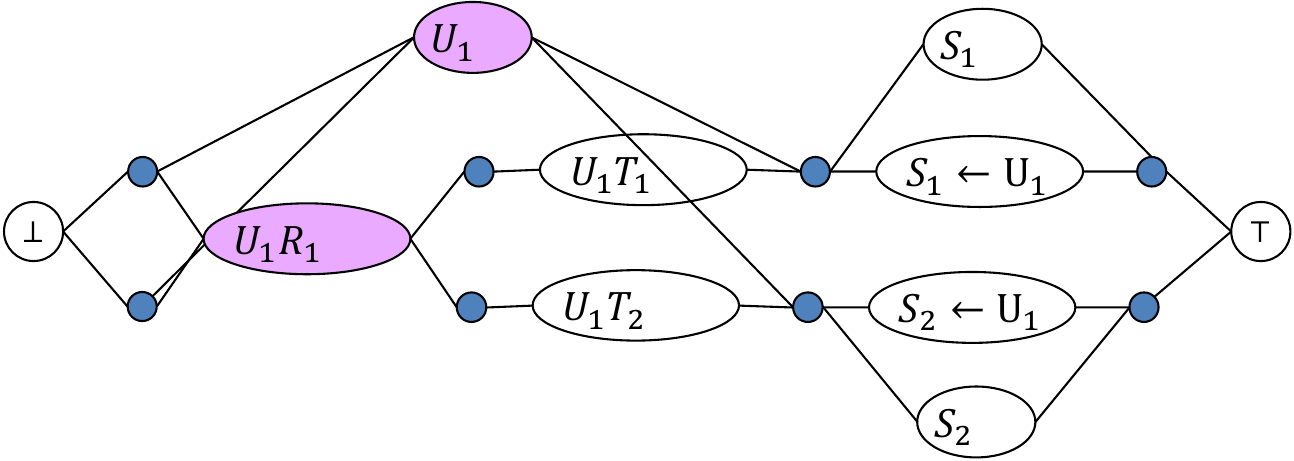}
        \caption{$\qtriangleunary$ instance, with shared $xy$ (Case 4)}
        \label{fig:TU-shared-xy}
    \end{figure}
    
    \item \emph{$x$ and $z$ values}:    
    In this case (\cref{fig:TU-shared-xz}), a leakage is possible if and only if one witness uses the $UR$ plan while the other uses the $S \!\leftarrow\! U$ plan. 
	However, if an $R$ or $T$ tuple is repeated, then the minimal factorization will never use the $S \!\leftarrow\! U$ plan, thus preventing the only potential leakage path.
    \begin{figure}
        \centering
        \includegraphics[scale=0.4]{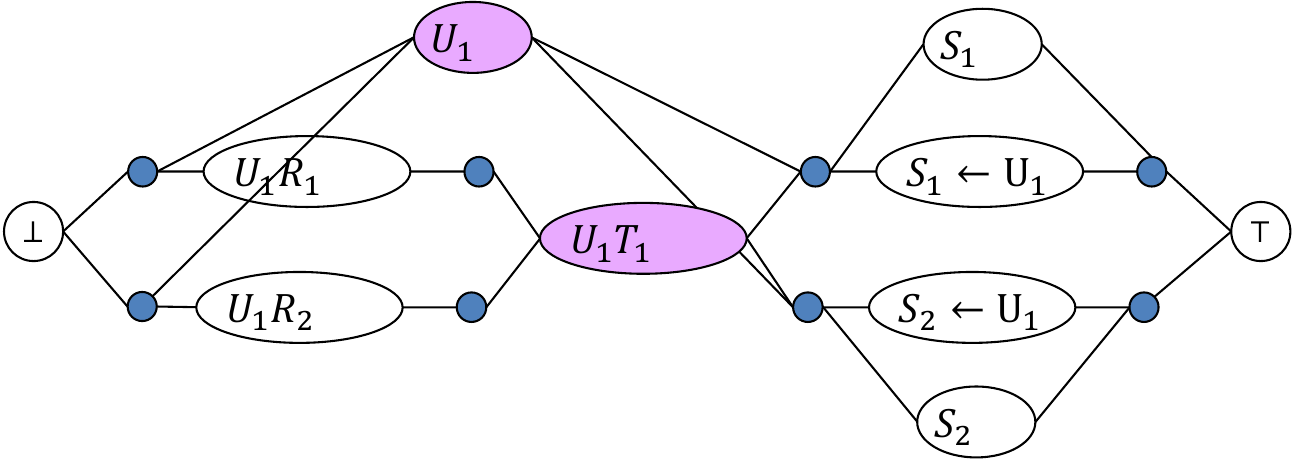}
        \caption{$\qtriangleunary$ instance, with shared $xz$ (Case 5)}
        \label{fig:TU-shared-xz}
    \end{figure}
	
    \item \emph{$y$ and $z$ values}:    
    This case (\cref{fig:TU-shared-yz}), is similar to case 4 -  since the shared node $S$ is at the side of the graph, we see that all paths through it pass through nodes corresponding with either only $\w_1$ or only $\w_2$, thus not allowing leakage.
    \begin{figure}
        \centering
        \includegraphics[scale=0.4]{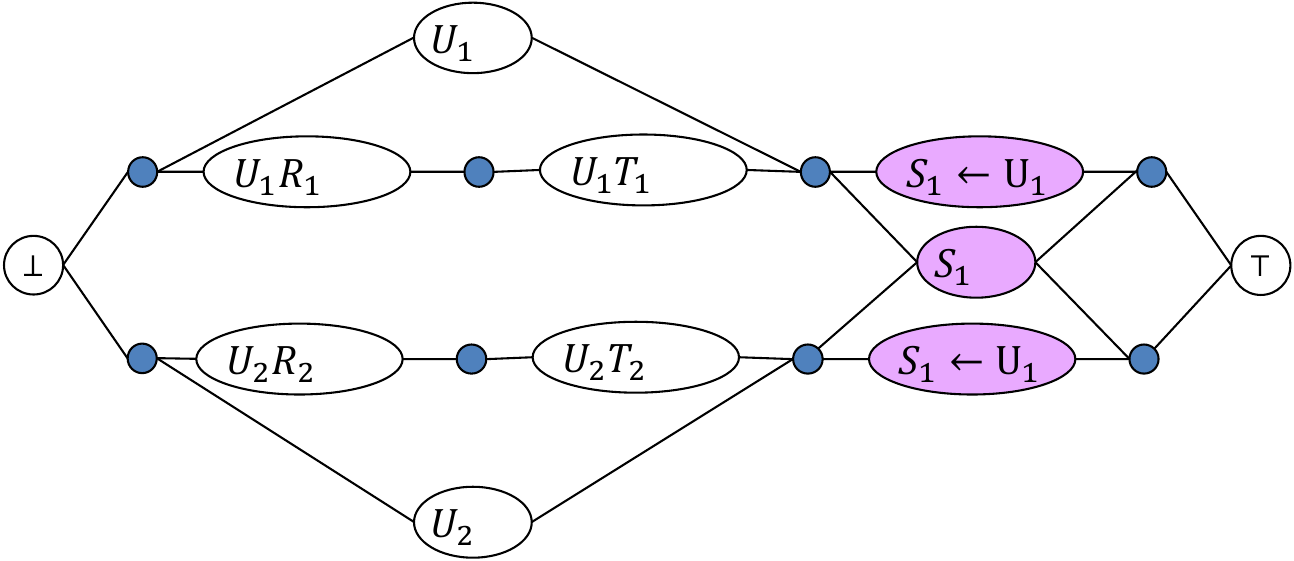}
        \caption{$\qtriangleunary$ instance, with shared $yz$ (Case 6)}
        \label{fig:TU-shared-yz}
    \end{figure}
    
    \item \emph{$x$, $y$ and $z$ values}:    
    This case is not possible under set semantics.
	\qedhere
\end{enumerate}
\label{proof:tu-flow-optimal}
\end{proof}

We show an RP and leakage-free ordering for $\qtriangleunary$, hence proving that the algorithm meets conditions for optimality described in \cref{sec:heuristic-when-optimal}.

\begin{proof}[\cref{prop:triangleUNARY:LP}]

As we have shown in \cref{prop:lp-easy-mincut} that the LP relaxation of $\minfact$ is correct when it is solved by the MFMC-based algorithm, this result follows given  \cref{prop:triangleUNARY:Mincut}.

\end{proof}

\subsubsection{4Chain query}
\label{sec:4Chain-appendix}

We next consider the Four Chain query 
$\qfourchain \datarule P(u,x), R(x,y),$ $S(y,z),$ $T(z,v)$
and prove that
$\minfact(\qfourchain,D)\in\PTIME$.
Just like in the previous subsections, we show the $\PTIME$ complexity by showing that the MFMC-based algorithm and LP relaxation are optimal.

The query has five minimal $\veo$s 
and in the first proof, we show that one RP-ordering is leakage-free and optimal.
We then use this factorization flow graph to show that the Linear Program is always correct for any $\qfourchain$ instance as well.

\thmfourchain*

\begin{proof}[Proof \cref{prop:fourchain:Mincut}]

$\qfourchain$ has $5$ minimal $\veo$s. Let us define an ordering $\Omega = [x \!\leftarrow\!(u,y \!\leftarrow\!{z} \!\leftarrow\!{v}),
x \!\leftarrow\!(u,z \!\leftarrow\!(y,v)),
y \!\leftarrow\!(x \!\leftarrow\! u,z \!\leftarrow\!  v),
z \!\leftarrow\!(x \!\leftarrow\!(u,y),v),
z \!\leftarrow\!(y \!\leftarrow\! x \!\leftarrow\! u,v)]$

Notice that there are $2$ plans each with $x$ or $z$ as a prefix, 1 plan with $y$ as a prefix, and no plan with $u$ or $v$ as the prefix.
Notice as well that in the ordering $\Omega$, the y plan is in the center and the $x \!\leftarrow\! z $ and $z \!\leftarrow\! x $ plans are at the sides.
\cref{fig:four-chain-single-witness} shows a single witness under this ordering.

\begin{figure}
    \centering
    \includegraphics[width=\columnwidth]{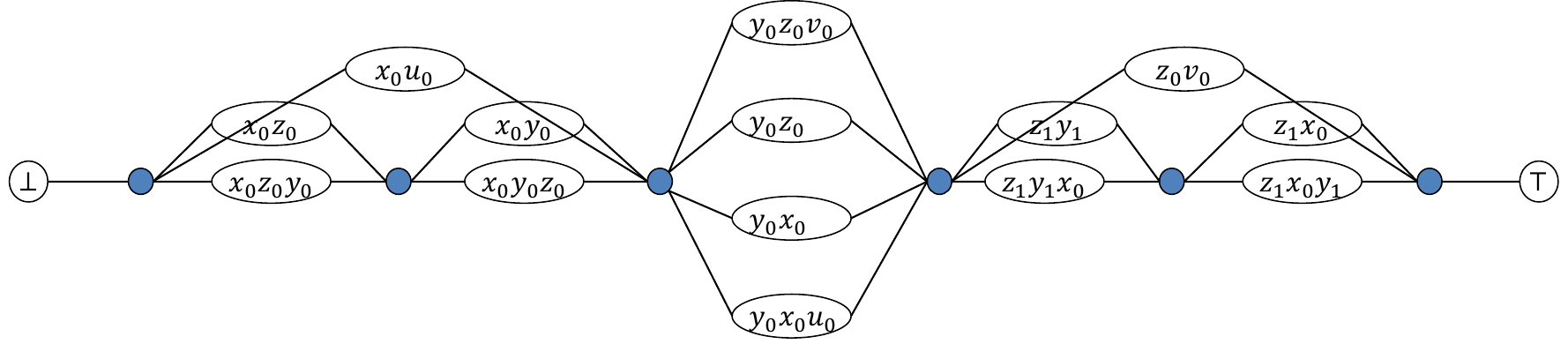}
    \caption{Flow graph with a single witness of $\qfourchain$ under ordering $\Omega$.}
    \label{fig:four-chain-single-witness}
\end{figure}

We first check that this ordering has the Running Prefixes Property. This is easily verified from the figure, no projection is parallel to a node it is not a prefix of. 

Next, we want to show that no flow graph under this ordering will have leakage. 
For this, we first prove a lemma that if there is leakage in a flow graph then there exists a leakage path comprised of nodes from just $2$ witnesses.

\begin{figure}
    \centering
    \includegraphics[scale=0.4]{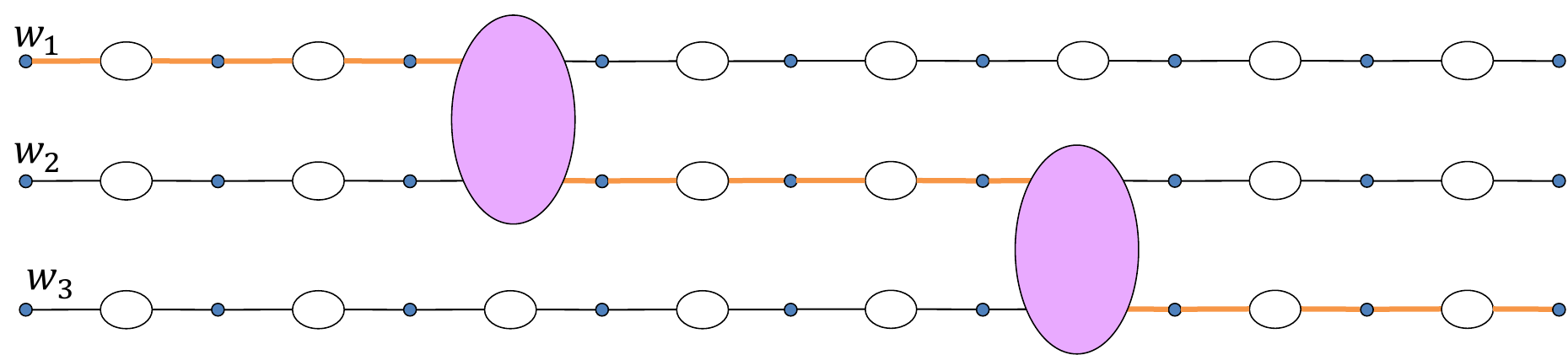}
    \caption{\cref{lem:min-leakage}: Showing that if there is leakage in the graph, there must be two witnesses that can together form a leakage path.
	}
    \label{fig:minimal-leakage-witnesses}
\end{figure}

\begin{lemma}
    \label{lem:min-leakage}
    If there exists leakage in a factorization flow graph F then there must exist two witnesses $w_1$, $w_2$ such that if we only consider a graph with cut nodes corresponding to these witnesses then the plans they use any minimal factorization do not cut the graph.
\end{lemma}

\begin{proof}[\cref{lem:min-leakage}]
Assume that the lemma is false and any leakage path in $F$ must traverse through $3$ witnesses. We know that for a leakage path to exist, the witnesses must have some shared nodes. 
Let us the minimal paths needed for such a set of witnesses $w_1$, $w_2$, $w_3$ in \cref{fig:minimal-leakage-witnesses}.
For a leakage path to exist, they must not choose any of the shared computations in the path. 
W.l.o.g., assume $w_1$ chooses a plan to the right of the shared computation. 
Notice that if $w_2$ chooses a plan in between the two shared computations then there isn't any more leakage.
So the only remaining case is if $w_2$ chooses a plan to the left of the shared computation with $w_3$, in which case there is a leakage path between $w_2$ and $w_3$.
\end{proof}

Now that we know that if there is leakage, there must be a leakage path between two witnesses, let us assume there is an instance of $\qfourchain$ under $\Omega$ where witnesses $w_1$, $w_2$ have a leakage path.

We do a case-wise analysis to show how this cannot exist, based on the variables that $w_1$ and $w_2$ share.
\begin{enumerate}
    \item They share no variables: No leakage is possible
    \item They share u,x,y,z, or v variables: In this case to the two witnesses share no nodes so no leakage is possible
    \item They share ux or zv or xz variables: In this case, the two witnesses share a variable - however, the variable is at the sides of the flow graph and hence cannot create a leakage path (there cannot be a crossover in the leakage path from $w_1$ to $w_2$
    \item They share uxzv: In this case, they share multiple nodes, but again, the shared nodes as on the sides.
    \item They share xy (or yz):  In this case, we can see there will be leakage if the $w_1$ chooses an $x$ plan and $w_2$ chooses a different $z$. 
    If the two witnesses share only $xy$ and no other variables, then $xy$ must be a part of at least $4$ witnesses because of join dependencies (every u that x is connected to must also connect to z and vice versa). In this case, you would always choose a y plan since y acts as the root of the Cartesian product.
    
    However, if they share other variables along with xy, the argument differs. The choosing of alternate plans could only happen if x and z were both connected to different y as well. In this case, the xy tuple would be repeated at least once, and since they use different prefixes and share tuples, all of the tuples $ux$, $yz$, $yv$ would be repeated. Meanwhile, a factorization where they both choose the same plan could only incur additional repeats on one of $ux$ or $zv$.
\end{enumerate}
	Notice that the above cases are exhaustive since they cover all possible prefix node combinations that can be shared between two witnesses. 
	That is, they cover the possibility of two witnesses sharing nodes
	ux, xz (equivalently zx, since the variables shared are the same), zv (in Case 3);  xy (equivalently yx), yz (equivalently zy), 
	xzy (equivalently xyz), yzv, yxu, zyx (equivalently zxy) (in Case 5).
	Notice that case 5 also includes when any other additional variables are shared alongside xy or yz. Thus, the above cases suffice to cover all possible prefix node sharing between two witnesses. Note that two witnesses can't share all 5 variables or all prefix nodes, as that would make them the same witness.
\end{proof}

\begin{proof}[\cref{prop:fourchain:LP}]

As we have shown in \cref{prop:lp-easy-mincut} that the LP relaxation of $\minfact$ is correct when it is solved by the MFMC-based algorithm, this result follows given \cref{prop:fourchain:Mincut}.

\end{proof}

\section{Proofs for \cref{sec:hardcase}: Hard Cases}

\subsection{Proof of \cref{prop:triadS}}
\thmactivetriadshard*

We notice that this proof 
is inspired by a result in \cite{DBLP:conf/pods/FreireGIM20}.
It can easily be adapted to also prove the hardness of triads for resilience 
and is considerably simpler than the original proof in \cite{FreireGIM15}.
The key idea is that in our proof construction, the chosen roots correspond to the minimum vertex cover,
or equivalently, the minimum number of input tuples to cover, in order to remove all witnesses.

\begin{proof}[Proof \cref{prop:triadS}]
Let $Q$ be a query with triad $\mathcal{T}=\set{R,S,T}$. We construct a hardness gadget \cref{fig:triad_hardness_gadget} that can be used to build a reduction from IndSet to any query $Q$.

The hardness gadget is built with $3$ witnesses. We assume that no variable is shared by all three elements of $\mathcal{T}$  
(we can ignore any
such variable by setting it to a constant). For the first $2$ witnesses, only the variables of atom $T$ are equal in both witnesses. Because $T$ is an independent atom, its tuple will be the only common tuple between $\w_1$ and $\w_2$. The nodes $x_1, x_2, x_3$ represent all the tuples of tables that are not part of the triad $\mathcal{T}$. Hence any query with a triad can create an instance of this gadget.

\begin{figure}
	\centering
	\tikzset{>=latex}
	\begin{tikzpicture}[thick]
	\draw [rounded corners] (-2,-0.5)--(2,-0.5)--(0,1.7)--cycle;
	\draw [rounded corners] (0,-0.5)--(4,-0.5)--(2,1.7)--cycle;
	\draw [rounded corners] (2,-0.5)--(6,-0.5)--(4,1.7)--cycle;
	\draw[orange] (-1,0) circle (0.25cm) node {$r_{a}$};
	\draw (0,0) circle (0.25cm) node {$s_1$};
	\draw (1,0) circle (0.25cm) node {$t_2$};
	\draw (2,0) circle (0.25cm) node {$r_3$};
	\draw (3,0) circle (0.25cm) node {$s_4$};
	\draw (4,0) circle (0.25cm) node {$t_5$};
	\draw (0,1) circle (0.25cm) node {$x_1$};
	\draw (2,1) circle (0.25cm) node {$x_2$};
	\draw (4,1) circle (0.25cm) node {$x_3$};
	\draw[orange] (5,0) circle (0.25cm) node {$r_{b}$};
	\draw (0,2) node {$\w_{1}$};
	\draw (2,2) node {$\w_{2}$};
	\draw (4,2) node {$\w_{3}$};
	\end{tikzpicture}
	\caption{\Cref{prop:triadS}: Hardness gadget for a query with triad $\{R,S,T\}$.}
	\label{fig:triad_hardness_gadget}
\end{figure}

We see that the gadget has $2$ minimal factorizations: 
$\textcolor{blue}{r_a} s_1 t_2 x_1 \vee s_4(r_3 t_2 x_2 \vee t_5 x_3 \textcolor{blue}{r_{b}})$ or as
$t_2(\textcolor{blue}{r_a} s_1 x_1 \vee r_3 s_4 x_2) \vee \textcolor{blue}{r_{b}} t_5 s_4 x_3$, 
both of which incur a penalty of $1$ and use one of $r_a$ or $r_b$ as root. 
We can represent these factorizations with gadget orientations, 
such that each gadget is oriented from the root of the factorization to the sink. 
Either of these factorizations must repeat some variable (thus have at least a penalty of 1), for $t_2$ or $s_4$ respectively. 
Since $r_a$ and $r_b$ may be connected to the other edges of the graph as well, they may incur further penalties if they are not root nodes in more than one factorization term.

To prove the problem is $\npc$, we reduce the Independent Set (IS) problem to  $\fact(Q)$. 
An independent set (IndSet) 
of an undirected Graph $G(V,E)$
is a subset of vertices $V' \subseteq V$, such that no two vertices in $V'$ are adjacent.
The $\IS$ problem asks, for a given graph $G$ and a positive integer $k \leq |V|$, 
whether $G$ contains an independent set $V'$ having $| V' | \geq k$,
which we write as $(G,k) \in \IS$.
$$
	\IS \leq \fact(Q_{3}^\star)
$$

\begin{proposition}
\label{prop:np-reduction-direction1}
$D$ has a factorization with a penalty $2|E|-k$ if $G$ has an independent set of size $k$.
\end{proposition}

\begin{proof}[Proof \cref{prop:np-reduction-direction1}]
If $G$ has an independent set of size $k$ then a factorization can be constructed corresponding to an edge orientation that points all edges towards nodes in the independent set. A minimum penalty of $|E|$ is guaranteed to be incurred as there is a penalty either on $t_3$ or $s_5$ for every gadget. 
What we need to reason about is the additional penalties that are incurred by the nodes at the end of the gadget.
For every node, it is counted in the factorization only once as root, and as many times as a non-root as it has incoming edges. 
If a node ever acts as a root, then its penalty is equal to the number of incoming edges, however, if it is never a root, then the penalty is the number of incoming edges $-1$. We construct an edge orientation such that any member of a maximum independent set never acts as the root, as shown in  \cref{fig:idp-set}. 
This construction is always possible since no members of the independent set are connected. 
The total factorization penalty then is the total incoming edges over all nodes - the size of the independent set. 
Since all $|E|$ edges are an incoming edge for some node, the penalty due to $k$ nodes is then $|E|-k$. 
The total penalty is hence $|E|+|E|-k$ or $|2|E|-k$.
\end{proof}

\begin{proposition}
\label{prop:np-reduction-direction2}
If $G$ has no independent set of $k$ then $D$ cannot have a factorization with penalty $2|E|- k$.
\end{proposition}
	
\begin{proof}[Proof \cref{prop:np-reduction-direction2}]
For $D$ to have a factorization with penalty $2|E|-k$ there would have to be $k$ gadgets with penalty exactly equal to $1$ thus implying the presence of $k$ sink nodes. Since $k$ sink nodes form an independent set of size $k$, $D$ cannot have a factorization with penalty $2|E|-k$.
\end{proof}

\Cref{prop:np-reduction-direction1,prop:np-reduction-direction2}
together finish the reduction.

\end{proof}

\begin{figure}
	\centering
	\begin{tikzpicture}[thick] 
	\node[circle, draw] (1) {$1$};
	\node[circle, draw] (3) [right of=1] {$3$};
	\node[circle, draw] (2) [below=7mm of 1, xshift=3mm] {$2$};
	\node[circle, draw] (5) [right of=3] {$5$};
	\node[circle, draw] (4) [below=7mm of 5, xshift=-3mm] {$4$};
	
	\draw [line width=0.4mm, >=latex] (1) -- (2);
	\draw [line width=0.4mm, >=latex] (2) -- (3);
	\draw [line width=0.4mm, >=latex] (2) -- (3);
	\draw [line width=0.4mm, >=latex] (3) -- (4);
	\draw [line width=0.4mm, >=latex] (2) -- (4);
	\draw [line width=0.4mm, >=latex] (4) -- (5);
	
	\end{tikzpicture} 
\hspace{10mm}
	\begin{tikzpicture}[thick, every node/.style={draw,circle},->, shorten >= 3pt, shorten <= 3pt]
	\node[circle, draw, blue] (1) {$1$};
	\node[circle, draw, blue] (3) [right of=1] {$3$};
	\node[circle, draw] (2) [below=7mm of 1, xshift=3mm] {$2$};
	\node[circle, draw, blue] (5) [right of=3] {$5$};
	\node[circle, draw] (4) [below=7mm of 5, , xshift=-3mm] {$4$};
	
	\draw [line width=0.4mm, >=latex] (2) -- (1);
	\draw [line width=0.4mm, >=latex] (2) -- (3);
	\draw [line width=0.4mm, >=latex] (2) -- (4);
	\draw [line width=0.4mm, >=latex] (4) -- (3);
	\draw [line width=0.4mm, >=latex] (4) -- (5);
	
	\end{tikzpicture}
\caption{\cref{prop:triadS}: Example graph with an independent set $\{1, 3, 5\}$. 
By orienting all edge directions towards the independent sets (and thus making them sink nodes), 
we achieve a factorization with penalty $2|E|-3 = 7$. 
The penalties at endpoints $3$ and $4$ are 1, elsewhere 0.
Additionally, each of the $5$ edges is encoded with the path gadget and has an internal penalty of $1$ each.}
\label{fig:idp-set}	
\end{figure}

\subsection{Proof of \cref{prop:dominatedtrianglehard}}
\label{sec:proofdominatedtrianglehard}

\thmsemideactivatedtriadhard*

A slight modification of the gadget from \cref{prop:triadS} can be used to show that this class of queries is hard for $\minfact$, although it is easy for $\res$.

\begin{proof}[Proof \cref{prop:dominatedtrianglehard}]
We construct a hardness gadget \cref{fig:dominated_triangle_hardness_gadget} that can be used to build a reduction from IndSet.
The hardness gadget is built with $3$ witnesses. 
It is essentially the same hardness as \cref{prop:triadS}, with the addition of a table that dominates all tables in the co-deactivated triad $A(w)$.
We keep the value of the variable $w$ fixed to $w_1$ in all gadgets.
Then $w_1$ must be the root in the $\minfact$, and the endpoints now denote the root of a factorization assuming $A(w_1)$ has been factored out.
We see that the endpoints $T(w,x,z)$ are independent of all other tables in the gadgets excluding $A(w)$.
Since only a single $A(w)$ tuple will be used through all the partitions, joining these gadgets will not lead to any unexpected witnesses.

We see that the gadget has $2$ minimal factorizations: 
both of which incur a penalty of $1$ and use one of $T(w_1, x_1, z_1)$ or $T(w_1, x_2, z_2)$ as root. 
We can use these factorizations as graph orientations as in \cref{prop:triadS} and the rest of the reduction is identical.
\end{proof}

\begin{figure}
	\includegraphics[scale=0.42]{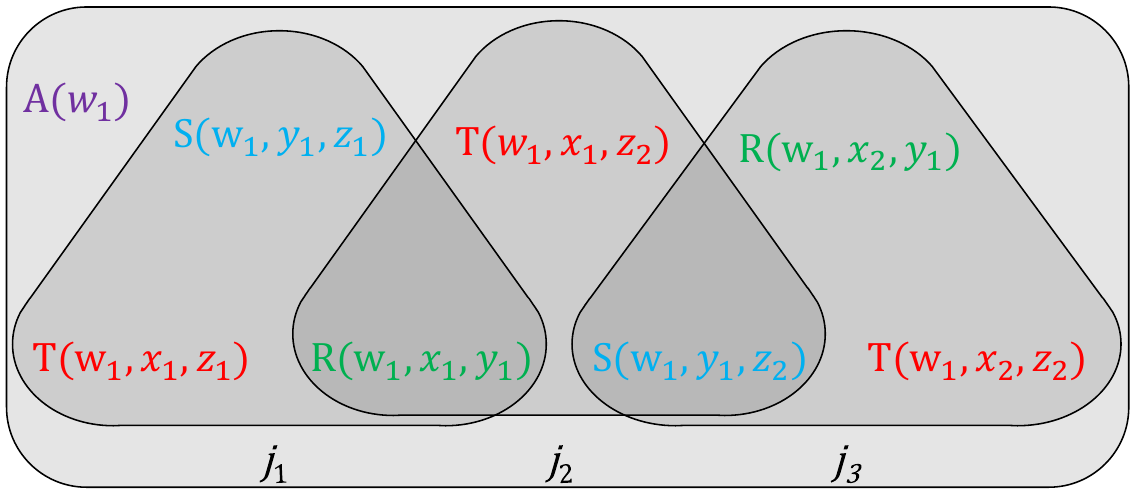}    
	\caption{\Cref{sec:proofdominatedtrianglehard}: Hardness gadget for $\qdominatedtriangle$}
	\label{fig:dominated_triangle_hardness_gadget}
\end{figure}

\section{Application: A complete approach for Approximate Probabilistic Inference}
\label{sec:application:pdbs}

\begin{definition}[Probabilistic query evaluation]
Given a Boolean query $Q$, a database $D$,
and a function $p$ that assigns an independent probability to each tuple.
Probabilistic query evaluation $\prob(Q, D, p)$
computes the marginal probability $\PP{Q,\D, p}$
(i.e.,\ the probability that $Q$ evaluates to {true} in a randomly chosen world).
\end{definition}

An important contribution of the database community to probabilistic inference 
has been the identification of tractable queries and database instances that allow exact evaluation in $\PTIME$.
Approaches towards identifying tractable cases are either at the query level (thus only look at the query and ignore the database \cite{DBLP:journals/vldb/DalviS07,DBLP:journals/jacm/DalviS12})
or at the data-level i.e. from the actual provenance polynomial \cite{DBLP:conf/icdt/RoyPT11, SenDeshpandeGetoor2010:ReadOnce}.
A practical concern with approaches that focus on tractable cases is that they 
provide only ``\emph{partial}'' solutions;
they offer none for the non-tractable cases
which then need separate approximate methods, 
usually based on Monte Carlo approximations \cite{DBLP:conf/sigmod/BoulosDMMRS05,DBLP:journals/vldb/DalviS07,re2007efficient}
or anytime approximation schemes based on {branch-and-bound} provenance decomposition methods \cite{DBLP:conf/sigmod/HeuvelIGGT19,DBLP:journals/vldb/FinkHO13,DBLP:conf/icdt/FinkO11,DBLP:conf/icde/OlteanuHK10}.
A different line of work is ``\emph{complete}'' in that it can answer all queries, 
but it is query-centric~\cite{gatterbauer2014oblivious, DBLP:journals/vldb/GatterbauerS17}
and 
thus \emph{unnecessarily approximates} cases that would allow an exact $\PTIME$ solution
if one looked at the concrete database instance.
Interestingly, a solution to the ``minimal provenance factorization" problem combines the best of both worlds (\cref{fig:connectionsWithPDBs}):
\emph{a complete approach that includes all known tractable cases for exact inference as special cases}.
We illustrate this in \cref{ex:1}.

\begin{figure}[]
\centering
\small
\renewcommand{\tabcolsep}{0.7mm}
\renewcommand{\arraystretch}{1.10}
\begin{tabular}[b]{@{\hspace{0pt}} p{13mm}  | p{40mm} | p{40mm} | @{\hspace{0pt}}}
	\multicolumn{1}{c}{}
	&
	\multicolumn{1}{c}
	{ \textbf{\normalsize Query-level} }	
		& 
		\multicolumn{1}{c}
		{ \textbf{\normalsize Data-level} }	
		\\
\cline{2-3}
	\textbf{Partial solution}
		&
		\cite{DBLP:journals/vldb/DalviS07}:
		dichotomy that has
		a $\PTIME$ solution iff the query is hierarchical
			&
			\cite{DBLP:conf/icdt/RoyPT11, SenDeshpandeGetoor2010:ReadOnce}:
			$\PTIME$ solution iff the provenance polynomial is read-once
			(which includes \cite{DBLP:journals/vldb/DalviS07} as special case)\\
\cline{2-3}
	\textbf{Complete solution}
		&
		\cite{DBLP:journals/vldb/GatterbauerS17}:
		$\PTIME$ approximation for any query
		(recovers the $\PTIME$ cases of \cite{DBLP:journals/vldb/DalviS07}, 
		but approximates some cases of \cite{DBLP:conf/icdt/RoyPT11})
			&
			\textbf{this paper}:
			finding the minimum factorization 
			includes all other three solutions \cite{DBLP:journals/vldb/DalviS07,DBLP:journals/vldb/GatterbauerS17,DBLP:conf/icdt/RoyPT11}
			as special cases
			\\
\cline{2-3}
\end{tabular}
\caption{Connection of the problem formulation and solutions in this paper with prior results on exact or approximate solutions for
probabilistic evaluation of sj-free conjunctive queries.}
\label{fig:connectionsWithPDBs}
\end{figure}

\begin{example}[Provenance]
	\label{ex:1}
	Consider again the Boolean 2-star query 
	$\qtwostar \datarule R(x),$ $S(x,y), T(y)$
	over the database in \cref{fig:RST_intro}
	(ignore tuple  $ \color{orange}s_{13}$ for now)
	as in \cref{ex:prov}.
	Each tuple is annotated with a Boolean variable 
	$ r_1, r_2,  \ldots$, 
	representing the independent event that the tuple is present in the database.
	The {provenance} $ \varphi$ is the Boolean expression that states which tuples need to be present for 
	$\qtwostar$ to be true:
	\begin{align}
		\varphi= {\color{blue}r_1} s_{11} t_1 \vee {\color{blue}r_1} s_{12} t_2 \vee r_2 s_{23} {\color{blue}t_3} \vee r_3 s_{33} {\color{blue}t_3}
		\label{eq:phiDNF}
	\end{align}
	Here, the color blue indicates tuple variables that are repeated in the expression.
	Evaluating the probability of $\qtwostar$ is \#P-hard 
	according to the dichotomy by Dalvi and Suciu for SJ-free CQs~\cite{DBLP:journals/vldb/DalviS07}. 
	This means that for database instances of increasing sizes, the evaluation becomes infeasible, in general.
	However, the dichotomy does not take the actual database into account, %
	and Roy et al.~\cite{DBLP:conf/icdt/RoyPT11} 
	and Sen et al.\ \cite{SenDeshpandeGetoor2010:ReadOnce} later independently 
	proposed a $\PTIME$ solution for particular database instances (including the one from this example)
	that allowed a read-once factorization.
	This is a factorized representation of the provenance polynomial in which 
	every variable occurs once, and which can be found in $\PTIME$ in the size of the database:
	\begin{align*}
		\varphi' &=r_1 (s_{11} t_1 \vee s_{12} t_2) \vee (r_2 s_{23} \vee r_3 s_{33}) t_3
	\end{align*}

	Another approach called query dissociation
	\cite{DBLP:journals/vldb/GatterbauerS17}
	can evaluate it only approximately with an upper bound
	that is calculated with a ``{probabilistic query plan}''
	\begin{align}
		& P  = \projp{x}\joinp{}{R(x),\projp{y}\joinp{}{S(x,y),T(y)}}	
		\label{eq:2starplan1}
	\end{align}
	where the probabilistic join operator $ \joinp{}{\ldots}$ in prefix notation
	and the probabilistic project-away operator with duplicate elimination $\projpd{}{}$ 
	compute the probability assuming that their input probabilities 
	are independent~\cite{DBLP:journals/tois/FuhrR97}. 
	This approach intuitively leads to a read-once upper-bound expression 
	in which formerly repeated variables are now replaced by fresh copies (indicated by the new indices):
	\begin{align}
		\varphi'' &=r_1 (s_{11} t_1 \vee s_{12} t_2) \vee r_2 s_{23} {\color{blue}t_{3,1}'} \vee r_3 s_{33} {\color{blue}t_{3,2}'}
		\label{eq:phifactorized}
	\end{align}
\end{example}

The exact approaches listed
above are important milestones 
as they delineate 
which cases can be calculated efficiently or not. 
However, an important problem is that even the data-level approaches 
\cite{DBLP:conf/icdt/RoyPT11, SenDeshpandeGetoor2010:ReadOnce} are only ``\emph{partial}'',
i.e.\ when we make a minor modification to the database, they do not work anymore,
as illustrated next.

\begin{example}[continued]
	\label{ex:incompleteApproaches}
Consider again adding one tuple $\color{red}s_{13}$ to the database as in \cref{ex:fact}: it  
is shown as a red dashed line in the join graph of \cref{fig:RST_intro}. 
After this addition, the exact approaches based on query level \cite{DBLP:journals/vldb/DalviS07} or 
on instance read-once formulas  \cite{DBLP:conf/icdt/RoyPT11} do not work anymore.

An approach based on query dissociation with minimal query plans 
\cite{DBLP:journals/vldb/GatterbauerS17}
would choose one of 
two query-level factorizations with either all tuples from $R$ or all tuples from $T$ as ``root variables.'' 
The solution with $R$'s as root is:
\begin{align*}
	\phi &=r_1 (s_{11} t_1 \vee s_{12} t_2 \vee {\color{red}s_{13}} {\color{blue}t_3}) 
	\vee r_2 (s_{23} {\color{blue}t_3}) \vee r_3 (s_{33} {\color{blue}t_3})
\end{align*}

This would lead to variable ${\color{blue}t_3}$ being repeated 3 times and expression size 13.
Similarly, approaches based on Shannon expansion
\cite{DBLP:conf/icde/OlteanuHK10}
need to start from such a factorization before repeatedly applying the expansion until arriving at a read-once expression.

However, there is an optimal factorization that repeats one variable only 2 times and has a total size of $12$.
Importantly, this factorization needs to use \emph{different} root variables for different witnesses:
\begin{align*}
	\phi' &=r_1 (s_{11} t_1 \vee s_{12} t_2 \vee {\color{red}s_{13}} {\color{blue}t_3}) 
	\vee (r_2 s_{23} \vee r_3 s_{33}) {\color{blue}t_3}
\end{align*}
\end{example}	

Our idea is to try to come as close as possible to a read-once factorization in the general case 
by using dissociation-based bounds with the fewest number of repeated variables necessary.
This achieves a complete solution (covering all queries and database instances).
The problem
is related to various 
topics in databases
and is best understood
in the modern formulation of provenance semirings~\cite{GKT07-semirings,Green:2017:SFD:3034786.3056125}. 
A minimal factorization can also be used as input to anytime or exact algorithms 
relying on repeated application of Shannon expansion~\cite{DBLP:journals/vldb/FinkHO13}.
Our problem generalizes the problem 
of determining whether a propositional formula is read-once 
to 
that of finding an expression of minimal size.

\begin{example}[\cref{ex:1} continued]
	Consider the following plan
	\begin{align}
		& P'  = \projp{x,y}\joinp{}{R(x),S(x,y),T(y)}	
		\label{eq:2starplan2}
	\end{align}
	It corresponds to a query dissociation 
	$\Delta' = (\{y\}, \emptyset, \{x\})$
	whereas the plan shown in \cref{eq:2starplan1} 
	corresponds to a dissociation
	$\Delta = (\emptyset,  \emptyset, \{x\})$.
	Thus, 
	$\Delta \preceq \Delta'$
	and we know from
	\cref{thm:minfactminveo} 
	that 
	$\len\big(\varphi(P^{\Delta})\big) \leq \len\big(\varphi(P^{\Delta'})\big)$
	over any database.
	Let's verify for the case of 
	the database in \cref{fig:RST_intro} without the red tuple.
	$\varphi(P^{\Delta})$
	corresponds to  
	\cref{eq:phifactorized}
	with length 11
	whereas 
	$\varphi(P^{\Delta'})$
	corresponds to 
	\cref{eq:phiDNF}
	with length 12.
\end{example}

\subsection{A Proof-of-Concept for Improving Probabilistic Inference}
\label{sec:piexpt}

\introparagraph{Setup}
With \cref{fig:expt-prob-inference}, we next generalize our setup from \cref{fig:RST_intro:b}:
We again consider a database of format where one of the $r$ tuples is connected to all $t$ tuples, and vice versa (\cref{fig:PI-database}).
But instead of 3 tuples, we now have $k$ tuples in tables $R$ and $T$, respectively (and $2k-1$ in $S$).
We consider the query $\qthreechain$, which has $2$ minimal query plans, and allow the $R$ and $T$ tables to take on probabilistic values (whereas tuples in $S$ are deterministic, i.e.\ have probabilities $1$).
We investigate if the minimal factorization indeed leads to better probabilistic bounds as $k$ increases.

\begin{figure}[t]
    \begin{tikzpicture}[node distance={3mm}, thick, main/.style = {draw, circle}] 
		\node (1) {$R$}; 
		\node[main] (2) [below = of 1] {$r_1$};
		\node[main] (3) [below = of 2] {$r_2$};
		\node[main] (4) [below = of 3] {$...$}; 
		\node[main] (5) [below = of 4] {$r_k$}; 
		\node (11) [right = 10mm of 1] {$S$};
		\node (6) [right = 10mm of 11] {$T$}; 
		\node[main] (7) [below = of 6] {$t_1$};
		\node[main] (8) [below = of 7] {$t_2$};
		\node[main] (9) [below = of 8] {$...$}; 
		\node[main] (10) [below = of 9] {$t_k$}; 
		\draw (2) -- (7); 
		\draw (2) -- (8); 
		\draw (2) -- (9); 
		\draw (2) -- (10); 
		\draw (3) -- (10); 
		\draw (4) -- (10); 
		\draw (5) -- (10); 
	\end{tikzpicture}
    \caption{\Cref{sec:piexpt}: Parameterized database instance for evaluating probabilistic inference with $2k$ probabilistic variables.
	}
    \label{fig:PI-database}
\end{figure}

\begin{figure}
	\centering
	\includegraphics[scale=0.45]{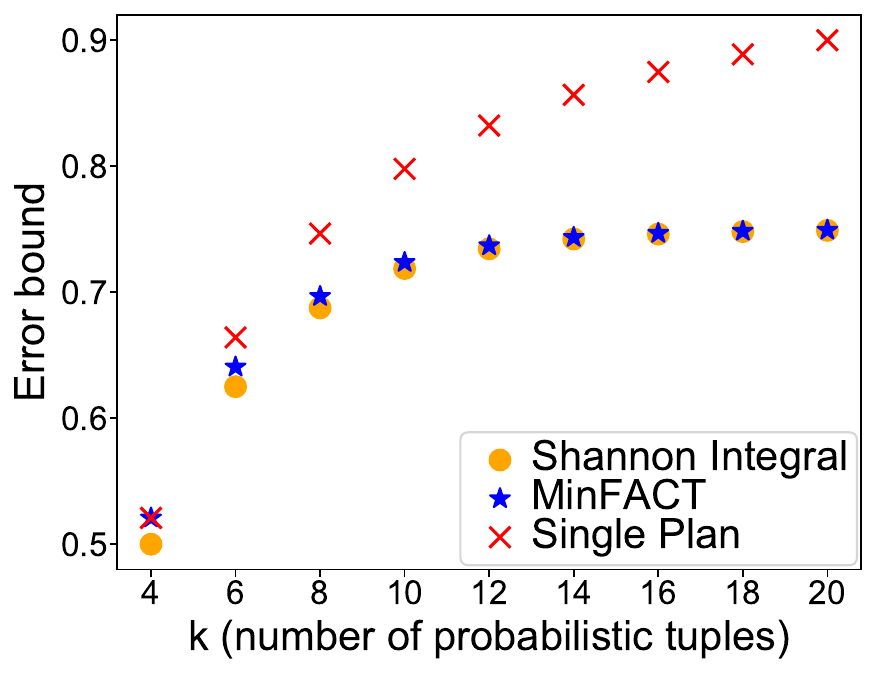}
	\caption{\Cref{sec:piexpt}: Expected value of exact and approximate probabilistic inference comparison for an increasing number of tuples (x-axis is $k$).}
	\label{fig:expt-prob-inference}
\end{figure}

\introparagraph{Methods}
Exact probabilistic inference is hard for this case,
and we use Shannon expansion to obtain the exact probabilities as a baseline.
As comparison, we use the minimal factorization of this expression (which is not read-once) and the expression obtained by evaluating the database under a single dissociated query plan.

\introparagraph{Evaluation}
For all three expressions, we calculate the expected probability assuming that each of the $2k$ tuples can take on a probability independently sampled from a uniform random probability in $[0, 1]$. To calculate that expected value, we integrate over all variables.
Thus, for an instance with just $k=2$ tuples in $R$ and $2$ $T$, respectively, we perform a quadruple integral.
We performed these integrations in Wolfram Mathematica \cite{Mathematica} and were only able to scale to $20$ repeated integrals.

\introparagraph{Results}
We plot the integral of the expected probability over the whole space on the y-axis,
and the x-axis shows the number $k$ of probabilistic variables. 
We see that with increasing variables $k$, $\minfact$ converges to the exact probabilistic inference given by Shannon expansion, while a single plan dissociation gives increasingly worse approximations.

\section{Optimizations for computing $\minfact$}
\label{sec:appendix:pruning}

We think of $\minfact$ as the problem of assigning $\veo$s to each witness to optimize $\texttt{len}$. 
Each witness chooses between $|\mveo|$ $\veo$s. 
However, it is not necessary that each must choose from the same set. 
For some witnesses, some $\veo$s can be directly eliminated (or pruned) and thus their set of choices for those witnesses is reduced. 
The pruning is done based on ``\emph{degree statistics}'' which summarize information about 
how many witnesses share a given subset of the domain values of a given witness.
This optimization can be thought of as a preprocessing step for both the ILP and the MFMC-based algorithm.

\begin{example}[VEO Pruning]
\label{example:pruning-intuition}
Consider query $\qthreestar$ query and a database instance that contains $2$ witnesses:
$\vec w_1 = (r_1, s_1, t_1, w_{111})$ and 
$\vec w_2 = (r_1, s_2, t_2, w_{122})$ 
(or $\vec w_1 = (x_1,y_1,z_1)$ and 
$\vec w_2 = (x_1,y_2,z_2)$ in the domain values perspective).
The only tuple they have in common is $r_1$ and the only variable they have in common is $x_1$. 
We say that the degree (or count) of $x_1$ in $\vec w_1$ is 2.
We can intuitively see that a minimal factorization would have $r_1$ as a ``root'' i.e. would use a $\veo$ in which the $x_1$ is a prefix. 
\end{example}

\introparagraph{Degree Statistics}
We leverage degree statistics of tuples participating in a witness to filter its potential query plan assignments. 
Given a witness $\vec w$ and strict subset of attributes $\{x_1, x_2, \hdots x_k\}$, 
we define $\Count(x_1, x_2, \hdots x_k)$ 
as the number of witnesses that have the same valuation of these attributes. 
Each degree statistic can be calculated by a simple group-by and aggregation. 
We calculate the degree statistics for all sets in the powerset of $\var(Q)$.
The number of such statistics is exponential in the size of the query, but polynomial in the data size, 
and in practice the time to compute all such statistics is small (see experimental figures).

\begin{example}[\cref{example:pruning-intuition} continued]
\label{degree-statistics}
We compute the degree statistics over the example: 
$\vec w_1$ has $c_x = 2, c_y = 1, c_z = 1, c_{xy}= 1, c_{yz} = 1 $ and $c_{xz} = 1$. 
$\vec w_2$ has the exact same statistics.
\end{example}

\introparagraph{Pruning Rule} 
Given two $\veo$s $v_1$ and $v_2$ with the roots of the trees being composed of the set of variables 
$\vec r_1$ and $\vec r_2$, 
we can apply the following rules:
\begin{align*}
\text{If } \Count(r_1) &== \Count(r_1 + r_2) \text{ and } \Count(r_1) \text{ \textless } \Count(r_2):\\
& \text{Eliminate Plan }v_1 \\
\text{Elif } \Count(r_2) &== \Count(r_1 + r_2) \text{ and } \Count(r_2) \text{ \textless } \Count(r_1):\\
& \text{Eliminate Plan }v_2 \\
\text{Elif } \Count(r_1) &== \Count(r_2) == \Count(r_1 + r_2) \text{ and } \\
\Count(r_2) &=  \Count(r_1):\\
& \text{Assign Equivalence Class } \{v_1, v_2\} \\
\end{align*}

\begin{example}[Pruning Rule]
    Continuing our setup, we see that for witness $x_1y_1z_1$, $\Count(y_1) = \Count(x_1y_1)$ and that $\Count(y_1)$ < $\Count(x_1)$. Thus, we can eliminate plans with the root $y_1$, i.e. the two plans $y \!\leftarrow\! x \!\leftarrow\! z $ and  $y \!\leftarrow\! z \!\leftarrow\! x$.
    
    We can apply the pruning rule again in the same manner to eliminate two plans with $z_1$ as the root. The same pruning will also apply to $w_2$.
\end{example}

\introparagraph{Recursive Application of Pruning Rule}
Consider 2 $\veo$s $v_1$ and $v_2$ with the same root $r$, but then we compare all the children subtrees of $r$ and apply the pruning rule recursively to further prune between plans with same root.

\begin{example}[Recursive Pruning]
    After the pruning in the last step, we are left with two plans $x \!\leftarrow\! y \!\leftarrow\! z$ and $x \!\leftarrow\! z \!\leftarrow\! y$.
    We now prune at tree level $2$.
    We see that $\Count(xy) = \Count(xz) = \Count(xyz)$. 
    Thus, we assign an equivalence class $\{xy, xz\}$
\end{example}

\introparagraph{Equivalence Classes}
At the end of all the pruning, we look at our collected equivalence classes. If two $\veo$s are in an equivalence class, this implies that all else being equal, both plans will lead to equal $\texttt{len}$. Thus, if we have multiple plans from the same equivalence class, we remove all but one arbitrarily.  

\begin{example}[Equivalence Classes]
    Since $\{xy, xz\}$ are in an equivalence class, and we have finished our pruning, we arbitrarily remove plans with the prefix $xz$.
    Thus, we are left with sole plan $x \!\leftarrow\! y \!\leftarrow\! z$.
    We were able to reduce both witnesses from $6$ to $1$ plan in this very easy example.
    (Note that in other cases, different witnesses may have different degree statistics and hence reach a different set of pruned plans).
\end{example}

\introparagraph{Proof of Correctness of Optimization}
\begin{enumerate}
    \item Correctness of the Pruning Rule: Consider the case where with $\veo$s $v_1$ and $v_2$ where $\Count(r_1) = \Count(r_1 + r_2)$ and $\Count(r_1) < \Count(r_2)$. 
    Let us assume that a witness $w_1$ necessarily needs $v_1$ for a minimal factorization of the instance.
    We take all witnesses that share $r_1$ and $r_2$ and if they use $v_1$ as their VEO, we switch them to $v_2$. 
    Notice that $\Count(r_1) = \Count(r_1 + r_2)$ implies that whichever witnesses share $r_1$, also share $r_2$. 
    So switching a subset of these witnesses over from $v_2$ to $v_1$ can repeat no additional variables since it is never worse to factor out $v_2$ first (which they all share anyway).
    For witnesses that contain $r_1$ and $r_2$ but did not use $v_1$ plan, they could possibly incur penalties on a subset of $r_1$ that they used as prefixes. However, as $\Count(r_1) < \Count(r_2)$, there is a greater prevalence of $r_2$ variables that can be repeated and so they are never negatively affected by the switch.
    Thus switching to $v_2$ can only lead to better or equal factorization hence the pruning rule is always correct.
    \item Correctness of Equivalence Classes: We can use the same logic as that of the pruning rule to show that when $\Count(r_1) = \Count(r_2) = \Count(r_1 + r_2) \text{ and} \Count(r_2) =  \Count(r_1)$, then neither $v_1$ nor $v_2$ can be worse than each other.
    \item Correctness of Repeated Application: If each pruning rule application is correct, then for every plan that is pruned, there exists an unpruned plan that leads to a better or equal $\texttt{len}$. Since the rules are applied linearly and not in parallel, we can safely repeatedly apply the rule until applicable.
\end{enumerate}

\section{Experiments}
\label{sec:experiments}
\label{SEC:EXPERIMENTS}

We implemented our ILP, LP relaxation, and the MFMC-based algorithm as LP\footnote{The MFMC-based algorithm can be encoded as an LP since the min-cut problem itself can be encoded as an LP}
by using Python 3.8.5 for the pre-processing of the provenance
and creating the problem encoding,
and then Gurobi Optimizer $8.1.0$~\cite{gurobi} for solving the respective optimization problems.
Our goals are to evaluate 1) the running times, and 2) the size of resulting factorizations.
In this section we illustrate our results and show three interesting takeaways:
1) ILP is, as expected, exponential in the size of the program, in general. 
The MFMC-based algorithm is not and can thus speed up the evaluation quite drastically. 
While it comes with no guarantees for hard cases, it approximates the minimum size quite well.
2) For queries that we have shown are in $\PTIME$, the ILP solver is comparably fast as the LP and all algorithms provide the correct solution.
3) The optimizations described in \cref{sec:appendix:pruning} further speed up the evaluations in a way that Gurobi cannot.

For both the ILP and MFMC-based algorithms, the \emph{optimizations} are applied during a preprocessing step and simplify the resulting ILP and MFMC formulations (\cref{sec:appendix:pruning}).
The intuition is that based on how each witness interacts with other witnesses (i.e.\ based on some \emph{degree statistics} obtained by simple group-bys), we can mark certain minimal query plans as unnecessary for a minimal factorization.\footnote{
We also believe that a similar pruning algorithm will be instrumental in achieving a $\PTIME$ factorization flow graph algorithm for other linear queries by eliminating nodes that cause leakage.}
Being able to find such optimization is notable since it implies we are \emph{identifying and leveraging structural properties of the problem} 
based on the minimal query plans that are difficult to extract in the same time by the state-of-the-art ILP solver Gurobi.

\subsection{Experimental Setup}
\label{sec:appendix:expt-details}

\introparagraph{Software and Hardware} We implement the algorithms using Python 3.8.5.\,
and solve the respective optimization problems with Gurobi Optimizer $8.1.0$~\cite{gurobi}, 
a commercial and highly optimized solver. 
The experiments are run on
an Intel\,\textsuperscript{\textregistered} Core\texttrademark\,i7-1065G7 CPU @ 1.30GHz machine with 132 GB RAM. 

\introparagraph{Experimental Protocol}
We focus on 5 queries: the two hard queries $\qthreestar$ and
$\qtriangle$, the two easy queries $\qtriangleunary$, $\qfourchain$, 
and finally query $\qfivechain$ which we hypothesize to be easy.
We first create a random database instance by fixing a number of tuples to sample, and
sampling attribute values independently for each attribute (and removing duplicates) and parameterized by total tuples in the instance.
We then run a provenance query and store the resulting provenance (and its size). 
On this database instance, we then run our five algorithms: ILP, ILP (OPT), MIN-CUT, MIN-CUT (OPT), and LP.
Here (OPT) refers to executing an algorithm after applying preprocessing optimizations described in \cref{sec:appendix:pruning}.

\begin{figure*}
    \begin{subfigure}[b]{0.95\textwidth}
        \centering
        \includegraphics[scale=0.3]{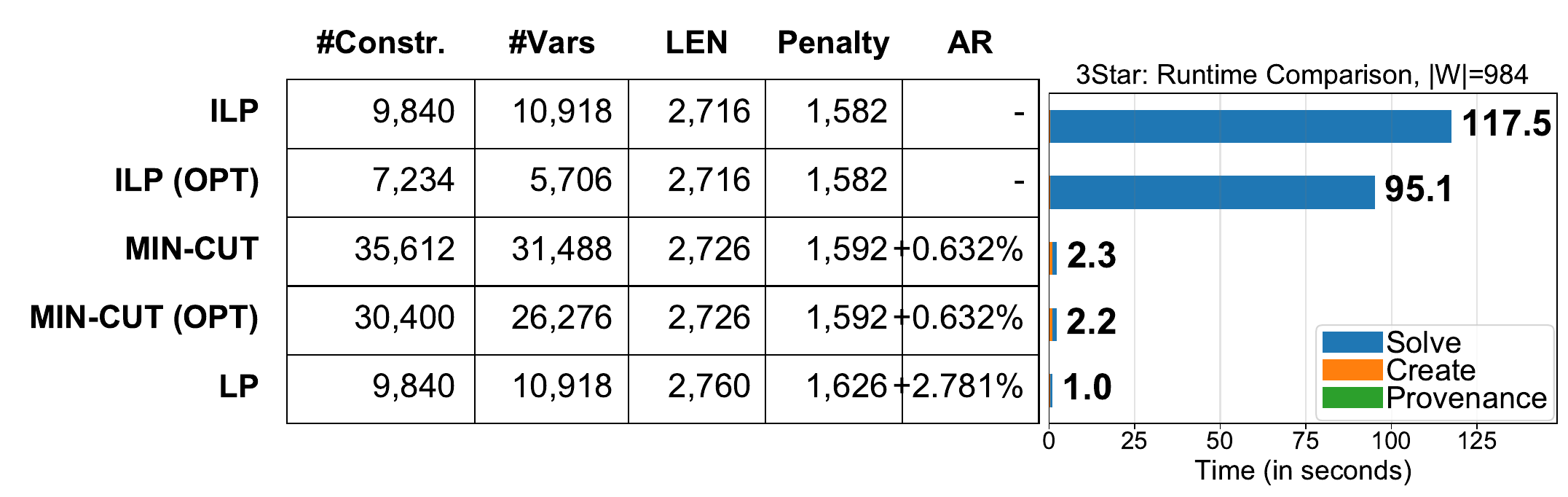}
        \vspace{-2mm}
		\caption{$\qthreestar$ (\npc query) instance with $984$ witnesses}	
        \label{fig:expt-threestar}
    \end{subfigure}
	\begin{subfigure}[b]{0.95\textwidth}
        \centering
        \includegraphics[scale=0.3]{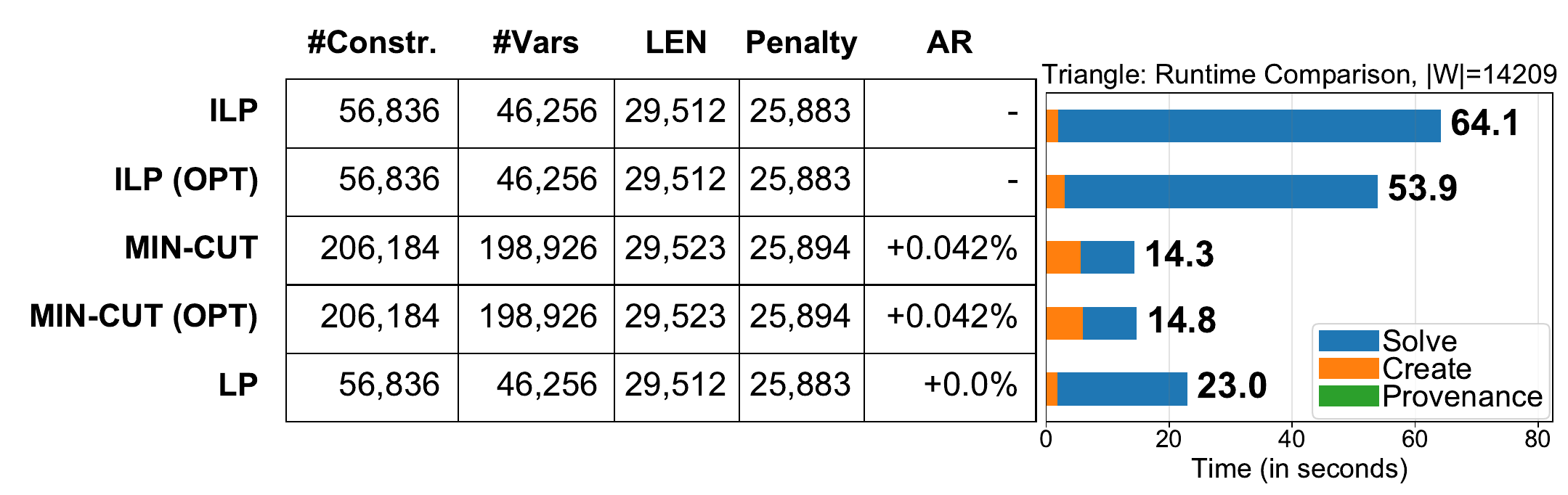}
        \vspace{-2mm}
		\caption{$\qtriangle$ ($\npc$ query) instance with $14,209$ witnesses}	
        \label{fig:expt-triangle}
    \end{subfigure}
    \begin{subfigure}[b]{0.95\textwidth}
        \centering
        \includegraphics[scale=0.3]{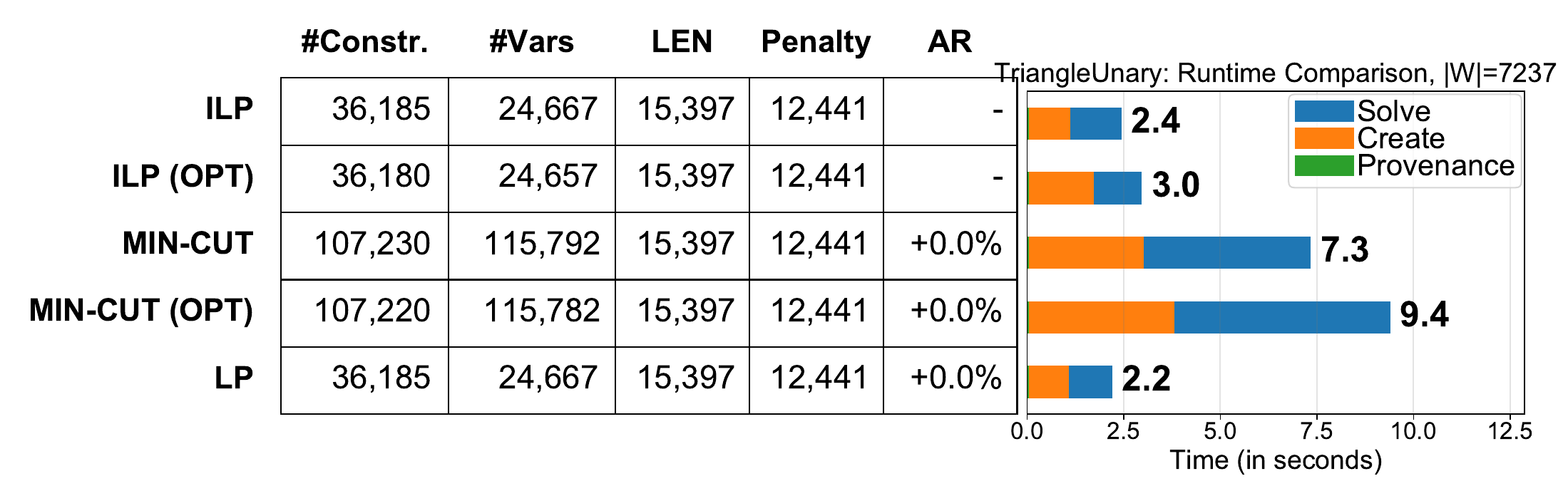}
        \vspace{-2mm}
		\caption{$\qtriangleunary$ ($\PTIME$ query)  instance with $7,237$ witnesses 
		}	
        \label{fig:expt-triangle-unary}
    \end{subfigure}
    \begin{subfigure}[b]{0.95\textwidth}
        \centering
        \includegraphics[scale=0.3]{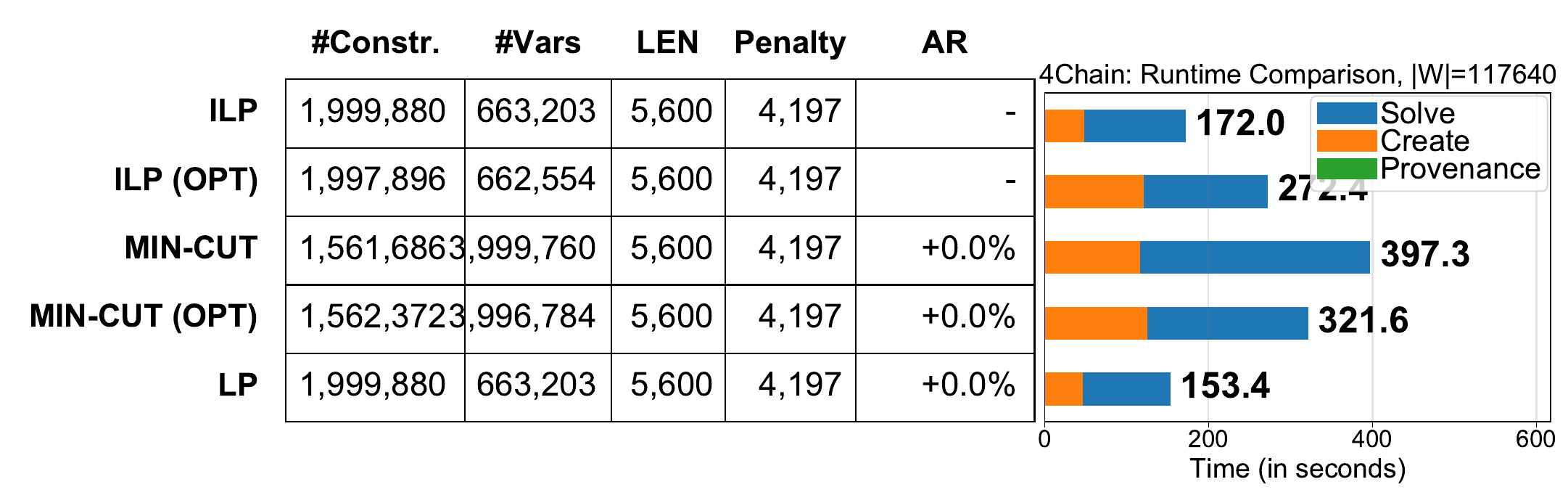}
        \vspace{-2mm}
		\caption{$\qfourchain$ ($\PTIME$ query) instance with $117,640$ witnesses}	
        \label{fig:expt-four-chain}
    \end{subfigure}
    \begin{subfigure}[b]{0.95\textwidth}
        \centering
        \includegraphics[scale=0.3]{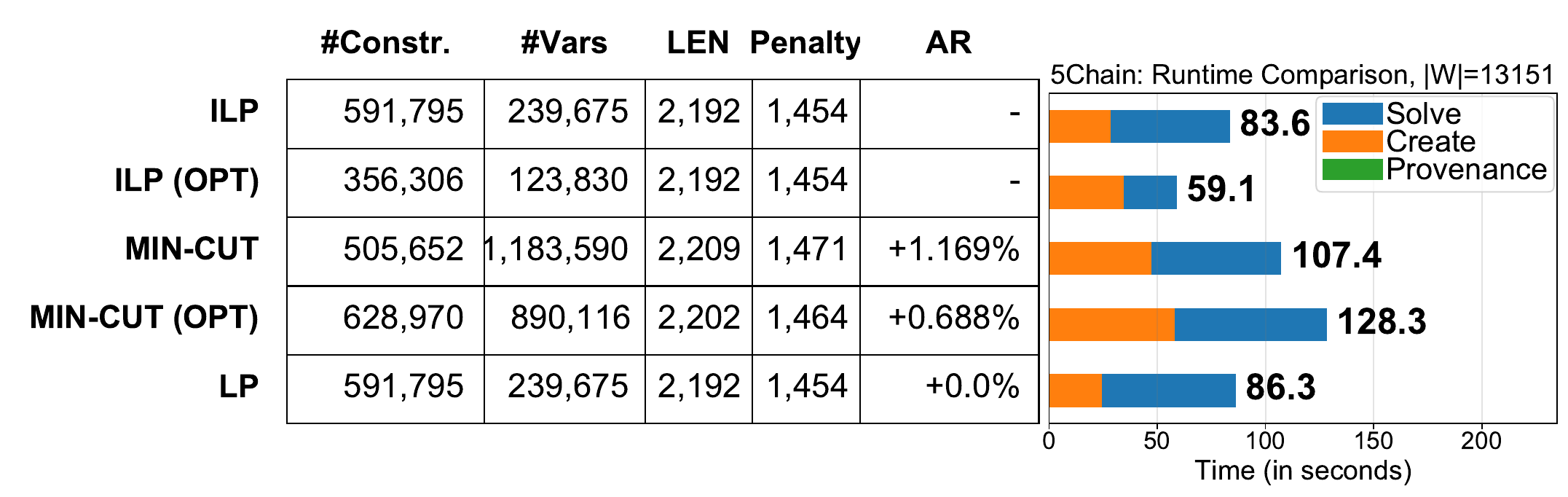}
        \vspace{-2mm}
		\caption{$\qfivechain$ ($\PTIME$ query) instance with $13,151$ witnesses}	
        \label{fig:expt-five-chain}
    \end{subfigure}
    \caption{Experimental study of time and factorization length for $5$ queries over all algorithms.}
    \label{fig:algo-compare-expts}
\end{figure*}

\subsection{Experimental results}
\label{sec:exptfigures}

In this section, we empirically evaluate $5$ queries discussed throughout the paper and not only gather some empirical support about the complexities we show in the paper but also gain some insight into the effectiveness of the optimizations and if advanced ILP Solvers can detect our $\PTIME$ Cases.

The ``Penalty'' of a factorization is the length ($\len$) minus the number of different tuple variables. 
The Approximation Ratio (AR) measures the relative increase of the penalty of an algorithm to the minimal penalty. 
The ILP and ILP (OPT) algorithms are guaranteed to find the optimal solution, i.e.\ $+0\%$. 
While the MFMC-based algorithm provides no guarantees for $\npc$ queries, we show that it provides a good approximation.

We can see that the MFMC-based algorithm contains more constraints and variables than the ILP problem (we must add a flow conservation constraint for each node in the graph, which includes nodes for corresponding ILP variables, plus connector nodes).
However, since it is solved as an LP optimization problem instead of an ILP optimization, it is much faster. 
We also notice that the optimized algorithm decreases the number of constraints in the MFMC-based algorithm by $15\%$. 
The number of constraints does not decrease in the ILP, but because the Query Plan Constraints involve fewer variables, the constraints are less complex and the ILP can be solved faster.

	\introparagraph{$\qthreestar$ \Cref{fig:expt-threestar}} The $3$-Star query is a hard query with an active triad and $|\mveo| =6$.
	For this query, the provenance computation time and problem creation time are small compared to the time needed to solve the ILP. 
	The \emph{optimized prepossessing step} reduces the time for the ILP solution by a tenth, while
	the MFMC-based algorithm is $\approx 60$ times faster than the ILP. 
	This is expected as the min-cut algorithm is $\PTIME$ while the ILP is not.
	We see here that the Linear Program is faster than the Flow algorithm, however, it has a worse approximation ratio (+2.781\% instead of +0.632\% - but still within the proved bound of $6$ for this query).

	\introparagraph{$\qtriangle$ (\cref{fig:expt-triangle})} The triangle query is a hard query with $|\mveo| =3$. 
	We see in the figure that since there are just $3$ minimal $\veo$s, the effect of the optimization is not very much. 
	In fact, for the MFMC-based algorithm, the longer time to create the optimized version is not paid off in the solve time. 
	However, as expected for a hard query, we see that the $\PTIME$ MFMC-based algorithm is must faster than the exact ILP. 
	The MFMC-based algorithm enables us to get a very close approximation in less than one-fourth of the time.  
	In this case, the Linear Program, also a $\PTIME$ algorithm, is slower than the MFMC-based algorithm but gives an exact solution!

    \introparagraph{ $\qtriangleunary$  (\cref{fig:expt-triangle-unary})} 
	Next, we look at the very structurally similar but easy query, $\qtriangleunary$ with $|\mveo|=3$. 
	Here we see surprisingly that the ILP is faster than the MFMC-based algorithm! 
	Notice that for this query we have shown that both algorithms are exact. 
	This is because an optimized ILP Solver like Gurobi can leverage the fact that LP = ILP and solve the ILP in $\PTIME$ (even though the ILP constraints do not follow traditionally $\PTIME$ structures like Total Unimodularity, it leverages some structure in the matrix that makes it $\PTIME$ solvable). 
	Since there are indeed more variables and constraints in the MFMC-based algorithm, the ILP turns out to be faster.

    \introparagraph{ $\qfourchain$  (\cref{fig:expt-four-chain})} 
	The $4$ chain query is $\PTIME$ like $\qtriangleunary$ but with $|\mveo| = 5$ and has a similar graph to $\qtriangleunary$.
	Here again, the guarantee of the MFMC-based algorithm being exact is fulfilled, the pruning offers some benefit, and the ILP is faster than the MFMC-based algorithm due to the smaller problem size and some innate structure discovered by the solver. 
    The Linear Program, which has always returned the optimal solution, is significantly faster. 

    \introparagraph{$\qfivechain$ (\cref{fig:expt-five-chain})} The $5$ chain query has $|\mveo| = 14$, and we believe it to be easy, but do not yet have proof for the same. 
    In this case, too, we believe that the Gurobi solver has some tricks that make the ILP run faster than exponential. 
    (Notice that a bigger instance with more than double the size of $\mveo$ is solved faster than $\qthreestar$ in \cref{fig:expt-threestar}).
    In addition, the Linear Program is optimal. 
	This backs up our hypothesis that the query is in $\PTIME$. 
    While the MFMC-based algorithm is not optimal in this case, we see that applying our pruning rules helps eliminate leakage paths and gives us a better approximation as well. 
	We hypothesize that there exists a set of pruning rules that make the MFMC-based algorithm optimal for all Linear Queries (which we hypothesize are all in $\PTIME$).

}{
}

\end{document}